\documentclass[12pt]{article}
\usepackage[a4paper, total={6.5in, 9in}]{geometry}
\usepackage{graphicx} 
\usepackage{braket}
\usepackage{amsmath,amssymb,amstext}
\usepackage{amsthm}
\usepackage{dsfont}

\newtheorem{lem}{Lemma}

\newtheorem{thm}{Theorem}
\newtheorem{cor}{Corollary}

\title{Diagonalizing Bose Gases in the Gross-Pitaevskii Regime and Beyond}
\author{Morris Brooks}

\date{October 2023}

\begin{document}

\maketitle

\begin{abstract}
We present a novel approach to the Bogoliubov theory of dilute Bose gases, allowing for an elementary derivation of the celebrated Lee-Huang-Yang formula in the Gross-Pitaevskii regime. Furthermore, we identify the low lying excitation spectrum beyond the Gross-Pitaevskii scaling, extending a recent result \cite{BCS} to significantly more singular scaling regimes. Finally, we provide an upper bound on the ground state energy in the Gross-Pitaevskii regime that captures the correct expected order of magnitude beyond the Lee-Huang-Yang formula.
\end{abstract}

\section{Introduction}
\label{Sec:Introduction}
Given a non-negative radial function with compact support $V\in L^1 \! \left( \mathbb{R}^3\right)$ let us introduce the rescaled two-particle interaction $V_L(x,y):=L^2 V\!\left(L(x-y)\right)$, where $L>0$. In the following we will study the many-particle Hamilton operator $H_{N,\kappa}$ acting on the $N$-particle Hilbert space $L^2_{\mathrm{sym}} \! \left(\Lambda^N\right)$, with $\Lambda:=\left[-\frac{1}{2},\frac{1}{2}\right]^3$ being the three-dimensional torus, given by
\begin{align}
\nonumber
H_{N,\kappa}: & =\sum_{i=1}^N (-\Delta)_{x_i}+\sum_{1\leq i<j\leq N}V_{N^{1-\kappa}}(x_i,x_j)\\
\label{Eq:Hamiltonian}   
&=\sum_{k\in 2\pi \mathbb{Z}^3} |k|^2 a_k^\dagger a_k \! + \! \frac{1}{2}\sum_{jk,mn\in 2\pi \mathbb{Z}^3}\left(V_{N^{1-\kappa}}\right)_{jk,mn}a_k ^\dagger a_j^\dagger a_m a_n,
\end{align}
where $a_k$ and $a_k^\dagger$ are the standard annihilation and creation operators on the bosonic Fock space $\mathcal{F}:=\bigoplus_{n=0}^\infty L^2_{\mathrm{sym}} \! \left(\Lambda^n\right)$ corresponding to the modes $e^{ikx}$ for $k\in 2\pi \mathbb{Z}^3$, $x_i-x_j$ refers to the distance between two particles on the torus $\Lambda$ and $0\leq \kappa < \frac{2}{3}$ is an additional scaling parameter. If not indicated otherwise, we will always assume that indices run in the set $2\pi \mathbb{Z}^3$, which we will usually neglect in our notation, and we write $k\neq 0$ in case the index runs in the set $2\pi \mathbb{Z}^3\setminus \{0\}$. 

The study of the low energy properties, such as the ground state energy $E_{N,\kappa}$, of dilute Bose gases described by the Hamilton operator $H_{N,\kappa}$ has a long standing history in the mathematics, as well the physics, literature. A rigorous derivation of the leading order 
\begin{align}
\label{Eq:Leading_Order}
    E_{N,\kappa}=4\pi\mathfrak{a} N^{1+\kappa}+o_{N\rightarrow \infty} \! \left(N^{1+\kappa}\right),
\end{align}
has been achieved in \cite{LY} for the thermodynamic limit in the regime of small densities $\rho$, where $\kappa$ is determined by the density $\rho$ via $\rho=N^{3\kappa-2}$, and later for the Gross-Pitaevskii regime $\kappa:=0$ in \cite{LSY}, with the scattering length $\mathfrak{a}$ defined as $8\pi \mathfrak{a}:=\int_{\mathbb{R}^3} V(x)\varphi(x)\mathrm{d}x$, where $\varphi$ is the radial solution of $(-2\Delta+V)\varphi=0$ subject to the boundary condition $\varphi\underset{|x|\rightarrow \infty}{\longrightarrow} 1$. More recently the subleading contribution to the asymptotics in Eq.~(\ref{Eq:Leading_Order}), famously conjectured to be the Lee-Huang-Yang term emerging from the underlying Bogoliubov theory \cite{B,LHY}, has been identified in the Gross-Pitaevskii regime in \cite{BBCS} and in the thermodynamic limit in \cite{FS1,FS2}. Following these landmark results, a lot of effort has been made to streamline and generalize the methods and results, see for example \cite{BCS,HST,FGJMO,BS,HHNST}.

In this work we want to combine the approach in \cite{FS1,FS2,FGJMO} with the one in \cite{BBCS,HST}, with the goal of finding an algebraic method, such as the one in \cite{FS1,FS2,FGJMO}, with a clean representation of important cancellations comparable to the one in \cite{HST}. Furthermore we will follow an alternative approach when it comes to the scattering length, which will be introduced via the Feshbach-Schur map discussed in Section \ref{Sec:The two-body Problem}. The following three hand-selected Theorem \ref{Th:GP}, Theorem \ref{Th:Beyond_GP} and Theorem \ref{Th:Upper_Bound} should be seen as an advertisement for the general method and its robustness, highlighting its various advantages. Starting with Theorem \ref{Th:GP}, we recover the, at this point, well known results in \cite{BBCS,HST}, following, what we would call, a short and elementary proof. All the relevant cancellations within the proof are exclusively discussed in Section \ref{Sec:The two-body Problem}.

\begin{thm}
\label{Th:GP}
    Let $E_N$ be the ground state energy of the operator $H_{N}:=H_{N,0}$. Then we have
    \begin{align}
    \label{Eq:GP_GSE}
        E_N= 4\pi\mathfrak{a}_N (N-1)+\frac{1}{2}\sum_{k\neq 0}\left\{\sqrt{|k|^4+16\pi \mathfrak{a} |k|^2}-|k|^2-8\pi \mathfrak{a}+\frac{(8\pi \mathfrak{a})^2}{2|k|^2}\right\}+o_{N\rightarrow \infty}(1),
    \end{align}
    where $\mathfrak{a}_N$ is the box scattering length defined in Eq.~(\ref{Eq:Box_Scattering_Length}).
\end{thm}

    While Eq.~(\ref{Eq:GP_GSE}) claims an identity, our proof will focus on the lower bound, as the upper bound is discussed in detail in Theorem \ref{Th:Upper_Bound}. Besides commenting on the length and complexity of the proof of Theorem \ref{Th:GP}, we also want to emphasise that our proof of Theorem \ref{Th:GP} and Theorem \ref{Th:Beyond_GP}, and especially Theorem \ref{Th:Upper_Bound}, do not require any cut-off parameters. However, it can be useful to introduce such a parameter anyways, as it shortens the proof of Theorem \ref{Th:GP} and allows us to slightly extend the range of possible values $\kappa$ in case of Theorem \ref{Th:Beyond_GP}, to be precise we can verify Theorem \ref{Th:Beyond_GP} without any cut-off parameter for $0\leq \kappa<\frac{1}{11}$. Furthermore we want to point out that the error terms in the proof of the lower bound in Theorem \ref{Th:GP} can be controlled by the particle number operator alone, and do not require the kinetic energy.

    In the subsequent Theorem \ref{Th:Beyond_GP} we extend the results in \cite{BCS} concerning the excitation spectrum of the operator $H_{N,\kappa}$ for $0\leq \kappa<\kappa_0$, where $\kappa_0$ is of the order of magnitude $10^{-3}$, to the significantly large range $0\leq \kappa<\frac{1}{8}$. We do however want to emphasise, that our proof of Theorem \ref{Th:Beyond_GP} relies on strong a priori estimates provided by recent results in \cite{F}, and later in a slightly different setting by \cite{BS,HHNST}.

\begin{thm}
\label{Th:Beyond_GP}  
Let $E_{N,\kappa}$ be the ground state energy of $H_{N,\kappa}$, $0\leq \kappa<\frac{1}{8}$ and assume that $V$ is radially symmetric decreasing. With the definition $\tau:=\frac{5}{2}\! \left(\frac{1}{8}-\kappa\right)>0$, $E_{N,\kappa}$ is given by 
\begin{align}
    \label{Eq:Beyond_GP_GSE}
   &  4\pi\mathfrak{a}_{N^{1-\kappa}}N^{\kappa} (N \! - \! 1) \! + \! \frac{1}{2} \! \sum_{k\neq 0} \! \left\{ \! \sqrt{|k|^4 \! + \! 16\pi \mathfrak{a} N^{\kappa}|k|^2} \! - \! |k|^2 \! - \! 8\pi \mathfrak{a} N^{\kappa} \! + \! \frac{(8\pi \mathfrak{a} N^{\kappa})^2}{2|k|^2} \! \right\}\! + \! O_{N\rightarrow \infty}\! \left(N^{-\tau}\right).
\end{align}
Let furthermore $E_{N,\kappa}^{(d)}$ denote the $d$-th eigenvalue, in increasing order, and let us enumerate the set $\left\{\sum_{k\neq 0}n_k \sqrt{|k|^4+16\pi \mathfrak{a} N^\kappa |k|^2}: n_k\in \mathbb{N}_0\right\}$ in increasing order by $\{\lambda_{N,\kappa}^{(1)},\lambda_{N,\kappa}^{(2)},\dots \}$. Then 
\begin{align*}
    E^{(d)}_{N,\kappa} -E_{N,\kappa} =  \lambda_{N,\kappa}^{(d)}   +  O_{N\rightarrow \infty}\! \left(N^{-\tau}\right) .
\end{align*}
\end{thm}

Again, the proof of Theorem \ref{Th:Beyond_GP} focuses on the lower bound, as the upper bound is discussed in a more general setting in the subsequent Theorem \ref{Th:Upper_Bound}. We want to point out here, that even though our algebraic manipulations do not make use of unitary operations, they are equally well suited to provide lower bounds and upper bounds on the ground state energy as well as the excitation spectrum. Furthermore we want to note that our resolution of the spectrum up to the order $N^{-\tau}$ is sharp enough to see the non-linear contribution of $|k|^4$ in the excitations $\sqrt{|k|^4+16\pi \mathfrak{a} N^\kappa |k|^2}$ for the slightly smaller range $0\leq \kappa<\frac{5}{48}$.

In our final Theorem \ref{Th:Upper_Bound}, we provide sufficient upper bounds on the excitation energy $E_{N,\kappa}^{(d)}$ in order to conclude the proof of Theorem \ref{Th:GP} and Theorem \ref{Th:Beyond_GP}. Furthermore, we derive an improved, and novel, upper bound on the ground state energy $E_N$ in the Gross-Pitaevskii regime, which gives a correction of the order $\frac{\log N}{N}$. We want to emphasise that the scale $\frac{\log N}{N}$ is indeed the expected order of magnitude for the next term in the asymptotic expansion in Eq.~(\ref{Eq:GP_GSE}), see \cite{W}, and we believe it is crucial at this point that our method does not rely on cut-off techniques in momentum space. 

\begin{thm}
\label{Th:Upper_Bound}
In the Gross-Pitaevskii regime $\kappa:=0$, we have the upper bound
\begin{align}
\label{Eq:Upper_Bound_Ground_State}
 E_{N} & \leq 4\pi\mathfrak{a}_N (N-1)+\frac{1}{2}\sum_{k\neq 0}\left\{\sqrt{|k|^4+16\pi \mathfrak{a} |k|^2}-|k|^2-8\pi \mathfrak{a}+\frac{(8\pi \mathfrak{a})^2}{2|k|^2}\right\}+C \frac{\log N}{N}
\end{align}
for a suitable $C>0$. Furthermore for $\kappa<\frac{2}{13}$ and $d\in \mathbb{N}$, $E_{N,\kappa}^{(d)}$ is bounded from above by
\begin{align}
\label{Eq:Upper_Bound_Excited_State}
4\pi\mathfrak{a}_N N^\kappa (N \! - \! 1) \! + \! \frac{1}{2}\sum_{k\neq 0}\left\{\sqrt{|k|^4 \! + \! 16\pi \mathfrak{a}  N^\kappa |k|^2} \! - \! |k|^2 \! - \! 8\pi \mathfrak{a} N^\kappa  \! + \! \frac{(8\pi \mathfrak{a} N^\kappa )^2}{2|k|^2}\right\} \! + \! \lambda_{N,\kappa}^{(d)} \! + \! C N^{\frac{13\kappa}{4}-\frac{1}{2}}.
\end{align}
\end{thm}

\section{The two-body Problem}
\label{Sec:The two-body Problem}
Before we come to the (approximate) diagonalization of the many-body operator $H_{N,\kappa}$, which will be the basis of verifying Theorem \ref{Th:GP}, Theorem \ref{Th:Beyond_GP} and Theorem \ref{Th:Upper_Bound}, we will first investigate the corresponding problem for the two body Hamiltonian $H:=-\Delta_2+V_L(x,y)$ defined on $L^2\! \left(\Lambda^2\right)$ with $\Delta_2:=\Delta_x+\Delta_y$. The goal of this Section is to find a transformation $S:L^2\! \left(\Lambda^2\right) \longrightarrow L^2\! \left(\Lambda^2\right)$, which removes the correlations of $H$ between low energy momentum pairs $(k_1,k_2)\in \mathcal{L}\subseteq \left(2\pi \mathbb{Z}\right)^6$ and high energy momentum pairs $\mathcal{H}:=\left(2\pi \mathbb{Z}\right)^6\setminus \mathcal{L}$. We will specify the set $\mathcal{L}$ later this section, for now let $\pi_{\mathcal{L}}$ and $\pi_{\mathcal{H}}=1-\pi_{\mathcal{L}}$ denote the corresponding projections in $L^2\! \left(\Lambda^2\right)$. With $\pi_\mathcal{L}$ and $\pi_\mathcal{H}$ at hand, $H$ can be represented as a block matrix
\begin{align*}
    H=\begin{bmatrix}
\pi_\mathcal{L} (-\Delta_2+V_L)\pi_\mathcal{L}   & \pi_\mathcal{L} V_L\pi_\mathcal{H}  \\
\pi_\mathcal{H} V_L\pi_\mathcal{L} & \pi_\mathcal{H} (-\Delta_2+V_L)\pi_\mathcal{H}  
\end{bmatrix}.
\end{align*}
It is now an elementary task in linear algebra to remove the correlation terms $\pi_\mathcal{H} V_L\pi_\mathcal{L}$ and $\pi_\mathcal{L} V_L\pi_\mathcal{H}$, and therefore bringing the operator $H$ in a block diagonal form. For this purpose let us define the Feshbach–Schur map $S:=1-R  V_L\pi_\mathcal{L}=\begin{bmatrix}
1   & 0  \\
-R  V_L\pi_\mathcal{L} & 1  
\end{bmatrix}$, where $R$ is the pseudo-inverse of $\pi_\mathcal{H} (-\Delta_2+V_L)\pi_\mathcal{H}$, and compute
\begin{align}
\label{Eq:Feshbach–Schur}
    S^\dagger H S & =\begin{bmatrix}
\pi_\mathcal{L} \left(-\Delta_2+V_L-V_L R V_L\right)\pi_\mathcal{L}   & 0  \\
0 & \pi_\mathcal{H} (-\Delta_2+V_L)\pi_\mathcal{H}
\end{bmatrix}.
\end{align}
Consequently $H=T^{\dagger}\left(-\Delta_2+\widetilde{V}_L\right)T$ where $T:=1+R  V_L\pi_\mathcal{L}=\begin{bmatrix}
1   & 0  \\
R  V_L\pi_\mathcal{L} & 1  
\end{bmatrix}$ is the inverse of the Feshbach–Schur map $S$ and the renormalized potential is defined as
\begin{align}
\label{Eq:Block_Potential}
  \widetilde{V}_L:=\begin{bmatrix}
\pi_\mathcal{L} \left(V_L-V_L R V_L\right)\pi_\mathcal{L}   & 0  \\
0 & \pi_\mathcal{H} V_L \pi_\mathcal{H}
\end{bmatrix}.  
\end{align}
In the following we will always use the concrete choice $\mathcal{L}:=\bigcup_{|k|<K}\{( k,0),(0, k)\}$, for a given $0< K\leq \infty$. At this point we want to emphasise that in the case $K=\infty$, the definition of $T$ and $\widetilde{V}_L$ do not include any cut-off parameter in momentum space. 

We can now relate the renormalized potential $V_{L}-V_{L} R V_{L}$ with the scattering properties of the potential $V(x)$. Following \cite{HST}, let us first define the box scattering length
\begin{align}
\label{Eq:Box_Scattering_Length}
\mathfrak{a}_L:=\frac{L}{8\pi}\Big(V_{L}-V_{L} R V_{L}\Big)_{0 0 ,00}=\frac{1}{8\pi}\left(\int V(x)\mathrm{d}x-L\braket{V_L,R V_L}\right),
\end{align}
which is independent of the choice of $K$. It has been shown in \cite{H} that the box scattering length satisfies $|\mathfrak{a}_L-\mathfrak{a}|\lesssim L^{-1}$, where $\mathfrak{a}$ is the scattering length of $V(x)$ introduced in Section \ref{Sec:Introduction}. Furthermore, also the other matrix entries of the renormalized potential $V_{L}-V_{L} R V_{L}$ appearing in Eq.~(\ref{Eq:Block_Potential}) can be related to the scattering length as we demonstrate in Lemma \ref{Lem:Coefficient_Comparison}.

\begin{lem}
\label{Lem:Coefficient_Comparison}
Let $V\in L^1 \! \left(\mathbb{R}^3\right)$. Then we have for all $k_i\in 2\pi\mathbb{Z}^3$ satisfying $k_1+k_2=k_3+k_4$
    \begin{align}
    \label{Eq:Coefficients_Renormalized_Potential}
        \left|L\Big(V_{L}-V_{L} R V_{L}\Big)_{k_1 k_2, k_3 k_4}-8\pi \mathfrak{a} \right|\leq CL^{-1}  \! \left( \! 1+\sum_{i=1}^4 |k_i| \! \right),
    \end{align}
   and $\left|L\Big(V_{L}-V_{L} R V_{L}\Big)_{k_1 k_2, k_3 k_4}\right|\leq C$ uniformly in $k_i$, for a suitable $C>0$. All estimates are uniform in the parameter $K$ introduced below Eq.~(\ref{Eq:Block_Potential}).
\end{lem}

In order to keep the proof of Lemma \ref{Lem:Coefficient_Comparison} in the Appendix \ref{Appendix:Coefficients} self-contained, we will provide a proof of the fact that $|\mathfrak{a}_L-\mathfrak{a}|\lesssim L^{-1}$ as well in Lemma \ref{Lem:Scattering_Comparison}. We want to emphasize at this point that the bounds in Lemma \ref{Lem:Coefficient_Comparison}, as well as the other estimates in this manuscript, are uniform in the $L^1$-norm of $V$. Furthermore it turns out that with the exception of Section \ref{Sec:Trial_States_and_their_Energy} and Appendix \ref{Appendix:A_priori_Condensation}, all the bounds depend only on quantities related to the renormalized potential $V_{N^{1-\kappa}}-V_{N^{1-\kappa}} R V_{N^{1-\kappa}}$.

\section{The many-body Diagonalization}
In this section we want to lift the diagonalization procedure for the two-particle Hamiltonian derived in Section \ref{Sec:The two-body Problem} to a diagonalization of the many-body Hamiltonian $H_{N,\kappa}$. This will be done by the introduction of a many-body counterpart to the inverse Feshbach–Schur map $T:L^2\! \left(\Lambda^2\right) \longrightarrow L^2\! \left(\Lambda^2\right)$ defined below Eq.~(\ref{Eq:Feshbach–Schur}), which is realized by the new sets of variables
\begin{align}
\label{Eq:Definition_c}
    c_k: & =a_k+\sum_{j , mn}\left(T-1\right)_{jk, mn} a_j^\dagger a_m a_n,\\
    \label{Eq:Definition_psi}
    \psi_{j k}: & =a_j a_k+\sum_{mn}\left(T-1\right)_{j k, m n} a_m a_n.
\end{align}
In order to illustrate the intimate relation between these variables and the underlying transformation $T$, we compute, using the canonical commutation relations (CCR) $\left[a_j,a_k^\dagger\right]=\delta_{j,k}$,
\begin{align}
\label{Eq:Transformed_Operator}
   & \ \ \  \ \ \  \ \ \ \sum_k |k|^2 c_k^\dagger c_k \! + \! \frac{1}{2}\sum_{jk,mn}\left(\widetilde{V}_{N^{1-\kappa}}\right)_{jk,mn}\psi_{j k}^\dagger \psi_{m n}   \\
   \nonumber
   & = \! \sum_k |k|^2 a_k^\dagger a_k \! + \! \frac{1}{2}\sum_{jk,mn}\bigg(T^\dagger\! \! \left(-\Delta_2 \! + \! \widetilde{V}_{N^{1-\kappa}}\right)\! T \! + \! \Delta_2\bigg)_{j k,mn}a_k^\dagger a_j^\dagger a_m a_n+\mathcal{R},
\end{align}
where $\mathcal{R}:=\sum_{jk,mn;j',m'n'}|k|^2 \overline{(T-1)_{j' k,m'n'}}(T-1)_{j k,m n} a_{n'}^\dagger a_{m'}^\dagger a_j^\dagger a_{j'}a_m a_n$ and $\widetilde{V}_{N^{1-\kappa}}$ is the renormalized potential defined in Eq.~(\ref{Eq:Block_Potential}). Using the coefficients 
\begin{align}
\label{Eq:w_Coefficient}
    w_k:=N(T-1)_{(-k)k,0 0}=\frac{N}{2|k|^2}\Big(V_{N^{1-\kappa}}-V_{N^{1-\kappa}} R V_{N^{1-\kappa}}\Big)_{(-k)k,0 0} \ \ \underset{N\rightarrow \infty}{\simeq} \ \ \frac{4\pi \mathfrak{a}N^\kappa}{|k|^2},
\end{align}
the residuum $\mathcal{R}$ can further be decomposed into a single-particle operator $\sum_{k\neq 0}|k|^2 w_k^2 a_k^\dagger a_k$ and an error term $\mathcal{E}_1:=\mathcal{R}-\sum_{k\neq 0}|k|^2 w_k^2 a_k^\dagger a_k$, which is small given the assumption that the number of excited particles $\mathcal{N}:=\sum_{k\neq 0}a_k^\dagger a_k$ is small compared to the total number of particles $N$, as we confirm later this Section. Hence the computation of the many-body operator in Eq.~(\ref{Eq:Transformed_Operator}) essentially reduces to the computation of the two-body operator
\begin{align*}
    T^\dagger \left(-\Delta_2 \! + \! \widetilde{V}_{N^{1-\kappa}}\right) T = -\Delta_2+V_{N^{1-\kappa}},
\end{align*}
which has been carried out in Section \ref{Sec:The two-body Problem}, see the comment below Eq.~(\ref{Eq:Feshbach–Schur}). Combining what we have so far, we can represent $H_{N,\kappa}$ in terms of the new variables $c_k$ and $\psi_{j k}$ as
\begin{align*}
    H_{N,\kappa}=\sum_k |k|^2 c_k^\dagger c_k \! + \! \frac{1}{2}\sum_{jk,mn}\left(\widetilde{V}_{N^{1-\kappa}}\right)_{jk,mn}\psi_{j k}^\dagger \psi_{m n}-\mathcal{R}.
\end{align*}
The advantage of this representation compared to Eq.~(\ref{Eq:Hamiltonian}) is that the renormalized potential $\widetilde{V}_{N^{1-\kappa}}$ has the block diagonal structure $\widetilde{V}_{N^{1-\kappa}}=\pi_{\mathcal{L}} (V_{N^{1-\kappa}}-V_{N^{1-\kappa}} R V_{N^{1-\kappa}})\pi_{\mathcal{L}}+\pi_{\mathcal{H}}V_{N^{1-\kappa}} \pi_{\mathcal{H}}$, see the definition of $\widetilde{V}_{N^{1-\kappa}}$ in Eq.~(\ref{Eq:Block_Potential}). Using the positivity of $V_{N^{1-\kappa}}$ therefore yields
\begin{align}
\nonumber
    H_{N,\kappa} & \geq \sum_k |k|^2 c_k^\dagger c_k \! + \! \frac{1}{2}\sum_{jk,mn}\Big(\pi_{\mathcal{L}} (V_{N^{1-\kappa}}-V_{N^{1-\kappa}} R V_{N^{1-\kappa}})\pi_{\mathcal{L}}\Big)_{jk,mn}\psi_{j k}^\dagger \psi_{m n}-\mathcal{R}\\
    \nonumber
    & = \sum_k |k|^2 c_k^\dagger c_k \! + \! \frac{1}{2}\sum_{jk,mn}\Big(\pi_{\mathcal{L}} (V_{N^{1-\kappa}}-V_{N^{1-\kappa}} R V_{N^{1-\kappa}})\pi_{\mathcal{L}}\Big)_{jk,mn}a_k^\dagger a_j^\dagger a_m a_n-\mathcal{R}\\
    \label{Eq:Throwing_Away_High_Momenta}
    &= \sum_k |k|^2 c_k^\dagger c_k + \frac{\sigma_0}{N}\,  a_0^{2\dagger}a_0^2+ \! \! \sum_{0<|k|< K} \! \!\frac{4\sigma_k}{N}\, a_0^\dagger a_0\, a_k^\dagger a_k-\mathcal{R},
\end{align}
with $\sigma_k:=\frac{N}{4}\Big(V_{N^{1-\kappa}}-V_{N^{1-\kappa}} R V_{N^{1-\kappa}}\Big)_{k 0 ,k0}+\frac{N}{4}\Big(V_{N^{1-\kappa}}-V_{N^{1-\kappa}} R V_{N^{1-\kappa}}\Big)_{0 k ,k0} \ \underset{N\rightarrow \infty}{\simeq} \ 4\pi \mathfrak{a}N^\kappa$, where we have used that $ \psi_{m n}=a_m a_n$ in case at least one of the indices $m$ or $n$ is zero. We further observe that $a_0^\dagger a_0=N-\mathcal{N}$ and $a_0^{2\dagger}a_0^2=N(N-1)-(2N-1)\mathcal{N}+\mathcal{N}^2$ hold on the $N$-particle Hilbert space $L^2_\mathrm{sym} \! \left(\Lambda^N\right)$, where $\mathcal{N}$ is defined below Eq.~(\ref{Eq:w_Coefficient}), which yields 
\begin{align*}
 \frac{\sigma_0}{N} a_0^{2\dagger}a_0^2+ \!  \!  \! \sum_{0<|k|< K} \frac{4\sigma_k}{N} a_0^\dagger a_0\, a_k^\dagger a_k=4\pi \mathfrak{a}_{N^{1-\kappa}}N^\kappa (N-1)+ \!  \!  \! \sum_{0<|k|< K} 4\sigma_k \, a_k^\dagger a_k-2\sigma_0\mathcal{N}-\mathcal{E}_2   
\end{align*}
with $\mathcal{E}_2:=\sum_{0<|k|< K}\frac{4\sigma_k}{N} \mathcal{N}a_k^\dagger a_k-\frac{\sigma_0}{N} \left(\mathcal{N}^2+\mathcal{N}\right)$, where we used $\sigma_0=4\pi \mathfrak{a}_{N^{1-\kappa}}N^\kappa $ by the definition of $\mathfrak{a}_{N^{1-\kappa}}$ in Eq.~(\ref{Eq:Box_Scattering_Length}). Utilizing $\mathcal{E}_1$ introduced below Eq.~(\ref{Eq:w_Coefficient}), we therefore obtain
\begin{align}
\label{Eq:Pseudo_Quadratic}
    H_{N,\kappa}-4\pi \mathfrak{a}_{N^{1-\kappa}}N^\kappa (N-1)\geq \sum_k |k|^2 c_k^\dagger c_k + \sum_{k\neq 0} \mu_k a_k^\dagger a_k-\mathcal{E}_1-\mathcal{E}_2
\end{align}
with $\mu_k:=\mathds{1}(|k|<K)4\sigma_k-2\sigma_0-|k|^2 w_k^2\ \underset{N, K\rightarrow \infty}{\simeq}8\pi\mathfrak{a}N^{\kappa}-\frac{\left(4\pi\mathfrak{a}N^\kappa\right)^2}{|k|^2}$. In the following we want to apply the theory of Bogoliubov operators in order to analyse the operator on the right hand side of Eq.~(\ref{Eq:Pseudo_Quadratic}). However the operators $c_k$ do not satisfy the CCR, not even in an approximate sense, and they do not leave $L^2_\mathrm{sym} \! \left(\Lambda^N\right)$ invariant. It is well understood at this point, that the second issue can be avoided by including an appropriate phase factor $a_0^\dagger \frac{1}{\sqrt{a_0 a_0^\dagger }}$, as has been done in \cite{S,DN,LNSS}. Regarding the first issue we define new variables
\begin{align}
\label{Eq:Definition_d}
    d_k:=a_0^\dagger \frac{1}{\sqrt{a_0 a_0^\dagger }}c_k-w_k a_{-k}^\dagger\frac{1}{\sqrt{a_0 a_0^\dagger }} a_0 ,
\end{align}
as well as the increment $\delta_k:=a_0^\dagger \frac{1}{\sqrt{a_0a_0^\dagger }}a_k-d_k$. As we will see in Lemma \ref{Lem:Delta_Control}, the increments can be regarded as being small and therefore the operators $d_k$ do satisfy approximate CCR. With $\delta_k$ at hand, we can now rewrite $ a_k^\dagger a_k = a_k^\dagger \frac{1}{\sqrt{a_0a_0^\dagger }}a_0 a_0^\dagger \frac{1}{\sqrt{a_0a_0^\dagger }}a_k=\left(d_k+\delta_k\right)^\dagger \left(d_k+\delta_k\right)$ and
\begin{align*}
    c_k^\dagger c_k \!  = \! c_k^\dagger   \frac{1}{\sqrt{a_0a_0^\dagger }}a_0 a_0^\dagger \frac{1}{\sqrt{a_0 a_0^\dagger}}c_k \! = \! \left(d_k \! + \! w_k d_{-k}^\dagger  \! + \! w_k\delta_{-k}^\dagger\right)^\dagger \!  \left(d_k \! + \! w_k d_{-k}^\dagger  \! + \! w_k\delta_{-k}^\dagger\right).
\end{align*}
Defining the coefficients $A_k:=|k|^2+|k|^2 w_k^2+\mu_k-\epsilon$, $B_k:=2|k|^2 w_k$ and $C_k:=2|k|^2 w_k^2$, we consequently obtain for $\epsilon\geq 0$ the identity
\begin{align}
\nonumber
    &  \sum_k |k|^2c_k^\dagger c_k \! +  \! \sum_{k\neq 0} \mu_k a_k^\dagger a_k  \! - \! \epsilon\mathcal{N}=   \sum_k  |k|^2  \left(d_k \! + \! w_k d_{-k}^\dagger \right)^\dagger \! \!  \left(d_k \! + \! w_k d_{-k}^\dagger\right) \! + \! \sum_{k\neq 0}(\mu_k \! - \! \epsilon)d_k^\dagger d_k \! - \! \mathcal{E}_3\\
   \label{Eq:Writen_In_d} 
    & \ =\frac{1}{2}  \sum_{k\neq 0} \! \left\{A_k\left(d_k^\dagger d_k \! + \! d_{-k}^\dagger d_{-k}\right) \! + \! B_k\left(d_k d_{-k} \! + \! d_{-k}^\dagger d_k^\dagger\right) \! + \! C_k \frac{[d_k,d_k^\dagger] \! + \! [d_{-k},d_{-k}^\dagger]}{2}\right\}\! -  \!  \mathcal{E}_3,
\end{align}
with 
\begin{align}
\label{Eq:Definition_third_Error}
  \mathcal{E}_3:=\sum_{k\neq 0}(\epsilon \! - \! \mu_k) \!  \! \left(\delta_k^\dagger \delta_k \!  + \! d_k^\dagger \delta_k \! + \! \mathrm{H.c.} \! \right) \! - \! \sum_k |k|^2  \! \left(w_k^2 \delta_k \delta_k^\dagger  \! +  \! w_k \left(d_k \! + \! w_k d_{-k}^\dagger\right)\delta_{-k}^\dagger  \! + \! \mathrm{H.c.}  \! \right).  
\end{align}
By Lemma \ref{Lem:Coefficient_Comparison} we have $A_k  \underset{N,K\rightarrow \infty}{\simeq} |k|^2+8\pi\mathfrak{a}N^\kappa-\epsilon$, as well as $B_k  \underset{N\rightarrow \infty}{\simeq} 8\pi\mathfrak{a}N^\kappa$ and $C_k  \underset{N \infty}{\simeq} \frac{\left(8\pi\mathfrak{a}N^\kappa\right)^2}{2|k|^2}$, and furthermore $A_k\geq B_k$ for $N\geq N_0$, $K\geq N^{\frac{\kappa}{2}}K_0$ and $\epsilon<\epsilon_0$. Consequently we can apply Bogoliubov's Lemma, see \cite{LS,NNRT}, which yields
\begin{align}
  \nonumber
   \frac{1}{2}\sum_{k\neq 0} & \left\{A_k\left(d_k^\dagger d_k+d_{-k}^\dagger d_{-k}\right)+B_k\left(d_k d_{-k}+d_{-k}^\dagger d_k\right)+C_k \frac{[d_k,d_k^\dagger]+[d_{-k},d_{-k}^\dagger]}{2}\right\}\\
 \label{Eq:Bogoliubov_Lemma}
    &\ \ \ \ \ \ \geq \frac{1}{2}\sum_{k\neq 0}\left\{\sqrt{A_k^2-B_k^2}-A_k+C_k\right\}\frac{[d_k,d_k^\dagger]+[d_{-k},d_{-k}^\dagger]}{2}.
\end{align}
Defining $\mathcal{E}_4:=\frac{1}{2}\sum_{k\neq 0}\left\{\sqrt{A_k^2-B_k^2}-A_k+C_k\right\}\left(1-\frac{[d_k,d_k^\dagger]+[d_{-k},d_{-k}^\dagger]}{2}\right)$, and combining the estimate in Eq.~(\ref{Eq:Pseudo_Quadratic}) with the identity in Eq.~(\ref{Eq:Writen_In_d}) and the lower bound in Eq.~(\ref{Eq:Bogoliubov_Lemma}), yields
\begin{align}
\label{Eq:Main}
    H_{N,\kappa}  -  4\pi \mathfrak{a}_{N^{1-\kappa}}N^\kappa (N-1)- \frac{1}{2}\sum_{k\neq 0}\left\{\sqrt{A_k^2-B_k^2}-A_k+C_k\right\}\geq \epsilon \mathcal{N}-\sum_{i=1}^4 \mathcal{E}_i,
\end{align}
for $\epsilon, N$ and $K$ satisfying the restrictions stated above Eq.~(\ref{Eq:Bogoliubov_Lemma}). We want to emphasise that the constant on the left hand side of Eq.~(\ref{Eq:Main}) contains the correct leading order energy $4\pi \mathfrak{a}_{N^{1-\kappa}}N^\kappa (N-1)$, and in the limit $N,K\rightarrow \infty$ and $\epsilon\rightarrow 0$, it contains the sub-leading Lee-Huang-Yang correction $\frac{1}{2} \! \sum_{k\neq 0} \! \left\{ \! \sqrt{|k|^4 \! + \! 16\pi \mathfrak{a} N^{\kappa}|k|^2} \! - \! |k|^2 \! - \! 8\pi \mathfrak{a} N^{\kappa} \! + \! \frac{(8\pi \mathfrak{a} N^{\kappa})^2}{2|k|^2} \! \right\}$ as well, which follows from the asymptotic behaviour of the coefficients $A_k$, $B_k$ and $C_k$ stated below Eq.~(\ref{Eq:Definition_third_Error}).

 \subsection{Proof of Theorem \ref{Th:GP}}
In order to verify Theorem \ref{Th:GP}, we need sufficient bounds on the error terms $\mathcal{E}_i$. For this purpose, we first derive estimates on the increment $\delta_k$ in the subsequent Lemma \ref{Lem:Delta_Control}, which we will use in Lemma \ref{Lem:GP_Error_Estimate} in order to control $\mathcal{E}_i$. For the rest of this subsection we will always assume that we are in the Gross-Pitaevskii Regime, i.e. we assume $\kappa=0$. Furthermore recall that we are working on the $N$-particle Hilbert space $L^2_\mathrm{sym} \! \left(\Lambda^N\right)$ and the claimed estimates hold restricted to this space.
\begin{lem}
\label{Lem:Delta_Control}
    There exists a $C>0$ such that $\sum_{k\neq 0}\left(\delta_k^\dagger \delta_k+\delta_k \delta_k^\dagger\right)\leq \frac{C}{N-\mathcal{N}}(\mathcal{N}+1)^2$.
\end{lem}
\begin{proof}
    Let us define $h(x):=\frac{1}{\sqrt{1+\frac{1}{N}-x}}-\sqrt{1-x}$ as well as the coefficients
    \begin{align}
    \label{Eq:Definition_f}
       & \ \ \ \ \ \ \ \ \ \ \ \ \ f_{\ell,k}:=  \left(\left(T-1\right)_{(\ell-k)k, \ell 0}+\left(T-1\right)_{(\ell-k)k, 0 \ell}\right)\\
      \nonumber
      &=\frac{1}{|k|^2 \! + \! |\ell \! - \! k|^2}\left\{\Big(V_{N^{1-\kappa}} \! - \! V_{N^{1-\kappa}} R V_{N^{1-\kappa}}\Big)_{(\ell-k)k, \ell 0} \! + \! \Big(V_{N^{1-\kappa}} \! - \! V_{N^{1-\kappa}} R V_{N^{1-\kappa}}\Big)_{(\ell -k)k, 0 \ell }\right\},
    \end{align}
    in case $\ell\neq 0$, and $f_{\ell,k}:=0$ otherwise. Then we can write $\delta_k=-\delta_k'+\delta_k''$ with $\delta_k':  =\sqrt{a_0^\dagger a_0}\sum_{|\ell|<K}f_{\ell,k}\, a^{\dagger}_{\ell-k}a_\ell$ and $\delta_k'':  =w_k\left(\frac{1}{\sqrt{a_0a_0^\dagger }}a_0-\frac{\sqrt{a_0^\dagger a_0}a_0}{N}\right)a_{-k}^\dagger=w_k a_{-k}^\dagger \frac{a_0}{\sqrt{N}} h\!\left(\frac{\mathcal{N}}{N}\right) $. Note that $|h(x)|\lesssim \frac{x+1}{\sqrt{1-x}}$, $|w_k|\lesssim \frac{1}{|k|^2}$ and $\left(\frac{a_0}{\sqrt{N}}\right)^\dagger \frac{a_0}{\sqrt{N}}\leq 1$. Consequently
    \begin{align*}
        \sum_{k\neq 0} (\delta_k'')^\dagger \delta_k''\lesssim \sum_{k\neq 0}h\!\left(\frac{\mathcal{N}}{N}\right)\frac{a_{-k}a_{-k}^\dagger}{|k|^4} h\!\left(\frac{\mathcal{N}}{N}\right) \! \lesssim \! h\!\left(\frac{\mathcal{N}}{N}\right)^2 \! \! (\mathcal{N} \! + \! 1) \! \lesssim  \! \frac{(\mathcal{N}+1)^3}{N(N \! - \! \mathcal{N})},
    \end{align*}
   where we have used $\sum_{k\neq 0}\frac{1}{|k|^4}<\infty$. Similarly the same estimate can be shown for $\sum_{k\neq 0} \delta_k''(\delta_k'')^\dagger $. In order to estimate $\delta'_k$, note that $|f_{\ell,k}|\leq \frac{D }{N |k|^2}$ for a suitable constant $D$ by Lemma \ref{Lem:Coefficient_Comparison}. In combination with $\sqrt{a_0^\dagger a_0}\leq \sqrt{N}$ this immediately yields the bounds $\pm \mathfrak{Re}\delta_k'\leq \frac{D }{\sqrt{N} |k|^2}\mathcal{N}$ and $\pm \mathfrak{Im}\delta_k'\leq \frac{D }{\sqrt{N} |k|^2}\mathcal{N}$. Since $\mathcal{N}$ commutes with both $\mathfrak{Re}\delta_k'$ and $\mathfrak{Im}\delta_k'$ we can square these inequalities and obtain $\left(\delta'_k\right)^\dagger \delta'_k+\delta'_k \left(\delta'_k\right)^\dagger\leq 4\left(\mathfrak{Re}\delta_k'\right)^2+4\left(\mathfrak{Im}\delta_k'\right)^2\leq \frac{8D^2}{N|k|^4}\mathcal{N}^2$.
\end{proof}
In the following let $\mathcal{F}^+_M\subset L^2_\mathrm{sym} \! \left(\Lambda^N\right)$ denote the spectral subspace $\mathcal{N}\leq M$, and let us define for an operator $X$ the restricted operator $X\Big|_{\mathcal{F}^+_M}$ as $X\Big|_{\mathcal{F}^+_M}:=\pi_{\mathcal{F}^+_M} X \pi_{\mathcal{F}^+_M}$, where $\pi_{\mathcal{F}^+_M}$ is the orthogonal projection onto $\mathcal{F}^+_M$.
\begin{lem}
\label{Lem:GP_Error_Estimate}
   For all $K<\infty$ and $0<r<1$, there exists a constant $C_{K,r}>0$ such that
    \begin{align*}
    \pm \mathcal{E}_i\Big|_{\mathcal{F}^+_M}\leq C_{K,r} \sqrt{\frac{M+1}{N}}(\mathcal{N}+1).
    \end{align*} 
     for all $M\leq \min\{rN,N-1\}$ and $i\in \{1,2,3,4\}$.
\end{lem}
\begin{proof}
   Using $|\sigma_k|\lesssim 1$, see Lemma \ref{Lem:Coefficient_Comparison}, we have $\pm \mathcal{E}_2\leq C N^{-1}(\mathcal{N}+1)^2$, which concludes the case $i=2$. Regarding the case $i=3$, we have by Cauchy-Schwarz and Lemma \ref{Lem:Delta_Control}
    \begin{align*}
        \pm \mathcal{E}_3\leq \epsilon \sum_{k\neq 0}\left\{a_k^\dagger a_k+c_k^\dagger c_k+d_k^\dagger d_k+(d_k+w_k d_{-k}^\dagger)^\dagger (d_k+w_k d_{-k}^\dagger)\right\}+\epsilon^{-1} \frac{C}{N-\mathcal{N}}(\mathcal{N}+1)^2
    \end{align*}
   for any $\epsilon>0$. Since $d_k=\frac{1}{\sqrt{a_0^\dagger a_0}}a_0^\dagger a_k-\delta_k$ we have $\sum_{k\neq 0}d_k^\dagger d_k\lesssim \mathcal{N}$ again by Lemma \ref{Lem:Delta_Control}. Furthermore $\sum_{k\neq 0}c_k^\dagger c_k=\sum_{k\neq 0}\left(d_k+w_k a_{-k}^\dagger a_0\frac{1}{\sqrt{a_0^\dagger a_0}}\right)^\dagger \left(d_k+w_k a_{-k}^\dagger a_0\frac{1}{\sqrt{a_0^\dagger a_0}}\right)\lesssim \mathcal{N}+\sum_{k\neq 0}w_k^2 a_k a_k^\dagger\lesssim \mathcal{N}+1$. Similarly $\sum_{k\neq 0}(d_k+w_k d_{-k}^\dagger)^\dagger (d_k+w_k d_{-k}^\dagger)\lesssim \mathcal{N}+1$. Hence
   \begin{align*}
       \pm \mathcal{E}_3\Big|_{\mathcal{F}^+_M}\lesssim \epsilon (\mathcal{N}+1)+\epsilon^{-1} \frac{M+1}{N-M}(\mathcal{N}+1),
   \end{align*}
   which concludes the proof of the case $i=3$ for the optimal choice $\epsilon:=\sqrt{\frac{M+1}{N-M}}$. Regarding the case $i=4$, note that $\left[a_0^\dagger \frac{1}{\sqrt{a_0 a_0^\dagger}}a_k,a_k^\dagger \frac{1}{\sqrt{a_0 a_0^\dagger}}a_0\right]\Big|_{\mathcal{F}^+_{N-1}}=1|_{\mathcal{F}^+_{N-1}}$, and therefore $\mathcal{E}_4|_{\mathcal{F}^+_{N-1}}$ reads
\begin{align*}
    \frac{1}{2}\sum_{k\neq 0}\left\{A_k-\sqrt{A_k^2-B_k^2}-C_k\right\}\left(\left[\delta_k,a_k^\dagger \frac{1}{\sqrt{a_0 a_0^\dagger}}a_0\right]+\left[a_0^\dagger \frac{1}{\sqrt{a_0 a_0^\dagger}}a_k,\delta_k^\dagger\right]+\left[\delta_k,\delta_k^\dagger\right]\right)\Bigg|_{\mathcal{F}^+_{N-1}}.
\end{align*}
Multiplying out the commutators and estimating the resulting products in the same way as we did in the case $i=3$, concludes the proof of the case $i=4$. Regarding the final case $i=1$, let us write $\mathcal{E}_1=\frac{1}{2}\sum_{(\ell,\ell')\in A_K}\left(\mathcal{G}^{\ell,\ell'}+\mathrm{H.c.}\right)+\sum_{k\neq 0}|k|^2 w_k^2 N^{-2}\left(a_0^{2\dagger }a_0^2-N^2\right) a_k^\dagger a_k$ with $A_K:=\{(\ell,\ell')\neq 0: |\ell|,|\ell'|<K\}$ and $\mathcal{G}^{\ell,\ell'}:=\sum_{k\neq 0}|k|^2 a^\dagger_{\ell'-k}|k|^2\overline{f_{\ell,k}}f_{\ell',k}  a_{\ell'-k}^\dagger  a_{\ell}^\dagger a_0^\dagger a_{\ell'} a_0 a_{\ell-k}$, where $f_{\ell,k}$ is defined in Eq.~(\ref{Eq:Definition_f}). Note that $\left||k|^2\overline{f_{\ell,k}}f_{\ell',k}\right|\lesssim N^{-2}$, see for example the proof of Lemma \ref{Lem:Delta_Control}. Furthermore we have for $(\ell,\ell')\neq (0,0)$ the estimate $\left\| \left(a_{\ell'}^\dagger a_0^\dagger a_\ell a_0\right)\big|_{\mathcal{F}^+_M} \right\|\lesssim \sqrt{M+1}(N+1)^{\frac{3}{2}}$. Consequently $\pm \left(\mathcal{G}^{\ell,\ell'}+\mathrm{H.c.}\right)\big|_{\mathcal{F}^+_M}\lesssim \sqrt{\frac{M+1}{N}}\mathcal{N}$, which concludes the proof, since the set $A_K$ is finite and $\pm \sum_{k\neq 0}|k|^2 w_k^2 N^{-2}\left(a_0^{2\dagger }a_0^2-N^2 \right) a_k^\dagger a_k\lesssim \frac{M+1}{N}\mathcal{N}$.
\end{proof}

With Lemma \ref{Lem:GP_Error_Estimate} at hand, we are now in a position to verify Theorem \ref{Th:GP}. In the following let $\Psi_N$ be the ground state of the operator $H_N$, $f,g:\mathbb{R}\longrightarrow [0,1]$ smooth functions satisfying $f^2+g^2=1$ and $f(x)=1$ for $x\leq \frac{1}{2}$ as well as $f(x)=0$ for $x\geq 1$. Then we define for $0<\rho<1$ the truncated states $\Theta_N:=\frac{f \!  \left(\frac{\mathcal{N}}{\rho N}\right) \! \Psi_N}{\|f \!  \left(\frac{\mathcal{N}}{\rho N}\right) \! \Psi_N\|}$. Note that the ground state $\Psi_N$ satisfies complete Bose-Einstein condensation $\frac{1}{N}\braket{\Psi_{N},\mathcal{N}\Psi_{N}}\underset{N\rightarrow \infty}{\longrightarrow}0$ according to the well known results in \cite{LSe}. Consequently we obtain by the IMS inequality, which can for example be found in \cite[Proposition 20]{HST}, respectively which follows from the methods presented in our more general Lemma \ref{Lem:IMS}, 
\begin{align*}
    \braket{\Theta_N,H_N\Theta_N}\leq E_{N}+\frac{N}{\rho N-2\braket{\Psi_{N},\mathcal{N}\Psi_{N}}}\frac{C'}{N}\leq E_{N}+\frac{C}{N}
\end{align*}
for suitable $C',C$. Since $\Theta_N$ is an element of $\mathcal{F}^+_{\rho N}$, we have by Lemma \ref{Lem:GP_Error_Estimate} 
\begin{align*}
   \left\langle \Theta_N,\left(\epsilon \mathcal{N}-\sum_{i=1}^4 \mathcal{E}_i\right)\Theta_N\right\rangle \geq \left(\epsilon-4C_{K,\frac{1}{2}}\sqrt{\rho+\frac{1}{N}}\right)\braket{\Theta_N,\mathcal{N}\Theta_N}-4C_{K,\frac{1}{2}}\sqrt{\rho+\frac{1}{N}}
\end{align*}
for $\epsilon>0$, $K<\infty$ and $0<\rho<\frac{1}{2}$. Assuming $\rho< \left(\frac{\epsilon}{4C_{K,\frac{1}{2}}}\right)^2$, as well as $N$ large enough, $\epsilon<1$ and $K\geq K_0$ large enough such that Eq.~(\ref{Eq:Main}) holds, we obtain the lower bound
\begin{align*}
    E_N-  4\pi \mathfrak{a}_N (N  -  1)\geq  \frac{1}{2}\sum_{k\neq 0}\left\{\sqrt{A_k^2-B_k^2}-A_k+C_k\right\}-\frac{C}{N}-4C_{K,\frac{1}{2}}\sqrt{\rho+\frac{1}{N}}.
\end{align*}
Since $|\sqrt{A_k^2-B_k^2}-A_k+C_k|\lesssim \frac{1}{|k|^4}$ by Lemma \ref{Lem:Coefficient_Comparison}, we conclude using dominated convergence
\begin{align*}
  \lim_{\epsilon\rightarrow 0,K\rightarrow \infty} & \lim_{\rho\rightarrow 0}\lim_{N\rightarrow \infty}\left[\sum_{k\neq 0}\left\{\sqrt{A_k^2-B_k^2}-A_k+C_k\right\}-\frac{C}{N}-4C_{K,\frac{1}{2}}\sqrt{\rho+\frac{1}{N}}\right]\\
   & \ =\sum_{k\neq 0}\left\{\sqrt{|k|^4+16\pi \mathfrak{a} |k|^2}-|k|^2-8\pi \mathfrak{a}+\frac{(8\pi \mathfrak{a})^2}{2|k|^2}\right\}. 
\end{align*}

\section{Beyond the Gross-Pitaevskii Regime}
\label{Sec:Beyond_GP}
In this section we will explain how to treat the case $\kappa>0$. For this purpose we use the following algebraic version of Bogoliubov's Lemma, see for example \cite{LS,S}, 
\begin{align*}
   &  \frac{1}{2}\sum_{k\neq 0} \left\{A_k\left(d_k^\dagger d_k+d_{-k}^\dagger d_{-k}\right)+B_k\left(d_k d_{-k}+d_{-k}^\dagger d_k\right)+C_k \frac{[d_k,d_k^\dagger]+[d_{-k},d_{-k}^\dagger]}{2}\right\}\\
   \nonumber
    & = \sum_{k\neq 0}e_k\left(\gamma_k d_k \! + \! \nu_k d_{-k}^\dagger\right)^\dagger \left(\gamma_k d_k \! + \! \nu_k d_{-k}^\dagger\right) \! + \! \frac{1}{2}\sum_{k\neq 0}\left\{\sqrt{A_k^2 \! - \! B_k^2} \! - \! A_k+C_k\right\}\frac{[d_k,d_k^\dagger] \! + \! [d_{-k},d_{-k}^\dagger]}{2}
\end{align*}
with $e_k:=\sqrt{(A_k)^2-(B_k)^2}$, $\gamma_k:=\frac{1}{\sqrt{1-\alpha_k^2}}$ and $\nu_k:=\frac{\alpha_k}{\sqrt{1-\alpha_k^2}}$ where we define the coefficients $\alpha_k:=\frac{1}{B_k}\left(A_k-\sqrt{A_k^2-B_k^2}\right)$. Note that by Lemma \ref{Lem:Coefficient_Comparison}, it is easy to see that $0\leq \nu_k\leq \gamma_k \lesssim \frac{N^{\frac{\kappa}{4}}}{\sqrt{|k|}}$ for $|k|\lesssim N^{\frac{\kappa}{2}}$ and $|\nu_k|\lesssim \frac{N^\kappa}{|k|^2}$ globally. Setting $\epsilon:=0$ and using the estimate in Eq.~(\ref{Eq:Pseudo_Quadratic}) together with the identity in Eq.~(\ref{Eq:Writen_In_d}), we therefore obtain
\begin{align}
\nonumber
     & H_{N,\kappa}  -4\pi \mathfrak{a}_{N^{1-\kappa}}N^\kappa (N-1)- \frac{1}{2}\sum_{k\neq 0}\left\{\sqrt{A_k^2-B_k^2}-A_k+C_k\right\}\\
     \label{Eq:Beyond_GP_Algebraic_Lower_Bound}
    &\ \ \ \ \ \  \geq \sum_{k\neq 0}e_k\left(\gamma_k d_k+\nu_k d_{-k}^\dagger\right)^\dagger \left(\gamma_k d_k+\nu_k d_{-k}^\dagger\right)-\sum_{i=1}^4 \mathcal{E}_i.
\end{align}
\subsection{Control of the Error Terms $\mathcal{E}_i$}
\label{Subsection:Error_Control}
In order to obtain a useful representation of the error terms $\mathcal{E}_1$ and $\mathcal{E}_3$, recall the definition $f_{\ell,k}$ from Eq.~(\ref{Eq:Definition_f}) and let us introduce $\Lambda_{k \ell, k' \ell'}:  =\delta_{k+\ell=k'+\ell'}|k-\ell'|^2 N\overline{f_{\ell',\ell'-k}}f_{\ell,\ell-k'}$ as well as $\Upsilon^{(1)}_{\ell,k}:  =2 N^{\frac{1}{2}} |k|^2\overline{f_{\ell,-k}}w_k$, $\Upsilon^{(2)}_{\ell,k}:  = -N^{\frac{1}{2}}|k|^2 \overline{w_k} f_{\ell,-k}$ and $\Upsilon^{(3)}_{\ell,k}:  = -N^{\frac{1}{2}}\left(\mathds{1}(|k|<K)\frac{4\sigma_k}{N}-\frac{2\sigma_0}{N}\right)f_{\ell+k,k}$. Furthermore we define $O_1:=N^{-\frac{3}{2}}a_0^\dagger a_0^2$ as well as $O_2:=O_3:=\frac{1}{\sqrt{a_0 a_0^\dagger}}\frac{a_0}{\sqrt{N}} \sqrt{a_0 a_0^\dagger}$. Using the decomposition $\delta_k=-\delta_k'+\delta_k''$ from the proof of Lemma \ref{Lem:Delta_Control}, we can then write
\begin{align}
\label{Eq:Representation_Of_Error_I}
 &  \ \ \ \  \ \ \mathcal{E}_1=\sum_{k \ell, k' \ell'}\Lambda_{k \ell, k' \ell'}\, a_{\ell'}^\dagger a_{k'}^\dagger  \frac{a_0^\dagger a_0}{N}  a_k a_\ell+\left(\sum_{\ell,k}\Upsilon^{(1)}_{\ell,k} O_1\, a_k^\dagger a_\ell^\dagger a_{k+\ell}+\mathrm{H.c.}\right)\\
\nonumber
& \ \ \ \  \ \ \ \  \  \ \ \ \  \ \ +\sum_{k\neq 0}|k|^2 w_k^2 N^{-2}\left(a_0^{2\dagger }a_0^2-N^2\right) a_k^\dagger a_k,\\
\label{Eq:Representation_Of_Error_II}
  & \ \ \ \  \ \ \mathcal{E}_3 =\left(\sum_{\ell,k}\Upsilon^{(2)}_{\ell,k}O_2\, a_k^\dagger a_\ell^\dagger a_{k+\ell}+\mathrm{H.c.}\right)+\left(\sum_{\ell,k}\Upsilon^{(3)}_{\ell,k}O_3\, a_k^\dagger a_\ell^\dagger a_{k+\ell}+\mathrm{H.c.}\right)\\
\nonumber
&  \  \ \  \  \  \  \  \ +\sum_k |k|^2 w_k^2\big([\delta_k,d_k^\dagger]+[\delta_k,d_k^\dagger]+[\delta_k,\delta_k^\dagger]\big)+F_1+F_1^\dagger + F_2,
\end{align}
where we have defined $F_1:=-\sum_{k\neq 0}a_k^\dagger \frac{1}{\sqrt{a_0 a_0^\dagger}}a_0\left(N\left(\mathds{1}(|k|<K)\frac{4\sigma_k}{N}-\frac{2\sigma_0}{N}\right)\delta''_k+|k|^2 w_k(\delta''_{-k})^\dagger\right)$ and $F_2:=\sum_k |k|^2 w_k \left( \delta_k^\dagger \delta_{-k}^\dagger+\delta_{-k}\delta_k\right)$.

This Subsection is devoted to the derivation of suitable bounds on the most prominent error contributions $\sum_{k \ell, k' \ell'}\Lambda_{k \ell, k' \ell'}\, a_{\ell'}^\dagger a_{k'}^\dagger  \frac{a_0^\dagger a_0}{N}  a_k a_\ell$ and $\sum_{\ell,k}\Upsilon^{(i)}_{\ell,k} O_i\, a_k^\dagger a_\ell^\dagger a_{k+\ell}$ in Lemma \ref{Lem:Lambda_Control} and Lemma \ref{Lem:Upsilon_Control}. The residual error terms are then taken care of in Appendix \ref{Appendix:Additional_Error_Estimates}. Let us at this point introduce the new operators $b_k:=\gamma_k a_0^\dagger \frac{1}{\sqrt{a_0 a_0^\dagger}}a_k+\nu_k a_{-k}^\dagger \frac{1}{\sqrt{a_0 a_0^\dagger}} a_0$ and the particle number operator in these new variables $\widetilde{\mathcal{N}}:=\sum_{k\neq 0}b_k^\dagger b_k$, which notably satisfy the CCR on $\mathcal{F}^+_{N-1}$, i.e. $\left[b_k,b_\ell^\dagger\right]\Psi=\delta_{k,\ell}\Psi$ for $\Psi\in \mathcal{F}^+_{N-1}$. Let us furthermore introduce $\mathcal{F}^{\leq}_{M_0}$ as the spectral subspace $\sum_{0<|k|<K}a_k^\dagger a_k\leq M_0$.

\begin{lem}
\label{Lem:Lambda_Control}
    There exists a constant $C>0$ such that we have for $0<\delta<\frac{1}{5}$
    \begin{align}
            \label{Eq:Lambda_Control_0}  
        & \pm \!   \!  \!  \sum_{k \ell, k' \ell'} \!  \!  \! \Lambda_{k \ell, k' \ell'}\, a_{\ell'}^\dagger a_{k'}^\dagger  \frac{a_0^\dagger a_0}{N} a_k a_\ell\bigg|_{\mathcal{F}^{\leq}_{M_0}} \!  \!   \!   \!   \! \leq \!   C\! \!  \left( \!  \! N^{\frac{5\kappa}{2}+3\delta}\frac{M_0 \!  + \!  1}{N}+N^{\frac{11\kappa}{2}+3\delta-1} \!  \! \right)  \! \! \left(  \!  \!  \sum_k |k|^{1-2\delta} b_k^\dagger b_k\Big|_{\mathcal{F}^{\leq}_{M_0}}  \! \!  \!  \! + \! 1  \!  \! \right)\!  \! ,\\
                  \label{Eq:Lambda_Control_pre_0}  
      & \ \  \ \  \ \ \pm   \! \!  \sum_{k \ell, k' \ell'}\Lambda_{k \ell, k' \ell'}\, a_{\ell'}^\dagger a_{k'}^\dagger  \frac{a_0^\dagger a_0}{N} a_k a_\ell       \! \leq  \! C N^{\frac{11\kappa}{2}+3\delta-1} \! \left(\widetilde{\mathcal{N}} \! + \! 1\right) \!  \! \left(\sum_k |k|^{1-2\delta} b_k^\dagger b_k \! + \! 1\right)  \!  \! .
    \end{align}
\end{lem}
\begin{proof}
Using $ a_k=\frac{1}{\sqrt{a_0 a_0^\dagger}}a_0\gamma_k b_k+\frac{1}{\sqrt{a_0 a_0^\dagger}}a_0\nu_{-k}b_{-k}^\dagger$ allows us to rewrite
\begin{align}
\label{Eq:First_Lambda_Indentification}
\sum_{k \ell, k' \ell'}\Lambda_{k \ell, k' \ell'}\, a_{\ell'}^\dagger a_{k'}^\dagger \frac{a_0^\dagger a_0}{N} a_k a_\ell=\sum_{k \ell, k' \ell'}\Lambda_{k \ell, k' \ell'}\, a_{\ell'}^\dagger \left(\gamma_{k'} b_{k'}+\nu_{-k'}b_{-k'}^\dagger\right)^\dagger O \left(\gamma_k b_k+\nu_{-k}b_{-k}^\dagger\right) a_\ell
\end{align}
with $O:=a_0^\dagger \frac{1}{\sqrt{a_0 a_0^\dagger}}\frac{a_0^\dagger a_0}{N}\frac{1}{\sqrt{a_0 a_0^\dagger}}a_0$. In order to establish a good upper bound on the matrix $\Lambda$, let us define the auxiliary matrix $\Lambda^{(\delta)}_{k \ell, k' \ell'}:=\Lambda_{k \ell,k' \ell'}|k|^{\delta-\frac{1}{2}}|k'|^{\delta-\frac{1}{2}}$. Making use of 
 \begin{align*}
        |\ell'-k|^2   & | f_{\ell,\ell-k'}| \! \leq  \! |(V_{N^{1-\kappa}} \! - \! V_{N^{1-\kappa}}RV_{N^{1-\kappa}})_{k'(\ell-k'),\ell  0} \! + \! (V_{N^{1-\kappa}} \! - \! V_{N^{1-\kappa}}RV_{N^{1-\kappa}})_{k'(\ell-k'),0 \ell}| \! \lesssim  \! N^{\kappa-1}
    \end{align*}
by Lemma \ref{Lem:Coefficient_Comparison}, we obtain that $\Lambda^{(\delta)}$ satisfies the weighted Schur test
   \begin{align}
    \label{Eq:Schur_Test}
        \sum_k p(k)\left|\Lambda^{(\delta)}_{k (k'+\ell'-k),k' \ell'}\right|\leq C p(k')N^{\kappa} \sup_{\ell'}\sum_k |k|^{2\delta-1}\left|f_{\ell',\ell'-k}\right|,
    \end{align}
  with the weight function $p(k):=|k|^{\delta-\frac{1}{2}}$ for a suitable constant $C$ and all $\ell'$ and $k'$. Defining $h_{\ell'}(k):=|k|^{3\delta+2}|f_{\ell',\ell'-k}|$, we obtain as a consequence of the weighted Schur Test in Eq.~(\ref{Eq:Schur_Test}) 
\begin{align*}
   &  \ \ \ \  \|\Lambda^{(\delta)}\|  \! \lesssim  \! N^{\kappa}\sup_{\ell'}\!\!\sum_k|k|^{2\delta-1}\!\!\left|f_{\ell',\ell'-k}\right| \! =  \! N^{\kappa}\sup_{\ell'} \! \sum_k   h_{\ell'}(k) |k|^{-(3+\delta)}\\
   &\lesssim N^{\kappa}\sup_{\ell'} \! \left(\sum_k   h_{\ell'}(k)^{\frac{1+\delta}{3\delta}} |k|^{-(3+\delta)}\right)^{\frac{3\delta}{1+\delta}}=N^{2\kappa-1}\sup_{\ell'}\! \left(\sum_k  |k|^{-2}\left(N^{1-\kappa}|k|^2 f_{\ell',\ell'-k}\right)^{\frac{1+\delta}{3\delta}}\right)^{\frac{3\delta}{1+\delta}},
\end{align*}
where we have used that there is an embedding of $L^{1}$ in $L^{\frac{1+\delta}{3\delta}}$ for the finite measure with discrete density $|k|^{-(3+\delta)}$. By Lemma \ref{Lem:Coefficient_Comparison} we know that $N^{1-\kappa}|k|^2 f_{\ell',\ell'-k}\lesssim 1$, which yields together with the assumption $\delta<\frac{1}{5}$, or equivalently $\frac{1+\delta}{3\delta}>2$, the estimate
\begin{align*}
& \sum_k  |k|^{-2}\left(N^{-\kappa}|k|^2 f_{\ell',\ell'-k}\right)^{\frac{1+\delta}{3\delta}} \! \lesssim \!  N^{2-2\kappa}\sum_k  |k|^{2}(f_{\ell',\ell'-k})^2 \! \lesssim N^{1-\kappa}\leq N,
\end{align*}
where we have used Lemma \ref{Lem:Increased_Decay}. Therefore $\|\Lambda^{(\delta)}\|\lesssim N^{2\kappa+3\delta-1}$, or equivalently $\pm \Lambda\lesssim N^{2\kappa+3\delta-1}(-\Delta_x)^{\frac{1}{2}-\delta}$. Consequently we obtain
\begin{align*}
\pm &  \! \sum_{k \ell, k' \ell'}\Lambda_{k \ell, k' \ell'}\, a_{\ell'}^\dagger a_{k'}^\dagger \frac{a_0^\dagger a_0}{N} a_k a_\ell \! \lesssim \!  N^{2\kappa+3\delta-1} \!   \! \!  \!  \! \sum_{ 0<|\ell|< K}\sum_{k\neq 0}|k|^{1-2\delta} a_{\ell}^\dagger \left(\gamma_{k} b_{k} \! + \! \nu_{-k}b_{-k}^\dagger\right)^\dagger  \!  \!  \left(\gamma_k b_k \! + \! \nu_{-k}b_{-k}^\dagger\right) \!  a_\ell\\
& \ \ \lesssim  N^{2\kappa+3\delta-1}\sum_{0<|\ell|<  K}\sum_{k\neq 0}|k|^{1-2\delta}(\gamma_k^2+\nu_k^2) a_{\ell}^\dagger  b_{k}^\dagger  b_k a_\ell+N^{2\kappa+3\delta-1}\left(\sum_{k\neq 0} |k|^{1-2\delta}\nu_k^2\right)\sum_{0<|\ell|<  K} a_{\ell}^\dagger  a_\ell
\end{align*}
by Eq.~(\ref{Eq:First_Lambda_Indentification}), where we have used $\|O\|\lesssim 1$ and $[b_k,b_k^\dagger]\leq 1$. Note that
\begin{align*}
N^{2\kappa+3\delta-1}\left(\sum_{k\neq 0} |k|^{1-2\delta}\nu_k^2\right)\sum_{0<|\ell|<  K} a_{\ell}^\dagger  a_\ell\lesssim N^{4\kappa+3\delta-1}\sum_{0<|\ell|<  K} a_{\ell}^\dagger  a_\ell\lesssim N^{\frac{11\kappa}{2}+3\delta-1}\! \left(\widetilde{\mathcal{N}}+1\right).
\end{align*}
Furthermore $\gamma_k^2+\nu_k^2\lesssim N^{\frac{\kappa}{2}}$ and $[ b_k , a_\ell]=-\delta_{k,-\ell}\frac{1}{\sqrt{a_0 a_0^\dagger }}a_0\nu_{k}$, and therefore
\begin{align*}
& \sum_{0<|\ell|<  K}\sum_{k\neq 0}|k|^{1-2\delta}(\gamma_k^2+\nu_k^2) a_{\ell}^\dagger  b_{k}^\dagger  b_k a_\ell \! \lesssim  \! N^{\frac{\kappa}{2}}\sum_{k\neq 0}|k|^{1-2\delta} b_k^\dagger \left(\sum_{0<|\ell|<  K}a_\ell^\dagger a_\ell  \! \right)  \! b_k  \! + \! N^{\frac{\kappa}{2}}\sum_{k\neq 0}|k|^{1-2\delta}\nu_k^2.
\end{align*}
Making use of the fact that $N^{-1}b_k^\dagger \left(\sum_{0<|\ell|<  K}a_\ell^\dagger a_\ell  \! \right)  \! b_k\bigg|_{\mathcal{F}^{\leq}_{M_0}}\leq \frac{M_0+1}{N} b_k^\dagger b_k$ concludes the proof of Eq.~(\ref{Eq:Lambda_Control_0}) and $b_k^\dagger \left(\sum_{0<|\ell|<  K}a_\ell^\dagger a_\ell  \! \right)  \! b_k\leq N^{\frac{3\kappa}{2}}\widetilde{\mathcal{N}} b_k^\dagger b_k$, see Lemma \ref{Lem:True_Particle_Number_Comparison}, yields Eq.~(\ref{Eq:Lambda_Control_pre_0}).
\end{proof}

\begin{lem}
\label{Lem:Upsilon_Control}
 For $r<1$ and $\kappa\leq 1$, there exists a $C>0$, such that for $i\in \{1,2,3\}$
    \begin{align}
    \label{Eq:Upsilon_Control_I}
        \pm  \! \left(\sum_{\ell,k}\Upsilon^{(i)}_{\ell,k}O_i\, a_k^\dagger a_\ell^\dagger a_{k+\ell}  \! +  \! \mathrm{H.c.}\right)\bigg|_{\mathcal{F}^{\leq}_{M_0}\cap \mathcal{F}^+_{rN}}  \!   \!   \! \leq   \! C N^{\frac{5\kappa}{2}}\sqrt{\frac{M_0  \! +  \! 1}{N}  \! +  \! N^{\frac{3\kappa}{2}-1}}\left(\widetilde{\mathcal{N}}\Big|_{\mathcal{F}^{\leq}_{M_0}\cap \mathcal{F}^+_{rN}}  \!   \! +  \! 1  \! \right)  \! .
    \end{align} 
    Furthermore, $\pm  \! \left(\sum_{\ell,k}\Upsilon^{(i)}_{\ell,k}O_i\, a_k^\dagger a_\ell^\dagger a_{k+\ell}  \! +  \! \mathrm{H.c.}\right)\bigg|_{\mathcal{F}^+_{rN}}\lesssim N^{\frac{13\kappa}{4}-\frac{1}{2}}\! \left(\widetilde{\mathcal{N}}+1\right)^2\bigg|_{\mathcal{F}^+_{rN}}$.
\end{lem}
\begin{proof}
Using $ a_k=\frac{1}{\sqrt{a_0 a_0^\dagger}}a_0\gamma_k b_k+\frac{1}{\sqrt{a_0 a_0^\dagger}}a_0\nu_{-k}b_{-k}^\dagger$, we can rewrite
\begin{align*}
    \left(\sum_{\ell,k}\Upsilon^{(i)}_{\ell,k}O_i\, a_k^\dagger a_\ell^\dagger a_{k+\ell} \! + \! \mathrm{H.c.}\right) \! = \! \left(\sum_{\ell,k}X^{(i)}_{\ell,k}\widetilde{O}_i b_k^\dagger a_\ell^\dagger b_{k+\ell} \! + \! \mathrm{H.c.}\right) \! + \! \left(\sum_{\ell,k}Y^{(i)}_{\ell,k}\widetilde{O}_i a_\ell^\dagger b_{k+\ell}b_{-k} \!  + \! \mathrm{H.c.}\right),
\end{align*}
where we define $X^{(i)}_{\ell,k}:=\gamma_k \gamma_{\ell+k}\Upsilon^{(i)}_{\ell,k}+\nu_{k}\nu_{k+\ell}\Upsilon^{(i)}_{\ell,-(k+\ell)}$, $Y^{(i)}_{\ell,k}:=\big(\nu_k \gamma_{k+\ell}+\nu_{k+\ell}\gamma_k\big)\Upsilon^{(i)}_{\ell,k}$ and $\widetilde{O}_i:=O_i$ in the case $i\in \{1,2\}$, and $X^{(3)}:=\nu_{-k}\gamma_{k+\ell}\overline{\Upsilon^{(3)}_{-k,k+\ell}}+\nu_{k+\ell}\gamma_{-k}\overline{\Upsilon^{(3)}_{k+\ell,-k}} $, $Y^{(3)}_{\ell,k}:=\big(\gamma_{-k}\gamma_{k+\ell}+\nu_{-k}\nu_{k+\ell}\big)\overline{\Upsilon^{(3)}_{-k,\ell+k}}$ and $\widetilde{O}_3:=O_3^\dagger \left(\frac{1}{\sqrt{a_0 a_0^\dagger}}a_0\right)^2$. Using $\left(T-1\right)_{(\ell-k)k, \ell 0}=\frac{1}{|k|^2+|\ell-k|^2}\Big(V_{N^{1-\kappa}}-V_{N^{1-\kappa}} R V_{N^{1-\kappa}}\Big)_{(\ell-k)k, \ell 0}$, the bounds from Lemma \ref{Lem:Coefficient_Comparison} and the simple observation that $|\gamma_k|,|\nu_k|\lesssim N^{\frac{\kappa}{4}}$, yields $|X_{\ell,k}|,|Y_{\ell,k}|\lesssim N^{\frac{5\kappa-1}{2}}\frac{1}{|\ell|^2+|k|^2}$. Consequently
\begin{align}
\label{Eq:Upper_Bound_X}
    \left(\sum_{\ell,k}X^{(i)}_{\ell,k}\widetilde{O}_i b_k^\dagger a_\ell^\dagger b_{k+\ell}+\mathrm{H.c.}\right)\leq \epsilon \sum_{k,\ell}\frac{b_{k+\ell}^\dagger b_{k+\ell}}{(|\ell|^2+|k|^2)^2}+\epsilon^{-1}N^{5\kappa-1}\sum_{k} \widetilde{O}_i b_k^\dagger \left(\sum_{0<\ell<K} a_\ell^\dagger a_\ell\right) b_k\widetilde{O}_i^\dagger,
\end{align}
where we have used that $X_{\ell,k}=0$ in case $|\ell|\geq K$. Note at this point that $[a_\ell b_k,\widetilde{O}_i^\dagger]=\gamma_k  a_\ell a_k\big[a_0^\dagger \frac{1}{\sqrt{a_0 a_0^\dagger}},\widetilde{O}_i^\dagger\big]+\nu_k a_\ell a_{-k}^\dagger \big[\frac{1}{\sqrt{a_0 a_0^\dagger}}a_0,\widetilde{O}_i^\dagger\big]$ and $\big[a_0^\dagger \frac{1}{\sqrt{a_0 a_0^\dagger}},\widetilde{O}_i^\dagger\big]^\dagger \big[a_0^\dagger \frac{1}{\sqrt{a_0 a_0^\dagger}},\widetilde{O}_i^\dagger\big]\Big|_{\mathcal{F}^+_{rN}}\lesssim \frac{1}{N^2}$ as well as $\big[\frac{1}{\sqrt{a_0 a_0^\dagger}}a_0,\widetilde{O}_i^\dagger\big]^\dagger \big[\frac{1}{\sqrt{a_0 a_0^\dagger}}a_0,\widetilde{O}_i^\dagger\big]\Big|_{\mathcal{F}^+_{rN}}\lesssim \frac{1}{N}^2$ and $\widetilde{O}_i\widetilde{O}_i^\dagger \Big|_{\mathcal{F}^+_{rN}}\lesssim 1$. Consequently
\begin{align}
\nonumber
    & \sum_{k}\! \!  \widetilde{O}_i b_k^\dagger \! \left(\! \sum_{0<\ell<K} a_\ell^\dagger a_\ell\! \right)\!  b_k\widetilde{O}_i^\dagger\Big|_{\mathcal{F}^+_{rN}}\! \! \! \lesssim \! \sum_{k} \!  b_k^\dagger \! \left(\! \sum_{0<\ell<K} a_\ell^\dagger a_\ell\! \right) \! b_k\Big|_{\mathcal{F}^+_{rN}}\! \! +\! \frac{1}{N^2}\! \sum_{k\neq 0}\! \gamma_k^2 a_k^\dagger\! \left(\! \sum_{0<\ell<K} a_\ell^\dagger a_\ell\! \right)\!  a_k\Big|_{\mathcal{F}^+_{rN}}\\
    \label{Eq:O_Operator_Control}
    &  + \! \frac{1}{N^2} \! \sum_{k\neq 0}|\nu_k|^2  a_k\! \left(\sum_{0<\ell<K} a_\ell^\dagger a_\ell\right)\!  a_k^\dagger\Big|_{\mathcal{F}^+_{rN}} \! \! \lesssim \! \sum_{k} b_k^\dagger\!  \left(\sum_{0<\ell<K} a_\ell^\dagger a_\ell\right) \! b_k\Big|_{\mathcal{F}^+_{rN}}\! \! +\! N^{2\kappa-1}\! \left(\widetilde{\mathcal{N}}+1\right)\! \Big|_{\mathcal{F}^+_{rN}}\! \! \! ,
\end{align}
where we have used $\sum_{0<\ell<K} a_\ell^\dagger a_\ell\leq N$ and applied Lemma \ref{Lem:True_Particle_Number_Comparison} in the last estimate. Using $\sum_{k,\ell}\frac{1}{(|\ell|^2+|k|^2)^2}b_{k+\ell}^\dagger b_{k+\ell}\lesssim \widetilde{\mathcal{N}}$ and that $b_k$ and $b_k^\dagger$ map $\mathcal{F}^{\leq}_{M_0}$ into $\mathcal{F}^{\leq}_{M_0+1}$, we have by Eq.~(\ref{Eq:Upper_Bound_X}) 
\begin{align*}
    & \left(\sum_{\ell,k}X^{(i)}_{\ell,k} \widetilde{O}_i b_k^\dagger a_\ell^\dagger b_{k+\ell} \! + \! \mathrm{H.c.}\right) \!  \! \bigg|_{\mathcal{F}^{\leq}_{M_0}\cap \mathcal{F}^+_{rN}} \!  \!  \!  \! \lesssim  \!   \! \left(\epsilon  \! + \! \epsilon^{-1}N^{5\kappa}\frac{M_0+1}{N} \! + \! \epsilon^{-1}N^{5\kappa}N^{2\kappa-2}\right) \!  \!  \left(\widetilde{\mathcal{N}} \! + \! 1\right)\bigg|_{\mathcal{F}^{\leq}_{M_0}\cap \mathcal{F}^+_{rN}}\\
    &\ \ \ \ \ \ \ \ \ \ \ \ \ \ \lesssim N^{\frac{5\kappa}{2}}\sqrt{\frac{M_0+1}{N}+N^{2\kappa-2}}\left(\widetilde{\mathcal{N}}\bigg|_{\mathcal{F}^{\leq}_{M_0}\cap \mathcal{F}^+_{rN}}+1\right)
\end{align*}
for an optimal choice of $\epsilon$. Regarding the $Y$-contributions, we are going to estimate
\begin{align}
\label{Eq:Split_Y_contribution}
    \left(\sum_{\ell,k}Y^{(i)}_{\ell,k} \widetilde{O}_i a_\ell^\dagger b_{k+\ell}b_{-k} +\mathrm{H.c.}\right)\lesssim \epsilon\, \widetilde{\mathcal{N}}+\epsilon^{-1}\sum_{k,\ell,\ell'}\overline{Y_{\ell',k}} Y_{\ell,k}\, \widetilde{O}_i a^\dagger_\ell b_{\ell+k}b_{\ell'+k}^\dagger a_{\ell'}\widetilde{O}_i^\dagger.
\end{align}
Defining the operator $G$ via its coefficients $G_{\ell k, \ell' k'}:=\delta_{\ell+k=\ell'+k'}\overline{Y_{\ell',k-\ell'}} Y_{\ell,k-\ell}$ and using the fact that $b_k$ satisfies the CCR on $\mathcal{F}^+_{N-1}$ as well as $a_\ell \mathcal{F}^+_{N-1}\subset \mathcal{F}^+_{N-2}$, we furthermore obtain
\begin{align*}
    & \sum_{k,\ell,\ell'} \! \overline{Y_{\ell',k}} Y_{\ell,k}\,\widetilde{O}_i a^\dagger_\ell b_{\ell+k}b_{\ell'+k}^\dagger a_{\ell'}\widetilde{O}_i^\dagger\Big|_{\mathcal{F}^+_{N-2}} \!  \!  \!  \!  \!  \! \!  =  \! \sum_{k,\ell} \! \overline{Y_{\ell,k}} Y_{\ell,k}\,a^\dagger_\ell   \widetilde{O}_i \widetilde{O}_i^\dagger a_{\ell}\Big|_{\mathcal{F}^+_{N-2}} \!  \!  \!  \!  \!  \!  + \!  \sum_{k\ell,k\ell'} \! G_{\ell k, \ell' k'}\widetilde{O}_i a_\ell^\dagger b_k^\dagger b_{k'}a_{\ell'}\widetilde{O}_i^\dagger\Big|_{\mathcal{F}^+_{N-2}}.
\end{align*}
Note that $\sum_{k,\ell}\overline{Y_{\ell,k}} Y_{\ell,k}\,a^\dagger_\ell   \widetilde{O}_i \widetilde{O}_i^\dagger a_{\ell}\lesssim \sum_{k,\ell}\overline{Y_{\ell,k}} Y_{\ell,k}\,a^\dagger_\ell   a_{\ell}\lesssim N^{5\kappa-1}\mathcal{N}\lesssim N^{5\kappa}N^{\frac{3\kappa}{2}-1}\left(\widetilde{\mathcal{N}}+1\right)$ by Lemma \ref{Lem:True_Particle_Number_Comparison}. Using furthermore the bound on the operator norm $\|G\|\lesssim N^{5\kappa-1}$ yields
\begin{align*}
   & \ \ \pm \sum_{k\ell,k\ell'}G_{\ell k, \ell' k'} \widetilde{O}_i a_\ell^\dagger b_k^\dagger b_{k'}a_{\ell'}\widetilde{O}_i^\dagger\lesssim N^{5\kappa-1}\sum_{0<|\ell|<K,k}\widetilde{O}_i a_\ell^\dagger b_k^\dagger b_k a_\ell\widetilde{O}_i^\dagger\\
   &  \lesssim N^{5\kappa-1}\sum_k \widetilde{O}_i b_k^\dagger \left(\sum_{0<|\ell|<K}a_\ell^\dagger a_\ell\right)b_k\widetilde{O}_i^\dagger +N^{5\kappa-1}\sum_k \widetilde{O}_i [b_k,a_{-k}]^\dagger [b_k,a_{-k}]\widetilde{O}_i^\dagger.
\end{align*}
Note that $\sum_k \widetilde{O}_i [b_k,a_{-k}]^\dagger [b_k,a_{-k}]\widetilde{O}_i^\dagger\Big|_{\mathcal{F}^{+}_{rN}} \lesssim \sum_k |\nu_k|^2\widetilde{O}_i \widetilde{O}_i^\dagger\Big|_{\mathcal{F}^{+}_{rN}}\lesssim \sum_k |\nu_k|^2\lesssim N^{\frac{3\kappa}{2}}$, and hence
\begin{align}
\label{Eq:UsefulEquation_Upsilon}
    \sum_{k,\ell,\ell'}\overline{Y_{\ell',k}} Y_{\ell,k}\, \widetilde{O}_i a^\dagger_\ell b_{\ell+k}b_{\ell'+k}^\dagger a_{\ell'}\widetilde{O}_i^\dagger \! \lesssim  \! N^{5\kappa-1}\sum_k \widetilde{O}_i b_k^\dagger  \! \left(\sum_{0<|\ell|<K}a_\ell^\dagger a_\ell\right) \! b_k\widetilde{O}_i^\dagger \! + \! N^{5\kappa-1}N^{\frac{3\kappa}{2}} \! \left(\widetilde{\mathcal{N}} \! + \! 1 \! \right) \! .
\end{align}
Therefore we obtain in combination with Eq.~(\ref{Eq:O_Operator_Control})
\begin{align*}
     & \pm  \!  \! \sum_{k,\ell,\ell'}\overline{Y_{\ell',k}} Y_{\ell,k}\, \widetilde{O}_i a^\dagger_\ell b_{\ell+k}b_{\ell'+k}^\dagger a_{\ell'}\widetilde{O}_i^\dagger\Big|_{\mathcal{F}^{\leq}_{M_0}\cap \mathcal{F}^{+}_{rN}} \!  \!  \! \lesssim  \! N^{5\kappa} \! \left( \! \frac{M_0 \! + \! 1}{N} \! + \! N^{2\kappa-2} \! + \! N^{\frac{3\kappa}{2}-1}\right)\! \left(\widetilde{\mathcal{N}}\Big|_{\mathcal{F}^{\leq}_{M_0}\cap \mathcal{F}^{+}_{rN}} \!  \! + \! 1 \! \right) \! .
\end{align*}
Using the optimal choice $\epsilon:= N^{\frac{5\kappa}{2}}\sqrt{\frac{M_0+1}{N}+N^{\frac{3\kappa}{2}-1}}$ concludes the proof of Eq.~(\ref{Eq:Upsilon_Control_I}). Regarding the second statement of the Lemma, note that $\sum_{0<\ell<K} a_\ell^\dagger a_\ell\lesssim N^{\frac{3\kappa}{2}-1}\! \left(\widetilde{\mathcal{N}}+1\right)$ by Lemma \ref{Lem:True_Particle_Number_Comparison}, and therefore we obtain by Eq.~(\ref{Eq:O_Operator_Control})
\begin{align*}
   \sum_{k}\! \!  \widetilde{O}_i b_k^\dagger \! \left(\! \sum_{0<\ell<K} a_\ell^\dagger a_\ell\! \right)\!  b_k\widetilde{O}_i^\dagger\Big|_{\mathcal{F}^+_{rN}} \! \! \lesssim N^{\frac{3\kappa}{2}}\! \left(\widetilde{\mathcal{N}}+1\right)^2\Big|_{\mathcal{F}^+_{rN}} \! .
\end{align*}
In combination with Eq.~(\ref{Eq:Upper_Bound_X}), respectively Eq.~(\ref{Eq:Split_Y_contribution}) and Eq.~(\ref{Eq:UsefulEquation_Upsilon}), this concludes the proof with the concrete choice $\epsilon:=N^{\frac{13\kappa}{4}-\frac{1}{2}}$.
\end{proof}

At this point we want to emphasise that our error estimates in Lemma \ref{Lem:Lambda_Control} and Lemma \ref{Lem:Upsilon_Control} can only produce good bounds in case we apply them for states in $\Psi\in \mathcal{F}^\leq_M$ with suitably small number of particles $M$. As we will demonstrate in Appendix \ref{Appendix:A_priori_Condensation}, we can always restrict our attention to such states.

\subsection{Proof of Theorem \ref{Th:Beyond_GP}}
\label{Subsection:Proof of Theorem_Beyond}

In the following let $0<\kappa<\frac{1}{8}$ and $d\in \mathbb{N}$, and let us introduce the concrete choices $\lambda_0:=\frac{3}{8}$, $\lambda:=\frac{21}{32}$ and $K:=N^{\frac{5}{16}+\frac{\kappa}{2}}$. As we explain later in this subsection, we can assume $E^{(d)}_{N,\kappa}\leq E_{N,\kappa}+CN^{\frac{\kappa}{2}}$ without loss of generality, and therefore there exists a $d$ dimensional subspace $\mathcal{V}_d \subseteq \mathcal{F}^{\leq}_{N^{\lambda_0}}\cap \mathcal{F}^+_{N^\lambda}$, such that according to Lemma \ref{Lem:IMS}
\begin{align}
\label{Eq:Advanced_Min_Max}
    E^{(d)}_{N,\kappa}\geq \sup_{\Psi\in \mathcal{V}_d:\|\Psi\|=1} \braket{\Psi, H_{N,\kappa}\Psi}-C N^{-\tau},
\end{align}
with $\tau:=\frac{5}{2} \! \left(\frac{1}{8}-\kappa\right)$. We want to emphasise at this point that Lemma \ref{Lem:IMS} depends crucially on the a priori estimates derived in \cite{F} as we explain in Appendix \ref{Appendix:A_priori_Condensation}. In order to control the error terms $\mathcal{E}_i$, we first observe that for any $\Psi$ in $\mathcal{V}_d $ with $\|\Psi\|=1$ we obtain
\begin{align*}
  \braket{\Psi,\mathcal{E}_2 \Psi}\leq \frac{1}{N}\braket{\Psi,\sum_{0<|k|<K} \sigma_k \mathcal{N}a_k^\dagger a_k \Psi}\lesssim N^{\kappa+\lambda_0-1}\braket{\Psi,\mathcal{N}\Psi}\lesssim N^{\frac{5}{2}\kappa+\lambda_0-1}\braket{\Psi,\big(\widetilde{\mathcal{N}}+1\big)\Psi},  
\end{align*}
where we used Lemma \ref{Lem:True_Particle_Number_Comparison} and $\widetilde{\mathcal{N}}$ is defined above Lemma \ref{Lem:Lambda_Control}. Again by Lemma \ref{Lem:True_Particle_Number_Comparison}, we have 
\begin{align*}
|\braket{\Psi,\sum_{k\neq 0}|k|^2 w_k^2 N^{-2}\left(a_0^{2\dagger }a_0^2 \! - \! N^2 \right) a_k^\dagger a_k\Psi}| \! \lesssim  \! N^{2\kappa-1}\braket{\Psi,\mathcal{N} \! \left(\mathcal{N} \! + \! 1\right)\Psi} \! \lesssim  \! N^{\frac{5\kappa}{2}+\lambda-1}\! \braket{\Psi,\left(\widetilde{\mathcal{N}} \! + \! 1\right)\Psi},    
\end{align*}
which is a term appearing in the definition of $\mathcal{E}_1$ in Eq.~(\ref{Eq:Representation_Of_Error_I}). Using the representation of $\mathcal{E}_1$ and $\mathcal{E}_3$ from Eq.~(\ref{Eq:Representation_Of_Error_I}) and Eq.~(\ref{Eq:Representation_Of_Error_II}), the estimates from Lemmata \ref{Lem:Lambda_Control}, \ref{Lem:Upsilon_Control} and \ref{Lem:Universal_Error_Estimate}, as well as Eq.~(\ref{Eq:Residuum_I}) and Eq.~(\ref{Eq:Residuum_II}), we therefore obtain
\begin{align*}
    \left\langle \Psi,\sum_{i=1}^4 \mathcal{E}_i \Psi\right\rangle \! \lesssim \! N^{-\tau} \! \left\langle \!  \Psi,\left( \! \widetilde{\mathcal{N}} \! + \! \sum_{k}|k|^{1-2\delta}b_k^\dagger b_k \! + \! 1 \! \right)\Psi \! \right\rangle \! \leq  \! N^{-\tau} \! \left\langle  \! \Psi,\left( \! \sum_{k}|k|^{1-2\delta}b_k^\dagger b_k \! + \! 1\right)\Psi \! \right\rangle.
\end{align*}
Further $\left\langle  \! \Psi,\left( \! \sum_{k}|k|^{1-2\delta}b_k^\dagger b_k \! + \! 1\right)\Psi \! \right\rangle\lesssim N^{-\frac{\kappa}{2}}\left\langle  \! \Psi,\sum_{k\neq 0}e_k\left(\gamma_k d_k \! + \! \nu_k d_{-k}^\dagger\right)^\dagger \left(\gamma_k d_k \! + \! \nu_k d_{-k}^\dagger\right) \Psi \! \right\rangle + 1$, by Eq.~(\ref{Eq:Comparison_With_Proper_Variables}). Combining what we have so far with the lower bound in Eq.~(\ref{Eq:Beyond_GP_Algebraic_Lower_Bound}) yields for a suitable constant $C$
\begin{align}
\nonumber
     & \braket{\Psi, H_{N,\kappa}\Psi}\geq 4\pi \mathfrak{a}_{N^{1-\kappa}}N^\kappa (N-1) + \frac{1}{2}\sum_{k\neq 0}\left\{\sqrt{A_k^2-B_k^2}-A_k+C_k\right\}-C N^{-\tau}\\
     \label{Eq:Useful_Lower_Bound}
    & \ \ \ \ \ +\left(1-CN^{-\tau}N^{-\frac{\kappa}{2}}\right)\left\langle  \! \Psi,\sum_{k\neq 0}e_k\left(\gamma_k d_k \! + \! \nu_k d_{-k}^\dagger\right)^\dagger \left(\gamma_k d_k \! + \! \nu_k d_{-k}^\dagger\right) \Psi \! \right\rangle.
\end{align}
We furthermore obtain by Lemma \ref{Lem:Approximate_LHY} the following estimate on $\sum_{k\neq 0}\left\{\sqrt{A_k^2-B_k^2}-A_k+C_k\right\}$
\begin{align}
\nonumber
    & \left|\sum_{k\neq 0}\left\{\sqrt{A_k^2-B_k^2}-A_k+C_k\right\}-\sum_{k\neq 0}\left\{\sqrt{|k|^4+16\pi \mathfrak{a} N^\kappa |k|^2}-|k|^2-8\pi \mathfrak{a} N^\kappa+\frac{(8\pi \mathfrak{a} N^\kappa)^2}{2|k|^2}\right\}\right|\\
    \label{Eq:Sum_Comparison}
    & \ \ \ \ \ \  \ \ \  \ \ \  \ \ \ \lesssim N^{4\kappa-1}\log N+\frac{N^{3\kappa}}{K}\lesssim N^{-\tau}.
\end{align}
Using $e_k\geq 0$ in Eq.~(\ref{Eq:Useful_Lower_Bound}) and the state $\Psi_0$ which spans the space $\mathcal{V}_1$, we can consequently verify the first statement Eq.~(\ref{Eq:Beyond_GP_GSE})
\begin{align*}
  &   E^{(1)}_{N,\kappa} \! \geq   \! \! \braket{\Psi_0, H_{N,\kappa}\Psi_0} \! - \! C_1 N^{-\tau}  \! \! \geq  \! 4\pi \mathfrak{a}_{N^{1-\kappa}}N^\kappa (N \! - \! 1)  \! +  \! \frac{1}{2} \sum_{k\neq 0} \! \left\{ \! \sqrt{A_k^2 \! - \! B_k^2} \! - \! A_k \! + \! C_k \! \right\} \! - \! C_2 N^{-\tau}\\
     & \geq  \! 4\pi \mathfrak{a}_{N^{1-\kappa}}N^\kappa (N   \! -  \!  1)   \!  + \!   \frac{1}{2}\sum_{k\neq 0}\left\{\sqrt{|k|^4+16\pi \mathfrak{a} N^\kappa |k|^2}-|k|^2-8\pi \mathfrak{a} N^\kappa+\frac{(8\pi \mathfrak{a} N^\kappa)^2}{2|k|^2}\right\} \! - \! C_3 N^{-\tau} \!,
\end{align*}
where $C_1,C_2,C_3>0$ are suitable constants.

In order to verify the second statement of Theorem \ref{Th:Beyond_GP}, recall the definition of $\lambda^{(d)}_{N,\kappa}$ in Theorem \ref{Th:Beyond_GP} and let us define $\widetilde{e}_k:=\min\{e_k,\lambda^{(d)}_{N,\kappa}+1\}$. By Eq.~(\ref{Eq:Useful_Lower_Bound}) we obtain for a suitable constant $C>0$ the lower bound
\begin{align}
\label{Eq:Useful_For_Spectrum}
   & \braket{\Psi, H_{N,\kappa}\Psi}\geq 4\pi \mathfrak{a}_{N^{1-\kappa}}N^\kappa (N-1) + \frac{1}{2}\sum_{k\neq 0}\left\{\sqrt{A_k^2-B_k^2}-A_k+C_k\right\}-CN^{-\tau}\\
   \nonumber
    & \ \ \ \  \ \ \ +\left(1-CN^{-\tau}N^{-\frac{\kappa}{2}}\right)\left\langle  \! \Psi,\sum_{k\neq 0}\widetilde{e}_k\left(\gamma_k d_k \! + \! \nu_k d_{-k}^\dagger\right)^\dagger \left(\gamma_k d_k \! + \! \nu_k d_{-k}^\dagger\right) \Psi \! \right\rangle .
\end{align}
Since $|\widetilde{e}_k|\lesssim N^{\frac{\kappa}{2}}$ uniformly in $k$, we furthermore have by Eq.~(\ref{Eq:Variable_Comparison_Spectrum}) for $\Psi\in \mathcal{V}_d$
\begin{align*}
& \left\langle \Psi ,\left\{\sum_{k\neq 0}\widetilde{e}_k b_k^\dagger b_k -\sum_{k\neq 0}\widetilde{e}_k\left(\gamma_k d_k \! + \! \nu_k d_{-k}^\dagger\right)^\dagger \left(\gamma_k d_k \! + \! \nu_k d_{-k}^\dagger\right)\right\}\Psi\right\rangle\\
 & \ \ \lesssim \left(N^{\frac{5\kappa+\lambda_0-1}{2}}+N^{2\kappa+\lambda-1}\right)\left(\braket{\Psi,\widetilde{\mathcal{N}}\Psi}+1\right)\lesssim N^{-\tau}\left(\braket{\Psi,\widetilde{\mathcal{N}}\Psi}+1\right),
\end{align*}
where the operators $b_k$ are introduced above Lemma \ref{Lem:Lambda_Control}. In combination with Eq.~(\ref{Eq:Useful_For_Spectrum}), Eq.~(\ref{Eq:Advanced_Min_Max}) and Eq.~(\ref{Eq:Sum_Comparison}), we therefore obtain for a suitable constant $C>0$
\begin{align*}
     & E_{N,\kappa}^{(d)}\geq 4\pi \mathfrak{a}_{N^{1-\kappa}}N^\kappa (N-1) +  \frac{1}{2}\sum_{k\neq 0}\left\{\sqrt{|k|^4+16\pi \mathfrak{a} N^\kappa |k|^2}-|k|^2-8\pi \mathfrak{a} N^\kappa+\frac{(8\pi \mathfrak{a} N^\kappa)^2}{2|k|^2}\right\}\\
    & \ \ \ \ \ \ \ \ \ \ +\left(1-C N^{-\tau}N^{-\frac{\kappa}{2}}\right)\sup_{\Psi\in \mathcal{V}_d:\|\Psi\|=1}\left\langle  \! \Psi,\sum_{k\neq 0}\widetilde{e}_k b_k^\dagger b_k\Psi \! \right\rangle-CN^{-\tau}.
\end{align*}
Following the work \cite{DN}, respectively by making use of the excitation map $U_N$ introduced in \cite{LNSS}, we note that the operators $b_k$ can be extended from $\mathcal{F}^+_{N-1}$ to operators satisfying the CCR, to be precise there exists a Hilbert space extension $L^2_{\mathrm{sym}} \! \left(\Lambda^N\right)\subseteq \mathfrak{h}$ and operators $\mathfrak{b}_k$ defined on $\mathfrak{h}$, such that the family $\{\mathfrak{b}_k:k\in 2\pi\mathbb{Z}^3\setminus \{0\}\}$ is unitarily equivalent to the standard annihilation operators and $\mathfrak{b}_k \Psi=b_k \Psi$ for $\Psi\in \mathcal{F}^+_{N-1}$. Denoting the $d$-th eigenvalue of $\sum_{k\neq 0}\widetilde{e}_k \mathfrak{b}_k^\dagger \mathfrak{b}_k$ as $\widetilde{\lambda}^{(d)}_{N,\kappa}$ and making use of the fact that $\mathcal{V}_d\subseteq \mathcal{F}^+_{N-1}$ for $N$ large enough, we obtain by the min-max principle
 \begin{align}
 \label{Eq:Space_Change}
     \sup_{\Psi\in \mathcal{V}_d:\|\Psi\|=1}\left\langle  \! \Psi,\sum_{k\neq 0}\widetilde{e}_k b_k^\dagger b_k\Psi \! \right\rangle= \sup_{\Psi\in \mathcal{V}_d:\|\Psi\|=1}\left\langle  \! \Psi,\sum_{k\neq 0}\widetilde{e}_k \mathfrak{b}_k^\dagger \mathfrak{b}_k\Psi \! \right\rangle\geq \widetilde{\lambda}^{(d)}_{N,\kappa}.
 \end{align}
 Since the operators $\mathfrak{b}_k$ are unitarily equivalent to standard annihilation operators, we know that the eigenvalues $\widetilde{\lambda}^{(d)}_{N,\kappa}$ are an enumeration of $\left\{\sum_{k\neq 0} n_k \widetilde{e}_k: n_k\in \mathbb{N}_0\right\}$ in increasing order. Using $\left||k|^2+8\pi \mathfrak{a} N^{\kappa}-2A_k\right|\lesssim N^{2\kappa-1}(1+|k|)$ for $|k|<K$ and $\left|8\pi \mathfrak{a} N^{\kappa}-2B_k\right|\lesssim N^{2\kappa-1}(1+|k|)$, we obtain $|e_k-\sqrt{|k|^4+16\pi \mathfrak{a} N^\kappa |k|^2}|\lesssim N^{2\kappa-1}$ for fixed $k$, and consequently $|\widetilde{\lambda}^{(d)}_{N,\kappa}-\lambda^{(d)}_{N,\kappa}|\lesssim N^{2\kappa-1}\leq N^{-\tau}$. Using that $|\lambda^{(d)}_{N,\kappa}|\lesssim N^{\frac{\kappa}{2}}$, we obtain for suitable $C,C'>0$
 \begin{align}
 \nonumber
      & E_{N,\kappa}^{(d)}- 4\pi \mathfrak{a}_{N^{1-\kappa}}N^\kappa (N-1) -  \frac{1}{2}\sum_{k\neq 0}\left\{\sqrt{|k|^4+16\pi \mathfrak{a} N^\kappa |k|^2}-|k|^2-8\pi \mathfrak{a} N^\kappa+\frac{(8\pi \mathfrak{a} N^\kappa)^2}{2|k|^2}\right\}\\
\label{Eq:Reference_Difference_lambda}
    & \ \ \  \ \ \   \ \ \   \ \ \   \geq \left(1-C' N^{-\tau}N^{-\frac{\kappa}{2}}\right)\lambda^{(d)}_{N,\kappa}- C' N^{-\tau}\geq \lambda^{(d)}_{N,\kappa}-CN^{-\tau}.
 \end{align}
 Together with Theorem \ref{Th:Upper_Bound}, this concludes the proof of Theorem \ref{Th:Beyond_GP}.

 \section{Trial States and their Energy}
 \label{Sec:Trial_States_and_their_Energy}
 In this Section we want to verify the upper bound on $E_{N,\kappa}$ in Theorem \ref{Th:Upper_Bound}. For this purpose let us use the concrete choice $K:=\infty$ for the cut-off parameter and let us keep the term $\frac{1}{2}\sum_{j k,mn} \left(\pi_\mathcal{H} V_{N^{1-\kappa}}\pi_\mathcal{H}\right)_{j k,mn}\psi_{j k}^\dagger \psi_{m n}=\frac{1}{2}\sum_{j k,mn\neq 0} \left( V_{N^{1-\kappa}}\right)_{j k,mn}\psi_{j k}^\dagger \psi_{m n}$, which we have bounded from below by $0$ in Eq.~(\ref{Eq:Throwing_Away_High_Momenta}), yielding the identity
 \begin{align}
 \label{Eq:Full_Representation}
     H'_{N,\kappa}  \!  =  \! \sum_{k\neq 0}e_k\left(\gamma_k d_k \! + \! \nu_k d_{-k}^\dagger\right)^\dagger \! \left(\gamma_k d_k \! + \! \nu_k d_{-k}^\dagger\right) \! + \! \frac{1}{2}\sum_{j k,mn\neq 0} \left( V_{N^{1-\kappa}}\right)_{j k,mn}\psi_{j k}^\dagger \psi_{m n} \! - \! \sum_{i=1}^4 \mathcal{E}_i,
 \end{align}
 with $H'_{N,\kappa}:=H_{N,\kappa}-4\pi \mathfrak{a}_{N^{1-\kappa}}N^\kappa (N-1)- \frac{1}{2}\sum_{k\neq 0}\left\{\sqrt{A_k^2-B_k^2}-A_k+C_k\right\}$. In order to obtain from this representation of $ H_{N,\kappa}$ an upper bound on its eigenvalues $E_{N,\kappa}^{(d)}$, we need to find trial states $\Psi_d$ which simultaneously annihilate the variables $\gamma_k d_k+\nu_k d_{-k}^\dagger$ for $k\neq 0$ and $\psi_{j k}$ for $j,k\neq 0$, at least in an approximate sense. This will be carried out in the following two Subsections \ref{Subsection:Annihilation_I} and \ref{Subsection:Annihilation_II}, using the operators $\mathfrak{b}_k$ introduced above Eq.~(\ref{Eq:Space_Change}) as well as $\mathfrak{a}_k:=\gamma_k \mathfrak{b}_k-\nu_k \mathfrak{b}_{-k}^\dagger$, which is an extension of $a_0^\dagger \frac{1}{\sqrt{a_0 a_0^\dagger}}a_k$ from $\mathcal{F}^+_{N-1}$ to $\mathfrak{h}$.

 \subsection{Annihilation of $\gamma_k d_k+\nu_k d_{-k}^\dagger$}
  \label{Subsection:Annihilation_I}
We start by expressing $\gamma_k d_k+\nu_k d_{-k}^\dagger$ in terms of the extensions $\mathfrak{b}_k$ of the operators $b_k$. In order to do this, note that $\gamma_k d_k+\nu_k d_{-k}^\dagger=b_k-\gamma_k \delta_k-\nu_k \delta_{-k}^\dagger$, where $\delta_k$ is as in Lemma \ref{Lem:Delta_Control}, and $b_k \Psi=\mathfrak{b}_k \Psi$, as well as $a_0^\dagger \frac{1}{\sqrt{a_0 a_0^\dagger}}a_k \Psi=\mathfrak{a}_k \Psi$, for states $\Psi\in \mathcal{F}^+_{N-1}$. Therefore we obtain
\begin{align}
\label{Eq:First_Identification_Variable_I}
\left(\gamma_k d_k+\nu_k d_{-k}^\dagger\right)\Psi=\left(\mathfrak{b}_k+\sum_p \sqrt{N}f_{p,k}\mathfrak{b}^\dagger_{p-k}\mathfrak{a}_p+\epsilon_k\right)\Psi,
\end{align}
for $\Psi\in \mathcal{F}^+_{N-1}$, with $\delta_k''$ as in Lemma \ref{Lem:Delta_Control} and
\begin{align*}
    \epsilon_k=-\nu_k \delta_{-k}^\dagger -\gamma_k \delta_k''+\sum_p\left(\gamma_k\gamma_{p-k}\sqrt{a_0^\dagger a_0}-\sqrt{N}\right) f_{p,k}\mathfrak{b}^\dagger_{p-k}\mathfrak{a}_p+\sum_{p}\gamma_k \nu_{p-k}\sqrt{a_0^\dagger a_0}f_{p,k} \mathfrak{b}_{k-p}\mathfrak{a}_p.
\end{align*}
In the following Lemma \ref{Lemma:Linear_Annihilation_epsilon} we show that the contribution coming from the $\epsilon_k$ can be regarded as a small error.
\begin{lem}
\label{Lemma:Linear_Annihilation_epsilon}
    There exists a constant $C>0$, such that we have for $\kappa<1$ the estimate
    \begin{align*}
        \sum_k e_k\,  \epsilon_k^\dagger \epsilon_k\leq C N^{6\kappa-1}\! \left(\widetilde{\mathcal{N}}+1\right)^3.
    \end{align*}
\end{lem}
\begin{proof}
    Due to the positivity of $e_k$ we can estimate the square separately for each term in the definition of $\epsilon_k$. Starting with $-\nu_k \delta_{-k}^\dagger$ we obtain by a slight adaptation of Eq.~(\ref{Eq:Strong_Delta_Control})
    \begin{align*}
        \sum_{k\neq 0}e_k \left(\nu_k \delta_{-k}^\dagger\right)^\dagger \nu_k \delta_{-k}^\dagger\lesssim \left(\sum_{k\neq 0}e_k \nu_k^2 |k|^{-4}\right) N^{5\kappa-1}\! \left(\widetilde{\mathcal{N}}+1\right)^3\lesssim N^{6\kappa-1}\! \left(\widetilde{\mathcal{N}}+1\right)^3.
    \end{align*}
    By a slight modification of the upper bound on $\left(\delta_k''\right)^\dagger \delta_k''$ obtained in Lemma \ref{Lem:Particle_Number_Comparison}, we have
    \begin{align*}
        \sum_{k\neq 0}e_k \left(\delta_k''\right)^\dagger \delta_k''\lesssim \frac{\sum_{k\neq 0}e_k \gamma_k^2 w_k^2}{N}N^{\frac{11\kappa}{2}-1} \! \left(\widetilde{\mathcal{N}}+1\right)\lesssim N^{\frac{11\kappa}{2}-1} \! \left(\widetilde{\mathcal{N}}+1\right).
    \end{align*}
    Regarding the term $\sum_p\left(\gamma_k\gamma_{p-k}\sqrt{\frac{a_0^\dagger a_0}{N}}-1\right) \sqrt{N}f_{p,k}\mathfrak{b}^\dagger_{p-k}\mathfrak{a}_p=X_k+\widetilde{X}_k$ with 
    \begin{align*}
      X_k: & =\sum_p\left(\sqrt{\frac{a_0^\dagger a_0}{N}}-1\right) \sqrt{N}f_{p,k}\mathfrak{b}^\dagger_{p-k}\mathfrak{a}_p, \\  
      \widetilde{X}_k: & =\sum_p\left(\gamma_k\gamma_{p-k}-1\right)\sqrt{\frac{a_0^\dagger a_0}{N}} \sqrt{N}f_{p,k}\mathfrak{b}^\dagger_{p-k}\mathfrak{a}_p,
    \end{align*}
    note that $\left|\sqrt{\frac{a_0^\dagger a_0}{N}}-1\right|^2\lesssim \frac{\mathcal{N}^2}{N^2}$ and $\left|\left(\gamma_k\gamma_{p-k}-1\right)\sqrt{\frac{a_0^\dagger a_0}{N}}\right|^2\lesssim \frac{N^{\kappa}}{|k|^{2}}+\frac{N^{\kappa}}{|p-k|^{2}}$. Since $e_k\lesssim N^{\frac{\kappa}{2}}|k|^2$ we have $\sum_{k\neq 0}e_k X_k^\dagger X_k \lesssim N^{\frac{\kappa}{2}}\sum_{k\neq 0}|k|^2 X_k^\dagger X_k\lesssim N^{\frac{3\kappa}{2}}\frac{\mathcal{N}^3}{N^2}\lesssim N^{6\kappa-2}(\widetilde{\mathcal{N}}+1)^3$ and furthermore
    \begin{align*}
         & \sum_{k\neq 0} \! |k|^2  \! \widetilde{X}_k^\dagger   \widetilde{X}_k  \! \lesssim \!  N^{\kappa} \! \sup_{p} \! \left\{ \!  \! \sum_k \!  |k|^2 \big(\sqrt{N}f_{p,k}\big)^2 \! \!  \left(\frac{1}{|k|^2} \! + \! \frac{1}{|p-k|^2}\right) \!  \! \right\} \! \sum_p  \! \mathfrak{a}_p^\dagger  \! \left( \! \widetilde{\mathcal{N}} \! + \! 1 \! \right) \!  \mathfrak{a}_p \! \lesssim \!  \frac{N^{\frac{9\kappa}{2}}}{N}\! \left( \! \widetilde{\mathcal{N}} \! + \! 1 \! \right)^2,
    \end{align*}
  where we have used $\mathcal{N}^3\lesssim N^{\frac{9\kappa}{2}}\left(\widetilde{\mathcal{N}}+1\right)^m$, $\sum_k |k|^2 \big(\sqrt{N}f_{p,k}\big)^2 \! \left(\frac{1}{|k|^2} \! + \! \frac{1}{|p-k|^2}\right)\lesssim N^{2\kappa-1}$ and $\sum_p a_p^\dagger \left(\widetilde{\mathcal{N}} \! + \! 1\right) a_p\lesssim N^{\frac{3\kappa}{2}}\! \left(\widetilde{\mathcal{N}}+1\right)^2$. The final term in $\epsilon_k$ can be estimated similarly.
\end{proof}
In order to cancel the term $\sum_p \sqrt{N}f_{p,k}\mathfrak{b}^\dagger_{p-k}\mathfrak{a}_p$ in Eq.~(\ref{Eq:First_Identification_Variable_I}), let us introduce the operators $G:=\frac{1}{2}\sum_{p,k}\sqrt{N}f_{p,k}\mathfrak{b}_k^\dagger \mathfrak{b}_{p-k}^\dagger\mathfrak{a}_{p}$ and $\Theta_\ell:=\frac{1}{2}\sum_{p}\nu_\ell \sqrt{N}f_{-\ell,p}\mathfrak{b}^\dagger_{-(\ell+p)}\mathfrak{b}^\dagger_p$, which allows us to write
\begin{align}
\label{Eq:Second_Identification_Variable_I}
\left(\gamma_k d_k+\nu_k d_{-k}^\dagger\right) \! \left(1 \! - \! G \right)\Psi =\left(\mathfrak{b}_k - \Theta_k  \! + \! \epsilon_k \! \left(1 \! - \! G \right) \! - \!  G\mathfrak{b}_k \! - \! \left(\sum_p \sqrt{N}f_{p,k}\mathfrak{b}^\dagger_{p-k}\mathfrak{a}_p\right) \! G\right)\Psi 
\end{align}
for $\Psi\in \mathcal{F}^+_{N-2}$, where we have used that $(1-G)\mathcal{F}^+_{N-2}\subset \mathcal{F}^+_{N-1}$.
\begin{cor}
\label{Corollary:Annihilation_I}
    Let us define the coefficients $\Pi_k:=\frac{1}{4}\sum_{j\neq -k} e_{j+k}\nu_{j+k}^2 N f_{j+k,j}\! \left(f_{j+k,j}+f_{j+k,k}\right)$. Then there exists a $C>0$, such that
    \begin{align}
    \label{Eq:Kinetic_Epsilon_Total}
        & \sum_k e_k\Xi_k^\dagger\,  \Xi_k\Big|_{\mathcal{F}^+_{N-2}}  \leq C N^{\frac{11\kappa}{2}-1} \! \left(\widetilde{\mathcal{N}}+1\right)^2 \! \left(\sum_k e_k \mathfrak{b}_k^\dagger \mathfrak{b}_k+1\right)\Big|_{\mathcal{F}^+_{N-2}}
    \end{align}
   with $\Xi_k:=\left(\gamma_k d_k+\nu_k d_{-k}^\dagger\right) \! \left(1 \! - \! G \right)-\mathfrak{b}_k+\Theta_k$. Furthermore we have $|\sum_k \Pi_k|\leq C  N^{\frac{9\kappa}{2}}\frac{\log N}{N}$ as well as $\pm \left(\sum_\ell e_\ell \Theta_\ell^\dagger \Theta_\ell-\sum_k \Pi_k\right)\leq C N^{\frac{9\kappa}{2}-1} \left(\widetilde{\mathcal{N}}+1\right)^2 \! \left(\sum_k e_k \mathfrak{b}_k^\dagger \mathfrak{b}_k+1\right)$.
\end{cor}
\begin{proof}
Defining $X_k:=\left(\sum_p \sqrt{N}f_{p,k}\mathfrak{b}^\dagger_{p-k}\mathfrak{a}_p\right) \! G$, we obtain by Eq.~(\ref{Eq:Second_Identification_Variable_I})
\begin{align}
\label{Eq:Second_Identification_Variable_I_alt}
    \sum_k e_k\Xi_k^\dagger\,  \Xi_k\Big|_{\mathcal{F}^+_{N-2}} \! \lesssim  \! \sum_k e_k (1\! - \! G)^\dagger \epsilon_{k}^\dagger \epsilon_{k}(1\! - \! G)\Big|_{\mathcal{F}^+_{N-2}} \!  \! + \! \sum_{k\neq 0} e_k \mathfrak{b}_k^\dagger G^\dagger G \mathfrak{b}_k\Big|_{\mathcal{F}^+_{N-2}} \! \! + \! \sum_{k\neq 0}e_k X_k^\dagger X_k\Big|_{\mathcal{F}^+_{N-2}}  \!  \! .
\end{align}
  First of all note that we have by Lemma \ref{Lemma:Linear_Annihilation_epsilon}
\begin{align*}
\sum_{ k\neq 0} & e_k (1\! - \! G)^\dagger \epsilon_{k}^\dagger \epsilon_{k}(1\! - \! G)\lesssim N^{6\kappa-1}\! (1\! - \! G)^\dagger \left(\widetilde{\mathcal{N}}+1\right)^3(1\! - \! G)\\
 & \lesssim N^{6\kappa-1}\left(\widetilde{\mathcal{N}}+1\right)^3\lesssim N^{\frac{11\kappa}{2}-1}\left(\widetilde{\mathcal{N}}+1\right)^2\left(\sum_k e_k \mathfrak{b}_k^\dagger \mathfrak{b}_k+1\right).
\end{align*}  
Furthermore we have that $G^\dagger G\lesssim N^{2\kappa-1}\! \left(\mathcal{N}+1\right)^2\lesssim N^{5\kappa-1}\! \left(\widetilde{\mathcal{N}}+1\right)^2$, and consequently $\sum_{k\neq 0} e_k \mathfrak{b}_k^\dagger G^\dagger G \mathfrak{b}_k\lesssim N^{5\kappa-1}\! \left(\widetilde{\mathcal{N}}+1\right)^2 \! \left(\sum_k e_k \mathfrak{b}_k^\dagger \mathfrak{b}_k+1\right)$. Regarding the $X_k$ term, note that 
\begin{align*}
    \sum_{k\neq 0}e_k \! \left( \! \sum_p \sqrt{N}f_{p,k}\mathfrak{b}^\dagger_{p-k}\mathfrak{a}_p \! \right)^\dagger  \! \left( \! \sum_p \sqrt{N}f_{p,k}\mathfrak{b}^\dagger_{p-k}\mathfrak{a}_p \! \right) \! \lesssim  \! N^\kappa \left(\mathcal{N} \! + \! \sum_{\ell\neq 0}\mathfrak{b}_\ell^\dagger \mathcal{N}\mathfrak{b}_\ell\right) \! \lesssim  \! N^{\frac{5\kappa}{2}} \! \left(\widetilde{\mathcal{N}}^2 \! + \! 1\right),
\end{align*}
where we have used $\sum_{k\neq 0}|k|^2 f_{p,k}^2=\sum_{k\neq 0}|k|^2 f_{p,p-k}^2\lesssim N^{\kappa -1 }$ by Lemma \ref{Lem:Increased_Decay}, and therefore $\sum_{k\neq 0}e_k X_k^\dagger X_k\lesssim N^{\frac{5\kappa}{2}} G^\dagger \left(\widetilde{\mathcal{N}}^2 \! + \! 1\right) G\lesssim N^{\frac{11\kappa}{2}-1}\left(\widetilde{\mathcal{N}}^4 \! + \! 1\right) $. This concludes the proof of Eq.~(\ref{Eq:Kinetic_Epsilon_Total}) by Eq.~(\ref{Eq:Second_Identification_Variable_I_alt}). In order to estimate $\sum_\ell e_\ell \Theta_\ell^\dagger \Theta_\ell$, let us define the operator $(\Pi)_{jk,mn}:=\frac{1}{4}\delta_{j+k=m+n}e_{j+k}\nu_{j+k}^2 N f_{j+k,j} f_{m+n,m}$, which allows us to write
\begin{align*}
    \sum_\ell e_\ell \Theta_\ell^\dagger \Theta_\ell=\sum_{k\neq 0} \Pi_k \! \left(2\mathfrak{b}_k^\dagger \mathfrak{b}_k+1\right)\! +\sum_{jk,mn}(\Pi)_{jk,mn}\mathfrak{b}_k^\dagger \mathfrak{b}_j^\dagger \mathfrak{b}_m \mathfrak{b}_n.
\end{align*}
Note that $|N f_{j+k,j}|\lesssim \frac{N^\kappa}{|j|^2+|k|^2}$ by Lemma \ref{Lem:Coefficient_Comparison} and $e_k\lesssim N^{\frac{\kappa}{2}}|k|^2$, and therefore
\begin{align*}
 |\Pi_k|\lesssim   N^{\frac{5\kappa}{2}-1}\sum_{j\neq -k}\frac{|j+k|\nu_{j+k}^2}{ \left(|j|^2  + |k|^2\right)^2}\lesssim N^{\frac{5\kappa}{2}-1}\sum_{j\neq 0}\frac{|j|\nu_{j}^2}{ \left(|j|^2  + |k|^2\right)^2}.
\end{align*}
Again by Lemma \ref{Lem:Coefficient_Comparison} we have $0\leq \nu_j\lesssim \frac{N^\kappa}{|j|^2}$, which immediately gives us $|\Pi_k|\lesssim N^{\frac{9\kappa}{2}-1}$ and therefore $\pm \sum_{k\neq 0} \Pi_k\mathfrak{b}_k^\dagger \mathfrak{b}_k\lesssim N^{\frac{9\kappa}{2}-1}\widetilde{\mathcal{N}}$. Making use of the fact that
\begin{align*}
   \sum_{k}\frac{1}{\left(|j|^2  + |k|^2\right)^2}\lesssim \int_{\mathbb{R}^3}\frac{\mathrm{d}x}{\left(|j|^2+|x|^2\right)^2}=\left(\int_{\mathbb{R}^3}\frac{\mathrm{d}x}{\left(1+|x|^2\right)^2}\right)\frac{1}{|j|}\lesssim \frac{1}{|j|}
\end{align*}
we obtain $\left|\sum_{0<|j|\leq N}\sum_k \frac{|j|\nu_{j}^2}{ \left(|j|^2  + |k|^2\right)^2}\right|\lesssim N^{2\kappa}\sum_{0<|j|\leq N}|j|^{-3}\lesssim N^{2\kappa}\log N$. Regarding the left over part note that $\left|\sum_{|j|>N}\sum_k \frac{|j|\nu_{j}^2}{ \left(|j|^2  + |k|^2\right)^2}\right|\lesssim \sum_{|j|>N}|j| \nu_j^2\leq N^{-1}\sum_{|j|>0}|j|^2 \nu_j^2\lesssim N^{\kappa-1}$, where we have used Lemma \ref{Lem:Increased_Decay}. Finally we have $\pm \Pi\lesssim N^{\frac{9\kappa}{2}-1}(-\Delta_2)$ and therefore we conclude that $\pm \sum_{jk,mn}(\Pi)_{jk,mn}\mathfrak{b}_k^\dagger \mathfrak{b}_j^\dagger \mathfrak{b}_m \mathfrak{b}_n\lesssim N^{\frac{9\kappa}{2}-1} \! \left(\widetilde{\mathcal{N}}+1\right)\left(\sum_k e_k \mathfrak{b}_k^\dagger \mathfrak{b}_k+1\right)$.
\end{proof}

\subsection{Annihilation of $\psi_{j k}$}
 \label{Subsection:Annihilation_II}
Using the definition of $f_{j,k}$ in Eq.~(\ref{Eq:Definition_f}) and introducing $\chi:=\left(a_0^\dagger \frac{1}{\sqrt{a_0 a_0^\dagger}}\right)^2$, we can rewrite
\begin{align}
\label{Eq:psi_in_terms_of_b}
 \chi \psi_{j k}\Psi=\left(\mathfrak{b}_j\mathfrak{b}_k +  \sqrt{N}f_{k+j,k}\mathfrak{a}_{k+j}+\epsilon_{j k}\right)\Psi 
\end{align}
for states $\Psi \in \mathcal{F}^{+}_{N-2}$ and $j,k\neq 0$, with
\begin{align*}
    \epsilon_{j k}: & =\delta_{j+k=0}\! \left(w_k-2\gamma_k \nu_k \right)+(\gamma_j \gamma_k-1)\mathfrak{b}_j \mathfrak{b}_k-\nu_j \gamma_k \mathfrak{b}_{-j}^\dagger \mathfrak{b}_k-\nu_k \gamma_j \mathfrak{b}_{-k}^\dagger \mathfrak{b}_j+\nu_k \nu_j\mathfrak{b}_{-k}^\dagger \mathfrak{b}_j^\dagger \\
    & \ \ \ \  \ \ + \delta_{j+k=0}w_k\left[\chi \frac{a_0^2}{N}-1\right]+\sum_{p\neq 0}\left(\sqrt{a_0^\dagger a_0-1}-\sqrt{N}\right)f_{p,k}\mathfrak{a}_k.
\end{align*}
As the following Lemma \ref{Lem:Potential_Epsilon} demonstrates, the operators $\epsilon_{j k}$ can be regarded as being small.
\begin{lem}
\label{Lem:Potential_Epsilon}
   We have the estimate $\sum_{j k, mn\neq 0}\left(V_{N^{1-\kappa}}\right)_{jk,mn}\epsilon_{jk}^\dagger \epsilon_{mn}\lesssim N^{4\kappa -1}\! \left(\widetilde{\mathcal{N}}^2+1\right)$.
\end{lem}
\begin{proof}
    Due to the sign of $V_{N^{1-\kappa}}$ it is enough to verify the statement for each term in the definition of $\epsilon_{j k}$ individually. Regarding the term $\delta_{j+k=0}\! \left(w_k-2\gamma_k \nu_k \right)$ note that $|w_k-2\gamma_k \nu_k|\lesssim \frac{N^\kappa}{|k|^2}\mathds{1}\! \left(|k|\leq N^{\frac{\kappa}{2}}\right)+\frac{N^{2\kappa}}{|k|^4}\mathds{1}\! \left(|k|> N^{\frac{\kappa}{2}}\right)$, and consequently
    \begin{align*}
        \sum_{j k, mn\neq 0}\left(V_{N^{1-\kappa}}\right)_{(-k)k,(-j)j}|w_k-2\gamma_k \nu_k||w_j-2\gamma_j \nu_j|\lesssim N^{4\kappa-1}.
    \end{align*}
    For the next term $(\gamma_j \gamma_k-1)\mathfrak{b}_j \mathfrak{b}_k$, note that the operator $A_{jk,mn}:=\left(V_{N^{1-\kappa}}\right)_{jk,mn}(\gamma_j \gamma_k-1)(\gamma_m \gamma_n-1)$ has an operator norm bounded by $CN^{\kappa-1}\sup_T \sum_p (\gamma_p \gamma_{p+T}-1)^2$, which follows from the weighted Schur test, see for example the proof of Lemma \ref{Lem:Lambda_Control}. Since $\sum_p (\gamma_p \gamma_{p+T}-1)^2\lesssim N^{\frac{3\kappa}{2}}$, uniformly in $T$, we obtain $ \sum_{j k, mn\neq 0}\left(A\right)_{jk,mn}\mathfrak{b}_k^\dagger \mathfrak{b}_j^\dagger\mathfrak{b}_m \mathfrak{b}_n\lesssim N^{\frac{5\kappa}{2}-1}\left(\widetilde{\mathcal{N}}+1\right)^2$. Coming to the next term $\nu_j \gamma_k \mathfrak{b}_{-j}^\dagger \mathfrak{b}_k$ we first notice that
    \begin{align}
    \label{Eq:Multiple_Errors_I}
        \sum_{j k, mn\neq 0}\left(V_{N^{1-\kappa}}\right)_{jk,mn}\! \left(\nu_j \gamma_k \mathfrak{b}_{-j}^\dagger \mathfrak{b}_k \right)^\dagger\! \left(\nu_m \gamma_n \mathfrak{b}_{-m}^\dagger \mathfrak{b}_n\right)=\sum_{q}\mu_q \mathfrak{b}_q^\dagger \mathfrak{b}_q+\sum_{j k, mn\neq 0}\left(A'\right)_{jk,mn}\mathfrak{b}_k^\dagger \mathfrak{b}_j^\dagger\mathfrak{b}_m \mathfrak{b}_n,
    \end{align}
    with $\mu_q:=\gamma_q^2\sum_p \left(V_{N^{1-\kappa}}\right)_{p q,p q}\nu_p^2$ and $\left(A'\right)_{jk,mn}:=\left(V_{N^{1-\kappa}}\right)_{jn,mk}\gamma_j \nu_k \gamma_m \nu_n$. Since $|\mu_q|\lesssim N^{\frac{7\kappa}{2}-1}$ and $\|A'\|\lesssim N^{\frac{7\kappa}{2}-1}$, the term in Eq.~(\ref{Eq:Multiple_Errors_I}) is bounded by $N^{\frac{7\kappa}{2} -1}\! \left(\widetilde{\mathcal{N}}^2+1\right)$. Note that the terms $\nu_k \gamma_j \mathfrak{b}_{-k}^\dagger \mathfrak{b}_j$ and $\nu_k \nu_j\mathfrak{b}_{-k}^\dagger \mathfrak{b}_j^\dagger$ can be analysed analogously. Coming to the term $\delta_{j+k=0}w_k\left[\chi \frac{a_0^2}{N}-1\right]$ we observe that $\left[\chi \frac{a_0^2}{N}-1\right]^\dagger \left[\chi \frac{a_0^2}{N}-1\right]\lesssim N^{-2}\! \left(\mathcal{N}+1\right)^2\lesssim N^{3\kappa-2}\! \left(\widetilde{\mathcal{N}}+1\right)^2$. Therefore $\sum_{j k, mn\neq 0}\left(V_{N^{1-\kappa}}\right)_{(-k)k,(-j)j}w_k w_j\left[\chi \frac{a_0^2}{N}-1\right]^\dagger \left[\chi \frac{a_0^2}{N}-1\right]\lesssim N^{3\kappa-2}\mu \! \left(\widetilde{\mathcal{N}}+1\right)^2$ with
    \begin{align*}
        \mu:= \!  \!  \! \sum_{j k, mn\neq 0} \! \!  \!   \!  \! \left(V_{N^{1-\kappa}}\right)_{(-k)k,(-j)j}w_k w_j \! = \! N^2 \!  \!  \braket{V_{N^{1-\kappa}},R V_{N^{1-\kappa}} R V_{N^{1-\kappa}}} \! \leq  \! N^2 \!  \!  \braket{V_{N^{1-\kappa}},R  V_{N^{1-\kappa}}} \! \lesssim \!  N^{1+\kappa}.
    \end{align*} \! 
    The final contribution $\sum_{p\neq 0}\left(\sqrt{a_0^\dagger a_0-1}-\sqrt{N}\right)f_{p,k}\mathfrak{a}_k$ can be dealt with analogously.
\end{proof}
In order to get rid of the term $\sqrt{N}f_{k+j,k}\mathfrak{a}_{k+j}$ in Eq.~(\ref{Eq:psi_in_terms_of_b}), let us again use the operator $G$ introduced above Eq.~(\ref{Eq:Second_Identification_Variable_I}), which allows us to write
\begin{align}
\label{Eq:Second_Identification_Variable_II}
    \chi \psi_{j k} \! \left(1 \! - \! G\right)\Psi =\left(\mathfrak{b}_j\mathfrak{b}_k +\epsilon_{j,k}  \left(1 \! - \!  G\right)- \widetilde{\epsilon}_{j,k}\right)\Psi 
\end{align}
for states $\Psi\in \mathcal{F}^{+}_{N-3}$, where we have used $(1-G)\mathcal{F}^{+}_{N-3}\subset \mathcal{F}^{+}_{N-2}$, with 
\begin{align*}
   \tilde{\epsilon}_{j,k}: \!   & \!  = \!  \!   \frac{1}{2}\sum_{p,q} \!  \sqrt{N}f_{p,q}\mathfrak{b}_q^\dagger \mathfrak{b}_{p-q}^\dagger \mathfrak{b}_j \mathfrak{b}_k \mathfrak{a}_p \!  + \!  \sum_p  \!  \sqrt{N}f_{p,k} \mathfrak{b}_{p-k}^\dagger  \mathfrak{b}_j  \mathfrak{a}_p \!  + \!  \sum_p  \!  \sqrt{N}f_{p,j} \mathfrak{b}_{p-j}^\dagger  \mathfrak{b}_k  \mathfrak{a}_p  \!  + \!  \sqrt{N}f_{k+j,k}\mathfrak{a}_{k+j} G.
\end{align*}
\begin{lem}
\label{Lem:Potential_Epsilon_Tilde}
   Let $\mathbb{K} :=\sum_k e_k \mathfrak{b}_k^\dagger \mathfrak{b}_k+\sum_{j k,mn\neq 0} \left( V_{N^{1-\kappa}}\right)_{j k,mn} \mathfrak{b}_k^\dagger \mathfrak{b}_j^\dagger \mathfrak{b}_m \mathfrak{b}_n$. Then we have 
   \begin{align*}
       \sum_{j k, mn\neq 0}\left(V_{N^{1-\kappa}}\right)_{jk,mn}\tilde{\epsilon}_{jk}^\dagger \tilde{\epsilon}_{mn}\lesssim N^{\frac{11\kappa}{2}-1}\! \left(\widetilde{\mathcal{N}}^4+1\right)\! \! \left(\mathbb{K}+1\right).
   \end{align*}
\end{lem}
\begin{proof}
 Note that we can carry out the estimate for each term in $\tilde{\epsilon}_{jk}$ individually due to the positivity of $V$.  Starting with the analysis of $\sqrt{N}f_{k+j,k}\mathfrak{a}_{k+j} G$, we obtain
 \begin{align*}
     \sum_{j k, mn\neq 0} & \left(V_{N^{1-\kappa}}\right)_{jk,mn}\left(\sqrt{N}f_{k+j,k}\mathfrak{a}_{k+j} G\right)^\dagger \left(\sqrt{N}f_{n+m,n}\mathfrak{a}_{n+m} G\right)\lesssim \sum_{p\neq 0}\mu_p G^\dagger  a_p^\dagger a_p G
 \end{align*}
 with $\mu_p:=N\underset{j+k=p}{\sum_{j k, mn\neq 0}}\left(V_{N^{1-\kappa}}\right)_{jk,mn}\overline{f_{k+j,k}}f_{n+m,n}$. Note that
 \begin{align*}
     \mu_p=N\braket{e^{ipx}V_{N^{1-\kappa}},R V_{N^{1-\kappa}} R e^{ipx}V_{N^{1-\kappa}}}\leq N\braket{e^{ipx}V_{N^{1-\kappa}},R  e^{ipx}V_{N^{1-\kappa}}}\lesssim N^{\kappa},
 \end{align*}
 and therefore $\sum_{p\neq 0}\mu_p G^\dagger  a_p^\dagger a_p G\lesssim N^{\kappa}G^\dagger \mathcal{N} G\lesssim N^{\frac{5\kappa}{2}}G^\dagger \left(\widetilde{\mathcal{N}}+1\right) G\lesssim N^{\frac{11\kappa}{2}-1}\left(\widetilde{\mathcal{N}}^3+1\right)$. Regarding the other terms, we are going to expand them further according to $\mathfrak{a}_k=\gamma_k \mathfrak{b}_k-\nu_k \mathfrak{b}_{-k}^\dagger$ and bringing them in normal order, e.g. $ \sum_p    \sqrt{N}f_{p,k} \mathfrak{b}_{p-k}^\dagger  \mathfrak{b}_j  \mathfrak{a}_p=\sum_p   \gamma_p \sqrt{N}f_{p,k} \mathfrak{b}_{p-k}^\dagger  \mathfrak{b}_j  \mathfrak{b}_p-\sum_p   \nu_p \sqrt{N}f_{p,k} \mathfrak{b}_{p-k}^\dagger  \mathfrak{b}_j  \mathfrak{b}_p^\dagger=\sum_p   \gamma_p \sqrt{N}f_{p,k} \mathfrak{b}_{p-k}^\dagger  \mathfrak{b}_j  \mathfrak{b}_p-\sum_p   \nu_p \sqrt{N}f_{p,k} \mathfrak{b}_{p-k}^\dagger \mathfrak{b}_p^\dagger \mathfrak{b}_j -\nu_j \sqrt{N}f_{j,k} \mathfrak{b}_{j-k}^\dagger$. Let us illustrate how to proceed, by analysing the $\nu_j \sqrt{N}f_{j,k} \mathfrak{b}_{j-k}^\dagger$ term in more detail
 \begin{align*}
   &  \ \ \ \ \sum_{j k, mn\neq 0}\left(V_{N^{1-\kappa}}\right)_{jk,mn}\left(\nu_j \sqrt{N}f_{j,k} \mathfrak{b}_{j-k}^\dagger\right)^\dagger \left(\nu_m \sqrt{N}f_{m,n} \mathfrak{b}_{m-n}^\dagger\right)\\
   & \lesssim N^{\kappa-1}\sum_{jk}|k|^{2}|j|^{1+\delta}\left(\nu_j \sqrt{N}f_{j,k} \mathfrak{b}_{j-k}^\dagger\right)^\dagger \left(\nu_j \sqrt{N}f_{j,k} \mathfrak{b}_{j-k}^\dagger\right)\lesssim N^{5\kappa-2}\widetilde{\mathcal{N}}+N^{2\kappa-1}\sum_j |j|^{1+\delta}\nu_j^2
 \end{align*}
for $0<\delta<1$, where we have used $\sum_{k\neq 0}|k|^2 f_{j,k}^2=\sum_{k\neq 0}|k|^2 f_{j,j-k}^2\lesssim N^{\kappa -1 }$ by Lemma \ref{Lem:Increased_Decay}. Furthermore $\sum_j |j|^{1+\delta}\nu_j^2\lesssim N^{2\kappa+\delta}$, again by Lemma \ref{Lem:Increased_Decay}, which concludes the argument. Regarding the final term we note that $A_{p,p'}:=\left(\sum_{q'}\sqrt{N}\gamma_{p'} f_{p',q'}\mathfrak{b}_{q'}^\dagger \mathfrak{b}_{p'-q'}^\dagger\right)^\dagger \left(\sum_q \sqrt{N}\gamma_p f_{p,q}\mathfrak{b}_q^\dagger \mathfrak{b}_{p-q}^\dagger\right)$ satisfies $A\lesssim N^{\frac{5\kappa}{2}-1} \! \left(\widetilde{\mathcal{N}}+1\right)^2$, and therefore
\begin{align*}
   & \ \ \ \  \sum_{j k, mn\neq 0}\left(V_{N^{1-\kappa}}\right)_{jk,mn}\left(\sum_{p,q} \!  \sqrt{N}f_{p,q}\mathfrak{b}_q^\dagger \mathfrak{b}_{p-q}^\dagger \mathfrak{b}_j \mathfrak{b}_k \mathfrak{a}_p\right)^\dagger \left(\sum_{p,q} \!  \sqrt{N}f_{p,q}\mathfrak{b}_q^\dagger \mathfrak{b}_{p-q}^\dagger \mathfrak{b}_m \mathfrak{b}_n \mathfrak{a}_p\right) \\
   & \lesssim N^{\frac{5\kappa}{2}-1} \sum_{j k, mn\neq 0}\left(V_{N^{1-\kappa}}\right)_{jk,mn} \mathfrak{b}_k^\dagger \mathfrak{b}_j^\dagger \! \left(\widetilde{\mathcal{N}}+1\right)^2 \! \mathfrak{b}_m \mathfrak{b}_n + N^{3\kappa-1}\widetilde{\mathcal{N}}\lesssim N^{3\kappa -1} \! \left(\widetilde{\mathcal{N}}+1\right)^2 \mathbb{K}. 
\end{align*}
\end{proof}

\begin{cor}
\label{Corollary:Annihilation_II}
Let $\Xi_{jk}:=\psi_{j k}(1 \! - \! G)-\mathfrak{b}_j \mathfrak{b}_k$. Then there exists a $C>0$ such that for $\kappa<\frac{2}{11}$
    \begin{align}
    \label{Eq:Potential_Epsilon_Total}
  \sum_{j k,mn\neq 0} \!  \!  \left( V_{N^{1-\kappa}}\right)_{j k,mn}\Xi_{jk}^\dagger  \, \Xi_{mn}\Big|_{\mathcal{F}^{+}_{N-3}} \! \leq  \! C N^{\frac{11\kappa}{2}-1} \! \left(\widetilde{\mathcal{N}} \! + \! 1\right)^4 \! \left(\mathbb{K} \! + \! 1\right)\Big|_{\mathcal{F}^{+}_{N-3}}.
    \end{align}
\end{cor}
\begin{proof}
By Eq.~(\ref{Eq:Second_Identification_Variable_II}) we have the bound
\begin{align*}
    &\sum_{j k,mn\neq 0} \!  \!  \left( V_{N^{1-\kappa}}\right)_{j k,mn}\Xi_{jk}^\dagger  \, \Xi_{mn}\Big|_{\mathcal{F}^{+}_{N-3}}\lesssim \sum_{j k, mn\neq 0}\left(V_{N^{1-\kappa}}\right)_{jk,mn}(1\! - \! G)^\dagger \epsilon_{jk}^\dagger \epsilon_{mn}(1\! - \! G)\Big|_{\mathcal{F}^{+}_{N-3}}\\
  &  \ \   \ \     \ \   \ \ \   \ \     \ \   \ \    \ \   \ \ \   \ \     \ \   \ \    \ \   \ \ + \sum_{j k, mn\neq 0}\left(V_{N^{1-\kappa}}\right)_{jk,mn}\tilde{\epsilon}_{jk}^\dagger \tilde{\epsilon}_{mn}\Big|_{\mathcal{F}^{+}_{N-3}}.
\end{align*}
    Note that $\sum_{j k, mn\neq 0}\left(V_{N^{1-\kappa}}\right)_{jk,mn}(1\! - \! G)^\dagger \epsilon_{jk}^\dagger \epsilon_{mn}(1\! - \! G)\lesssim N^{\frac{7\kappa}{2}-1}\! (1\! - \! G)^\dagger \left(\widetilde{\mathcal{N}}^2+1\right)(1\! - \! G)\lesssim N^{\frac{7\kappa}{2}-1}\left(\widetilde{\mathcal{N}}^5+1\right)$ by Lemma \ref{Lem:Potential_Epsilon}. In combination with Lemma \ref{Lem:Potential_Epsilon_Tilde} this concludes the proof.
\end{proof}

\subsection{Proof of Theorem \ref{Th:Upper_Bound}}
Regarding the verification of Theorem \ref{Th:Upper_Bound}, let us define for $\delta>0$ the trial states
\begin{align*}
\Psi_d:=Z_d^{-1}(1 \! - \! G)\chi_\delta\, \Gamma_d,
\end{align*}
where $\chi_\delta:=\mathds{1} \! \left(\mathcal{N}\leq N^{\frac{3\kappa}{2}+\delta}\right)=\mathds{1} \! \left(\sum_{k\neq 0}\mathfrak{a}^\dagger_k \mathfrak{a}_k\leq N^{\frac{3\kappa}{2}+\delta}\right)$ and $\Gamma_d$ is the $d$-th eigenvector of the operator $\sum_k e_k\mathfrak{b}_k^\dagger \mathfrak{b}_k$, i.e. $\Gamma_d$ is the eigenvector of $\sum_k e_k\mathfrak{b}_k^\dagger \mathfrak{b}_k$ corresponding to the eigenvalue $\widetilde{\lambda}_{N,\kappa}^{(d)}$ defined above Eq.~(\ref{Eq:Space_Change}), and $Z_j$ is a normalization constant. Note that $\Psi_d\in \mathfrak{h}$ is an element of $\mathcal{F}^{+}_{N-1}\subset L^2_\mathrm{sym} \! \left(\Lambda^N\right)$ for $N$ large enough, and therefore an appropriate trial state. The following Lemma \ref{Lem:Exponential_Decay} demonstrates that the eigenstates $\Gamma_d$ are already in the range of the projections $\chi_\delta$ up to an error that decays faster than any power in $\frac{1}{N}$.
\begin{lem}
    \label{Lem:Exponential_Decay}
    For any $\lambda\in \mathbb{R}$, $m\in \mathbb{N}$ and $\delta>0$, there exists a constant $C>0$ such that
    \begin{align}
    \label{Eq:Exponential_Decay}
        \bigg\langle (1 \! - \! \chi_\delta) \Gamma_d,\Big(\sum_{k\neq 0}\mathfrak{b}^\dagger_k \mathfrak{b}_k\Big)^m   \mathbb{K} \, (1 \! - \! \chi_\delta) \Gamma_d\bigg\rangle\leq C N^{-\lambda},
    \end{align}
   where $\mathbb{K}$ is defined in Lemma \ref{Lem:Potential_Epsilon_Tilde}. Furthermore $|Z_d-1|\leq C' N^{\frac{7\kappa}{4}-\frac{1}{2}}$ for a suitable $C'$.
\end{lem}
\begin{proof}
    In the following let $\widetilde{\mathbb{K}}:=\sum_k |k|^2 \mathfrak{a}_k^\dagger  \mathfrak{a}_k+\frac{1}{2}\sum_{k \ell, m n}(V_{N^{1-\kappa}})_{k \ell , m n} \mathfrak{a}_k \mathfrak{a}_\ell \mathfrak{a}_m \mathfrak{a}_n$ and note 
    \begin{align*}
      \Big(\sum_{k\neq 0}\mathfrak{b}^\dagger_k \mathfrak{b}_k\Big)^m \mathbb{K} \lesssim N^{\gamma_m} \Big(\sum_{k\neq 0}\mathfrak{a}^\dagger_k \mathfrak{a}_k+1\Big)^{m} \left(\widetilde{\mathbb{K}}+1\right)=:N^{\gamma_m} \mathcal{M}
    \end{align*}
  by Lemma \ref{Lem:True_Particle_Number_Comparison}, with $\gamma_m:=\frac{3(m+1)\kappa}{2}+1$. Since $1-\chi_\delta\leq N^{-\ell  \left(\frac{3\kappa}{2}+\delta\right)}\Big(\sum_{k\neq 0}\mathfrak{a}^\dagger_k \mathfrak{a}_k\Big)^\ell$ for any $\ell$, and both $\chi_\delta$ and $\mathcal{N}$ commute with the non-negative operator $\mathcal{M}$, we have $ (1 \! - \! \chi_\delta) \mathcal{M} (1 \! - \! \chi_\delta) \lesssim  N^{-\ell  \left(3\kappa+2\delta\right)} \Big(\sum_{k\neq 0}\mathfrak{a}^\dagger_k \mathfrak{a}_k+1\Big)^{2\ell}\mathcal{M}  $, and therefore $\bigg\langle (1 \! - \! \chi_\delta) \Gamma_d,\Big(\sum_{k\neq 0}\mathfrak{b}^\dagger_k \mathfrak{b}_k\Big)^m   \mathbb{K} \, (1 \! - \! \chi_\delta) \Gamma_d\bigg\rangle$ is bounded from above by
    \begin{align*}
       &\ \ \ \ \ N^{\gamma_m-\ell  \left(3\kappa+2\delta\right)} \Big\langle \Gamma_d,\Big(\sum_{k\neq 0}\mathfrak{a}^\dagger_k \mathfrak{a}_k+1\Big)^{2\ell+m} \left(\widetilde{\mathbb{K}}+1\right) \Gamma_d\Big\rangle\\
       & \lesssim N^{2\gamma_m-2\ell\delta}\bigg\langle \Gamma_d,\Big(\sum_{k\neq 0}\mathfrak{b}^\dagger_k \mathfrak{b}_k+1\Big)^{2\ell+m} \left(\mathbb{K}+1\right) \Gamma_d\bigg\rangle\lesssim N^{2\gamma_m-2\ell\delta},
    \end{align*}
    where we used Lemma \ref{Lem:True_Particle_Number_Comparison} and that $\bigg\langle \Gamma_d,\Big(\sum_{k\neq 0}\mathfrak{b}^\dagger_k \mathfrak{b}_k+1\Big)^{2\ell+m+1} \left(\mathbb{K}+1\right) \Gamma_d\bigg\rangle$ is finite since $\Gamma_d$ is an eigenstate. Choosing $2\ell\delta\geq 2\gamma_m+\lambda$ concludes the proof of Eq.~(\ref{Eq:Exponential_Decay}). Regarding the estimate on $Z_d$, note that $Z_d^2=\|\chi_\delta \Gamma_d\|^2-2\mathfrak{Re}\braket{\chi_\delta\Gamma_d,G\chi_\delta \Gamma_d}+\braket{\chi_\delta\Gamma_d,G^\dagger G\chi_\delta \Gamma_d}$. Making use of the operator inequality $G^\dagger G \lesssim N^{\frac{7\kappa}{2}-1}\! \left(\widetilde{\mathcal{N}}+1\right)^2$ therefore yields
    \begin{align*}
        |Z_d^2-1|\leq 1-\|\chi_\delta \Gamma_d\|^2+N^{\frac{7\kappa}{4}-\frac{1}{2}}\! \left(\|\chi_\delta \Gamma_d\|^2+\braket{\chi_\delta \Gamma_d,\! \left(\widetilde{\mathcal{N}}+1\right)^2 \! \chi_\delta \Gamma_d}\right).
    \end{align*}
    By Eq.~(\ref{Eq:Exponential_Decay}) we know that $1-\|\chi_\delta \Gamma_d\|^2\lesssim N^{-\lambda}$ for any $\lambda>0$ and $\braket{\chi_\delta \Gamma_d,\! \left(\widetilde{\mathcal{N}}+1\right)^2 \! \chi_\delta \Gamma_d}\lesssim \braket{ \Gamma_d,\! \big(\sum_{k\neq 0}\mathfrak{b}^\dagger_k \mathfrak{b}_k+1\big)^2 \Gamma_d}+N^{-\lambda}\lesssim 1$, which concludes the proof. 
\end{proof}

In order to compute the energy of $\Psi_d$, we are first going to combine the results from Subsection \ref{Subsection:Annihilation_I} and \ref{Subsection:Annihilation_II}, as well as the results in Section \ref{Sec:Beyond_GP}, yielding the following Theorem \ref{Th:Bounds_Without_Cut_Off}.
\begin{thm}
\label{Th:Bounds_Without_Cut_Off}
   For $\kappa<\frac{2}{13}$, let us define the operator $W_\mathrm{log}:=-\frac{1}{2}\sum_{p,k}e_k \nu_k \sqrt{N}f_{-k,p}\mathfrak{b}^\dagger_{-(k+p)}\mathfrak{b}_k^\dagger \mathfrak{b}^\dagger_p$ and the ($N$ dependent) constant $C_\mathrm{log}:=\sum_{k} \Pi_k-D_\mathrm{log}$, where $\Pi_k$ is defined in Corollary \ref{Corollary:Annihilation_I}, $D_\mathrm{log}:=\sum_{k,k'}\nu_{k'}\nu_k \gamma_{k'}\gamma_k \Lambda_{k (-k), k' (-k')}$ and $\Lambda$ is introduced above Eq.~(\ref{Eq:Representation_Of_Error_I}). Then we have
    \begin{align}
    \label{Eq:Bounds_Without_Cut_Off}
  (1 \! - \! G)^\dagger H_{N,\kappa}' (1 \! - \! G)\Big|_{\mathcal{F}^{+}_{N-3}}=\left(C_\mathrm{log}+ \mathbb{K} + W_\mathrm{log}+W_\mathrm{log}^\dagger + X\right)\Big|_{\mathcal{F}^{+}_{N-3}},
    \end{align}
    where $X$ is an operator satisfying $\pm X\Big|_{\mathcal{F}^+_{rN}}\leq  C_r  N^{\frac{11\kappa}{4}-\frac{1}{2}} \! \left(\widetilde{\mathcal{N}}+1\right)^6 \! \Big(\mathbb{K}+1\Big)\Big|_{\mathcal{F}^+_{rN}}$ for $0<r<1$. Furthermore $ \big| \! \left\langle \chi_\delta  \Gamma_1,X \, \chi_\delta \Gamma_1\right\rangle \! \big|  \lesssim \frac{\sqrt{log N}}{N}$ as well as $ \big| \! \left\langle \chi_\delta  \Gamma_1,W_\mathrm{log} \, \chi_\delta \Gamma_1\right\rangle \! \big| \lesssim \frac{1}{N}$ in the case $\kappa=0$, and $|D_\mathrm{log}|\lesssim N^{\frac{9\kappa}{2}}\frac{\log N}{N}$.
\end{thm}
\begin{proof}
We first derive the estimate on $D_\mathrm{log}$. By Lemma \ref{Lem:Coefficient_Comparison} we have the estimates $|\nu_{k'}|\lesssim \frac{N^{\kappa}}{|k'|^2}$ and $|\Lambda_{k (-k), k' (-k')}|\lesssim \frac{N^{2\kappa-1}}{|k|^2+|k'|^2}$. Combining this with $|\gamma_k|\lesssim N^{\frac{\kappa}{4}}$ yields
\begin{align*}
& |D_\mathrm{log}| \! \lesssim  \! \frac{N^{\frac{7\kappa}{2}}}{N}\sum_{k,k'}|\nu_k|\frac{1}{|k'|^2 (|k|^2 \! + \! |k'|^2)} \! \lesssim  \! \frac{N^{\frac{7\kappa}{2}}}{N}\sum_{k}|\nu_k|\int_{\mathbb{R}^3}\frac{\mathrm{d}k'}{|k'|^2 (|k|^2 \! + \! |k'|^2)} \! =  \! \frac{N^{\frac{7\kappa}{2}}}{N}\mu \sum_{k}|k|^{-1} |\nu_k|
\end{align*}
with $\mu:=\int_{\mathbb{R}^3}\frac{\mathrm{d}k'}{|k'|^2 (1+|k'|^2)}<\infty$. Furthermore note that we have
\begin{align*}
\frac{N^{\frac{7\kappa}{2}}}{N}\mu\sum_{0<|k|\leq N}|k|^{-1} |\nu_k|\lesssim \frac{N^{\frac{9\kappa}{2}}}{N}\sum_{0<|k|\leq N}|k|^{-3}\lesssim \frac{N^{\frac{9\kappa}{2}}}{N}\int_{1}^N \frac{\mathrm{d}r}{r}=N^{\frac{9\kappa}{2}}\frac{\log N}{N},
\end{align*}
as well as $\sum_{|k|>N}|k|^{-1} |\nu_k|\leq \sqrt{\sum_{|k|>N}|k|^{-4}}\sqrt{\sum_{k}|k|^{2}|\nu_k|^2}\lesssim N^{\frac{\kappa}{2}}$ by Lemma \ref{Lem:Increased_Decay}, which provides the desired bound on $D_\mathrm{log}$. In order to archive the decomposition in Eq.~(\ref{Eq:Bounds_Without_Cut_Off}), recall the definition of $\Xi_k$ and $\Xi_{jk}$ from Corollary \ref{Corollary:Annihilation_I} and Corollary \ref{Corollary:Annihilation_II}. By Eq.~(\ref{Eq:Full_Representation}) we can write
    \begin{align*}
&   \ \   \ \   \ \   \ \   \ \   \ \   \ \   \ \  (1  -  G)^\dagger \left(H_{N,\kappa}' +\sum_{i=1}^4 \mathcal{E}_i\right)(1  -  G)\\
&= \sum_{k\neq 0}e_k \left(\mathfrak{b}_k-\Theta_k+\Xi_k\right)^\dagger \left(\mathfrak{b}_k-\Theta_k+\Xi_k\right)+\sum_{j k,mn\neq 0} \!  \!  \left( V_{N^{1-\kappa}}\right)_{j k,mn}\left(\mathfrak{b}_k\mathfrak{b}_j+\Xi_{jk}\right)^\dagger  \left(\mathfrak{b}_m\mathfrak{b}_n+\Xi_{mn}\right).  
    \end{align*}
    With the definitions $\widetilde{X}_1:  =\sum_{k\neq 0}e_k \Xi_k^\dagger\, \Xi_k+\sum_{j k,mn\neq 0} \!  \!  \left( V_{N^{1-\kappa}}\right)_{j k,mn}\Xi_{jk}^\dagger  \, \Xi_{mn}- \! \left(Y \! + \! Y^\dagger \right)$, where $Y:=\sum_{k\neq 0}e_k \Xi_k^\dagger \, \Theta_k$ and $\widetilde{X}_2:=\sum_{k\neq 0}e_k \Xi_k^\dagger \mathfrak{b}_k+\sum_{j k,mn\neq 0} \!  \!  \left( V_{N^{1-\kappa}}\right)_{j k,mn}\Xi_{jk}^\dagger  \, \mathfrak{b}_m \mathfrak{b}_n$, we obtain
    \begin{align}
    \label{Eq:Pre_Split}
        (1  -  G)^\dagger \left(H_{N,\kappa}' +\sum_{i=1}^4 \mathcal{E}_i\right)(1  -  G)=\sum_{k\neq 0}A_k+\mathbb{K}+ W_\mathrm{log}+W_\mathrm{log}^\dagger +\widetilde{X}_1+\widetilde{X}_2+\widetilde{X}_2^\dagger,
    \end{align}
    where we used that $\sum_{k\neq 0}e_k \left(\mathfrak{b}_k-\Theta_k\right)^\dagger \left(\mathfrak{b}_k-\Theta_k\right)+\sum_{j k,mn\neq 0} \!  \!  \left( V_{N^{1-\kappa}}\right)_{j k,mn}\mathfrak{b}_k^\dagger \mathfrak{b}_j^\dagger  \mathfrak{b}_m\mathfrak{b}_n$ is identical to $\sum_{k\neq 0}A_k+\mathbb{K}+ W_\mathrm{log}+W_\mathrm{log}^\dagger$. Let us furthermore define $X:=X_1+X_2$ with
    \begin{align*}
        X_1: & =\widetilde{X}_1+(1  -  G)^\dagger\left(\mathcal{E}_\mathrm{log}+\mathcal{E}_\mathrm{odd}-\sum_{i=1}^4 \mathcal{E}_i\right)(1  -  G),\\
        X_2: & =\widetilde{X}_2+\widetilde{X}_2^\dagger + D_\mathrm{log}-(1  -  G)^\dagger \left( \mathcal{E}_\mathrm{log}+\mathcal{E}_\mathrm{odd}\right)(1  -  G),
    \end{align*}
    where $\mathcal{E}_\mathrm{log}:=\sum_{k \ell, k' \ell'}\Lambda_{k \ell, k' \ell'}\, a_{\ell'}^\dagger a_{k'}^\dagger  \frac{a_0^\dagger a_0}{N} a_k a_\ell$ and $\mathcal{E}_{\mathrm{odd}}:=\sum_{i=1}^3\left(\sum_{\ell,k}\Upsilon^{(i)}_{\ell,k} O_i\, a_k^\dagger a_\ell^\dagger a_{k+\ell}+\mathrm{H.c.}\right)$, with $\Lambda$, $\Upsilon^{(i)}$ and $O_i$ introduced above Eq.~(\ref{Eq:Representation_Of_Error_I}). Using Eq.~(\ref{Eq:Pre_Split}) it is easy to check that $X$ allows a decomposition of $(1  -  G)^\dagger H_{N,\kappa}' (1  -  G)$ according to Eq.~(\ref{Eq:Bounds_Without_Cut_Off}). As a consequence of Corollary \ref{Corollary:Annihilation_I} and Corollary \ref{Corollary:Annihilation_II} $\pm \left(\widetilde{X}_2+\widetilde{X}_2^\dagger\right)\big|_{\mathcal{F}^+_{N-3}}\lesssim N^{\frac{11\kappa}{4}-\frac{1}{2}} \! \left(\widetilde{\mathcal{N}}+1\right)^6 \! \Big(\mathbb{K}+1\Big)\big|_{\mathcal{F}^+_{N-3}}$, and by Lemma \ref{Lem:Upsilon_Control} we have
    \begin{align*}
       \pm (1  -  G)^\dagger \mathcal{E}_\mathrm{odd}(1  -  G)\Big|_{\mathcal{F}^+_{rN}}\lesssim N^{\frac{13\kappa}{4}-\frac{1}{2}}(1  -  G)^\dagger \left(\widetilde{\mathcal{N}}+1\right)(1  -  G)\Big|_{\mathcal{F}^+_{rN}}\lesssim N^{\frac{13\kappa}{4}-\frac{1}{2}}\left(\widetilde{\mathcal{N}}+1\right)\Big|_{\mathcal{F}^+_{rN}}.
    \end{align*}
In order to an analyse $\mathcal{E}_{\mathrm{log}}$, define $\alpha_\delta:=\sup_p \left\{\sum_k |k|^{1-2\delta} \big(\sqrt{N}f_{p,k}\big)^2\right\}\lesssim N^{2\kappa-1}$ and observe
  \begin{align}
  \label{Eq:G-Sandwich}
    & G^\dagger\!  \left(\widetilde{\mathcal{N}} \! + \! 1\right) \!  \! \left(\sum_k |k|^{1-2\delta} b_k^\dagger b_k \! + \! 1\right)\! G \lesssim \alpha_{-\frac{1}{2}} \sum_p \mathfrak{a}_p^\dagger \left(\widetilde{\mathcal{N}} \! + \! 1\right)^3 \!  \! \left(\sum_k |k|^{1-2\delta} b_k^\dagger b_k \! + \! 1\right)\mathfrak{a}_p\\
    \nonumber
     &  \ \ \   + \alpha_\delta \sum_p \mathfrak{a}_p^\dagger \left(\widetilde{\mathcal{N}} \! + \! 1\right)^3 \! \mathfrak{a}_p \lesssim N^{4\kappa-1}\left(\widetilde{\mathcal{N}} \! + \! 1\right)^4 \!  \! \left(\sum_k |k|^{1-2\delta} b_k^\dagger b_k \! + \! 1\right).
  \end{align}    
By Lemma \ref{Lem:Lambda_Control} we therefore have $\pm (1  -  G)^\dagger \mathcal{E}_\mathrm{odd}(1  -  G)\lesssim N^{\frac{11\kappa}{2}+3\delta-1}\left(\widetilde{\mathcal{N}} \! + \! 1\right)^4 \!  \! \left(\sum_k |k|^{1-2\delta} b_k^\dagger b_k \! + \! 1\right)$. Choosing $\delta$ small enough consequently yields $
\pm X_2\Big|_{\mathcal{F}^+_{rN}}\lesssim N^{\frac{11\kappa}{4}-\frac{1}{2}} \! \left(\widetilde{\mathcal{N}}+1\right)^6 \! \Big(\mathbb{K}+1\Big)\Big|_{\mathcal{F}^+_{rN}}$. We further have $\pm \widetilde{X}_1\Big|_{\mathcal{F}^+_{N-3}}\lesssim N^{\frac{11\kappa}{2}-1}\sqrt{\log N} \! \left(\widetilde{\mathcal{N}}+1\right)^6 \! \Big(\mathbb{K}+1\Big)\Big|_{\mathcal{F}^+_{N-3}}$ by Corollary \ref{Corollary:Annihilation_I} and Corollary \ref{Corollary:Annihilation_II}. Combining the results from Eq.~(\ref{Eq:Representation_Of_Error_I}), Eq.~(\ref{Eq:Representation_Of_Error_II}), Lemma \ref{Lem:Universal_Error_Estimate} and Lemma \ref{Lem:Particle_Number_Comparison}, and using the estimates $\pm \mathcal{E}_2\lesssim N^{4\kappa-1} \! \left(\mathcal{N}+1\right)^2$ and $\pm \sum_{k\neq 0}|k|^2 w_k^2 N^{-2}\left(a_0^{2\dagger }a_0^2-N^2\right) a_k^\dagger a_k\lesssim N^{5\kappa-1} \! \left(\widetilde{\mathcal{N}}+1\right)^2$, which are similar to the ones already obtained in Subsection \ref{Subsection:Proof of Theorem_Beyond}, yields
    \begin{align*}
        & \pm (1-G)^\dagger \left(\mathcal{E}_\mathrm{odd}+\mathcal{E}_\mathrm{log}-\sum_{i=1}^4 \mathcal{E}_i\right)(1-G)\bigg|_{\mathcal{F}^+_{rN}}\lesssim N^{6\kappa-1}\! (1-G)^\dagger \left(\widetilde{\mathcal{N}}+1\right)^3 (1-G) \Big|_{\mathcal{F}^+_{rN}}\\
        &  \ \ \ \  \ \ \ \  \lesssim N^{6\kappa-1}\!\left(\widetilde{\mathcal{N}}+1\right)^3 \Big|_{\mathcal{F}^+_{rN}}\lesssim N^{\frac{11\kappa}{2}-1}\! \left(\widetilde{\mathcal{N}}+1\right)^2 \! \Big(\mathbb{K}+1\Big)\Big|_{\mathcal{F}^+_{rN}}.
    \end{align*}
Hence $\pm X_1\Big|_{\mathcal{F}^+_{rN}}\lesssim N^{\frac{11\kappa}{2}-1}\sqrt{\log N} \! \left(\widetilde{\mathcal{N}}+1\right)^6 \! \Big(\mathbb{K}+1\Big)\Big|_{\mathcal{F}^+_{rN}}$, which concludes together with the corresponding bound on $X_2$ the proof of $\pm X\Big|_{\mathcal{F}^+_{rN}}\leq  C_r  N^{\frac{11\kappa}{4}-\frac{1}{2}} \! \left(\widetilde{\mathcal{N}}+1\right)^6 \! \Big(\mathbb{K}+1\Big)\Big|_{\mathcal{F}^+_{rN}}$. 

In the following let $\kappa:=0$. Using again the bound on $X_1$, we obtain
\begin{align*}
|\braket{\chi_\delta \Gamma_1, X_1 \chi_\delta \Gamma_1}|\lesssim \frac{\sqrt{\log N}}{N} \! \left\langle \chi_\delta \Gamma_1, \left(\widetilde{\mathcal{N}}+1\right)^6 \! \Big(\mathbb{K}+1\Big) \chi_\delta \Gamma_1\right\rangle\lesssim \frac{\sqrt{\log N}}{N},
\end{align*}
where we have used that $\mathfrak{b}_k \Gamma_1=0$, and therefore $\left(\widetilde{\mathcal{N}}+1\right)^6 \! \Big(\mathbb{K}+1\Big)\Gamma_1=\Gamma_1$, together with Lemma \ref{Lem:Exponential_Decay}. Regarding $ \left\langle \chi_\delta  \Gamma_1,X_2 \, \chi_\delta \Gamma_1\right\rangle$, note $\pm \left(\widetilde{X}_2+\widetilde{X}_2^\dagger \right)\Big|_{\mathcal{F}^+_{N-3}}\lesssim \epsilon \left(\widetilde{\mathcal{N}}+1\right)^4\left(\mathbb{K}+1\right)\Big|_{\mathcal{F}^+_{N-3}}+\epsilon^{-1}\mathbb{K}\Big|_{\mathcal{F}^+_{N-3}}$ by Corollary \ref{Corollary:Annihilation_I} and Corollary \ref{Corollary:Annihilation_II}. Since $\mathbb{K}\Gamma_1=0$, we obtain $ \left\langle \chi_\delta  \Gamma_1,\mathbb{K} \, \chi_\delta \Gamma_1\right\rangle\lesssim N^{-2\lambda}$ for any $\lambda>0$ by Lemma \ref{Lem:Exponential_Decay}, and consequently $\Big|  \! \left\langle \chi_\delta  \Gamma_1,\widetilde{X}_2 \, \chi_\delta \Gamma_1\right\rangle \! \Big|\lesssim N^{-\lambda}$. By a similar argument we obtain  $\left\langle \chi_\delta  \Gamma_1,W_\mathrm{log} \, \chi_\delta \Gamma_1\right\rangle\lesssim N^{-\lambda}$. In order to estimate the $\mathcal{E}_\mathrm{odd}$ term, note that $\Gamma_1$ is a quasi-free state with respect to $\mathfrak{b}_k$, respectively $\mathfrak{a}_k$, and $\mathcal{E}_\mathrm{odd}$ contains only odd powers in $\{a_k:k\in 2\pi\mathbb{Z}^3\setminus \{0\}\}$. Therefore $ e^{i\pi \mathcal{N}}\chi_\delta \Gamma_1=\chi_\delta e^{i\pi \sum_{k\neq 0}\mathfrak{a}_k^\dagger \mathfrak{a}_k}\Gamma_1=\chi_\delta \Gamma_1$ is invariant under the transformation $e^{i\pi \mathcal{N}}$, while $e^{-i\pi \mathcal{N}} \mathcal{E}_\mathrm{odd} e^{i\pi \mathcal{N}}=-\mathcal{E}_\mathrm{odd}$, and therefore $\left\langle \chi_\delta \Gamma_1, \mathcal{E}_\mathrm{odd}\chi_\delta \Gamma_1\right\rangle=0$, and the same holds for $G^\dagger \mathcal{E}_\mathrm{odd} G$. Hence
\begin{align*}
   \left\langle \chi_\delta \Gamma_1, (1  -  G)^\dagger \mathcal{E}_\mathrm{odd}(1  -  G)\chi_\delta \Gamma_1\right\rangle= -2\mathfrak{Re}\left\langle \chi_\delta \Gamma_1, \mathcal{E}_\mathrm{odd}G\chi_\delta\Gamma_1\right\rangle.
 \end{align*}
 A rough, but easy, estimate shows that $G^\dagger G\lesssim N^{-1}\left(\widetilde{\mathcal{N}}+1\right)^3$ and $\mathcal{E}_\mathrm{odd}^2\lesssim N^{-1}\left(\widetilde{\mathcal{N}}+1\right)^3$ in the case $\kappa=0$, and hence $\pm \left(\mathcal{E}_\mathrm{odd}G+G^\dagger \mathcal{E}_\mathrm{odd}\right)\lesssim N^{-1}\left(\widetilde{\mathcal{N}}+1\right)^3$. Therefore we obtain that $\Big| \! \left\langle \chi_\delta \Gamma_1, (1  -  G)^\dagger \mathcal{E}_\mathrm{odd}(1  -  G)\chi_\delta \Gamma_1\right\rangle \! \Big|\lesssim \frac{1}{N}$. In order to analyse the final term $\mathcal{E}_\mathrm{log}$, note that we have by Lemma \ref{Lem:Lambda_Control} together with Eq.~(\ref{Eq:G-Sandwich}) for $\epsilon>0$
 \begin{align*}
 & \ \ \ \ \ \ \ \ \ \pm \left\{(1-G)^\dagger\mathcal{E}_{\mathrm{log}}(1-G)-\mathcal{E}_\mathrm{log}\right\}\\
 & \lesssim \epsilon N^{3\delta-1} \left(\widetilde{\mathcal{N}} \! + \! 1\right) \!  \! \left(\sum_k |k|^{1-2\delta} b_k^\dagger b_k \! + \! 1\right)+\frac{2}{\epsilon}N^{3\delta-1}G^\dagger \left(\widetilde{\mathcal{N}} \! + \! 1\right) \!  \! \left(\sum_k |k|^{1-2\delta} b_k^\dagger b_k \! + \! 1\right)G\\
 &\lesssim \left(\epsilon N^{3\delta-1}+\frac{2}{\epsilon}N^{3\delta-2}\right)\left(\widetilde{\mathcal{N}} \! + \! 1\right) \!  \! \left(\sum_k |k|^{1-2\delta} b_k^\dagger b_k \! + \! 1\right)\lesssim \frac{1}{N}\left(\widetilde{\mathcal{N}} \! + \! 1\right)\left(\mathbb{K} \! + \! 1\right).
 \end{align*}
 Furthermore note that we have $\mathcal{E}_\mathrm{log}=\sum_{k \ell, k' \ell'}\Lambda_{k \ell, k' \ell'}\, \mathfrak{a}_{\ell'}^\dagger \mathfrak{a}_{k'}^\dagger  \widetilde{O} \mathfrak{a}_k \mathfrak{a}_\ell$, where the operator $\widetilde{O}:=\left(a_0^\dagger \frac{1}{\sqrt{a_0 a_0^\dagger }}\right)^2\frac{a_0^\dagger a_0}{N}\left(\frac{1}{\sqrt{a_0 a_0^\dagger }}a_0\right)^2$ satisfies $\pm (\widetilde{O}-1)\lesssim \frac{\mathcal{N}}{N}$ and $\Lambda$ satisfies $\|\Lambda\|\lesssim 1$. Therefore $\pm \left\{\mathcal{E}_\mathrm{log}-\mathcal{E}'_\mathrm{log}\right\}\lesssim \frac{1}{N}\mathcal{N}^3\lesssim \frac{1}{N}(\widetilde{\mathcal{N}}+1)^3$ with $\mathcal{E}'_{\mathrm{log}}:=\sum_{k \ell, k' \ell'}\Lambda_{k \ell, k' \ell'}\, \mathfrak{a}_{\ell'}^\dagger \mathfrak{a}_{k'}^\dagger   \mathfrak{a}_k \mathfrak{a}_\ell$, and hence
 \begin{align*}
 \pm \Big\langle \chi_\delta \Gamma_1, \Big\{(1-G)^\dagger\mathcal{E}_{\mathrm{log}}(1-G)-\mathcal{E}'_\mathrm{log}\Big\} \chi_\delta \Gamma_1\Big\rangle \lesssim \frac{1}{N}.
 \end{align*}
 Using again that $\|\Lambda\|\lesssim 1$ and Lemma \ref{Lem:Exponential_Decay} further yields 
\begin{align*}
\left| \! \Big\langle \chi_\delta \Gamma_1, \mathcal{E}'_{\mathrm{log}} \chi_\delta \Gamma_1  \!  \Big\rangle \! -\! \Big\langle \Gamma_1, \mathcal{E}'_{\mathrm{log}} \Gamma_1  \! \Big\rangle \! \right|\! \lesssim\!  \frac{1}{N} \braket{\Gamma_1, \! (\mathcal{N}\! +\! 1)^2 \Gamma_1}\! +\! N \! \braket{(1\! -\! \chi_\delta)\Gamma_1, \! (\mathcal{N}\! +\! 1)^2 (1\! -\! \chi_\delta)\Gamma_1}\! \lesssim \! \frac{1}{N}.
\end{align*} 
Finally we note that $\mathfrak{a}_k \mathfrak{a}_\ell \Gamma_1=-\gamma_k\nu_k \delta_{\ell+k=0} \Gamma_1+\nu_k \nu_\ell \mathfrak{b}_k \mathfrak{b}_\ell \Gamma_1$, which allows us to compute $\Big\langle \Gamma_1, \mathcal{E}'_{\mathrm{log}} \Gamma_1  \! \Big\rangle=D_\mathrm{log}+\sum_{k \ell}\Lambda_{k \ell, k \ell}\nu_k^2 \nu_\ell^2$. Combining what we have so far, and using that $|\sum_{k \ell}\Lambda_{k \ell, k \ell}\nu_k^2 \nu_\ell^2|\lesssim \frac{1}{N}$ and $\left|D_\mathrm{log}-D_\mathrm{log}\|\chi_\delta \Gamma_1\|^2\right|\lesssim N^{-\lambda}$ for $\lambda>0$, we obtain 
\begin{align*}
\left| \! \Big\langle \chi_\delta \Gamma_1, \{(1-G)^\dagger \mathcal{E}_{\mathrm{log}}(1-G)-D_\mathrm{log}\} \chi_\delta \Gamma_1  \!  \Big\rangle \right|\lesssim \frac{1}{N}.
\end{align*}
This concludes the proof of $ \big| \! \left\langle \chi_\delta  \Gamma_1,X \, \chi_\delta \Gamma_1\right\rangle \! \big|  \lesssim \frac{\sqrt{log N}}{N}$.
\end{proof}
With Theorem \ref{Th:Bounds_Without_Cut_Off} at hand, we are in a position to verify Theorem \ref{Th:Upper_Bound} by writing 
\begin{align}
\label{Eq:Upper_Bound_Starting_Point}
   & \ \  \ \ \braket{\Psi_d, H_{N,\kappa}' \Psi_d}=Z_d^{-2} \left\langle \chi_\delta \Gamma_d, \left(C_\mathrm{log}+ \mathbb{K} + \! \left(W_\mathrm{log}+W_\mathrm{log}^\dagger\right) \! + X_1 +X_2\right) \chi_\delta \Gamma_d\right \rangle\\
   \nonumber
   &=\lambda_{N,\kappa}^{(d)} \! + \! \left(Z_d^{-2} \left\langle \chi_\delta \Gamma_d, \mathbb{K} \chi_\delta \Gamma_d\right \rangle  \! -  \! \lambda_{N,\kappa}^{(d)}\right) \! + \! Z_d^{-2} \left\langle \chi_\delta \Gamma_d, \left(C_\mathrm{log} \! + \! W_\mathrm{log} \! + \! W_\mathrm{log}^\dagger \! + \! X_1 \! + \! X_2\right) \chi_\delta \Gamma_d\right \rangle.
\end{align}
In the following we are going to decompose $Z_d^{-2} \left\langle \chi_\delta \Gamma_d, \mathbb{K} \chi_\delta \Gamma_d\right \rangle  \! -  \! \lambda_{N,\kappa}^{(d)}$ as
\begin{align*}
    Z_d^{-2}\! \Big(\! \!  \left\langle \chi_\delta \Gamma_d, \mathbb{K} \chi_\delta \Gamma_d\right \rangle \! -\!   \left\langle  \Gamma_d, \mathbb{K}  \Gamma_d\right \rangle \! \!  \Big)\! +\! Z_d^{-2}\! \! \left(\left\langle \Gamma_d , \mathbb{K}\Gamma_d\right\rangle \! - \! \widetilde{\lambda}_{N,\kappa}^{(d)}\right)\! +\! Z_d^{-2}\! \! \left( \widetilde{\lambda}_{N,\kappa}^{(d)}\! -\! \lambda^{(d)}_{N,\kappa}\! \right)\! +\! \lambda^{(d)}_{N,\kappa}\! \left(1\! -\! Z_d^{-2}\right).
\end{align*}
The first term $\left\langle \chi_\delta \Gamma_d, \mathbb{K} \chi_\delta \Gamma_d\right \rangle \! -\!   \left\langle  \Gamma_d, \mathbb{K}  \Gamma_d\right \rangle$ decays faster than any power in $N^{-\lambda}$ by Lemma \ref{Lem:Exponential_Decay}. Regarding $\left\langle \Gamma_d , \mathbb{K}\Gamma_d\right\rangle \! - \! \widetilde{\lambda}_{N,\kappa}^{(d)}$, note that $\Gamma_d$ is an eigenvector of $\sum_{k\neq 0}e_k \mathfrak{b}_k^\dagger \mathfrak{b}_k$ to the eigenvalue $\widetilde{\lambda}^{(d)}_{N,\kappa}$, and as such has a finite number of excitations, i.e. there exists a finite set $I_d\subset 2\pi\mathbb{Z}^3\setminus \{0\}$ and a $C_d$, such that $\mathfrak{b}_k \Gamma_d=0$ for $k\notin I$ and $\|\mathfrak{b}_j \mathfrak{b}_k \Gamma_d\|\leq C$, and therefore we obtain
\begin{align}
\label{Eq:Eigen_Argument}
    \pm  \left(\left\langle \Gamma_d , \mathbb{K}\Gamma_d\right\rangle - \widetilde{\lambda}_{N,\kappa}^{(d)}\right)= \pm \sum_{jk, mn}\left(V_{N^{1-\kappa}}\right)_{jk , mn}\braket{\mathfrak{b}_j \mathfrak{b}_k \Gamma_d,\mathfrak{b}_m \mathfrak{b}_n \Gamma_d}\leq C_d^2 |I_d|^4 \|\widehat{V}\|_\infty N^{\kappa-1},
\end{align}
Regarding the third term we have $\left|\widetilde{\lambda}_{N,\kappa}^{(d)}\! -\! \lambda^{(d)}_{N,\kappa}\right|\lesssim N^{2\kappa-1}$, see the estimate above Eq.~(\ref{Eq:Reference_Difference_lambda}), and again by Lemma \ref{Lem:Exponential_Decay}, we have $\left|1\! -\! Z_d^{-2}\right|\lesssim N^{-\lambda}$ for any $\lambda>0$. Therefore
\begin{align*}
    \left| Z_d^{-2} \left\langle \chi_\delta \Gamma_d, \mathbb{K} \chi_\delta \Gamma_d\right \rangle  \! -  \! \lambda_{N,\kappa}^{(d)}\right|\lesssim N^{2\kappa-1}.
\end{align*}
Note that $Z_d^{-2} \left\langle \chi_\delta \Gamma_d, \left(\! X_1 \! + \! X_2\right) \chi_\delta \Gamma_d\right \rangle\lesssim N^{\frac{11\kappa}{4}-\frac{1}{2}}\braket{\Gamma_d,\left(\widetilde{\mathcal{N}}+1\right)^6 \! \Big(\mathbb{K}+1\Big)\Gamma_d}+C_\lambda N^{-\lambda}\lesssim N^{\frac{13\kappa}{4}-\frac{1}{2}}$ for any $\lambda\geq 1-\frac{13\kappa}{4}$ by Theorem \ref{Th:Bounds_Without_Cut_Off} and Lemma \ref{Lem:Exponential_Decay}. Arguing similar as in Eq.~(\ref{Eq:Eigen_Argument}) we further obtain $\big| \! \braket{\Gamma_d, W_\mathrm{log} \Gamma_d}\! \big|\lesssim N^{2\kappa-\frac{1}{2}}$, and therefore $\big| \! Z_d^{-2}\braket{\chi_\delta \Gamma_d, W_\mathrm{log} \chi_\delta \Gamma_d}\! \big|\lesssim N^{2\kappa-\frac{1}{2}}$ by Lemma \ref{Lem:Exponential_Decay}. Regarding the constant $C_\mathrm{log}=\sum_{k} \Pi_k-D_\mathrm{log}$, note that we have $|\sum_k \Pi_k|\lesssim N^{\frac{9\kappa}{2}}\frac{\log N}{N}$ by Corollary \ref{Corollary:Annihilation_I} and $|D_\mathrm{log}|\lesssim N^{\frac{9\kappa}{2}}\frac{\log N}{N}$ by Theorem \ref{Th:Bounds_Without_Cut_Off}, and therefore
 \begin{align}
 \label{Eq:Estimate_on_eigenvalues}
     \braket{\Psi_d, H_{N,\kappa}' \Psi_d}\leq \widetilde{\lambda}_{N,\kappa}^{(d)}+ C N^{\frac{13\kappa}{4}-\frac{1}{2}}.
 \end{align}
 In case of $\kappa=0$ and $d=1$ we can further improve Eq.~(\ref{Eq:Estimate_on_eigenvalues}). Starting again with Eq.~(\ref{Eq:Upper_Bound_Starting_Point}), we make use of the fact that $\braket{\Gamma_1,\mathbb{K}\Gamma_1}=0$ and use the estimates $ \big| \! \left\langle \chi_\delta  \Gamma_1,X \, \chi_\delta \Gamma_1\right\rangle \! \big|  \lesssim \frac{\sqrt{\log N}}{N}$ and $ \big| \! \left\langle \chi_\delta  \Gamma_1,W_\mathrm{log} \, \chi_\delta \Gamma_1\right\rangle \! \big| \lesssim \frac{1}{N}$ from Theorem \ref{Th:Bounds_Without_Cut_Off}, in order to obtain for any $\lambda>0$ and suitable constants $C',C>0$ by Lemma \ref{Lem:Exponential_Decay}
 \begin{align}
  \label{Eq:Estimate_on_ground_state_energy}
     \braket{\Psi_1, H_{N,0}' \Psi_1}\leq Z_d^{-2} C_\mathrm{log}+\braket{\Gamma_1,\mathbb{K}\Gamma_1}+C'\frac{\sqrt{\log N}}{N}+C' N^{-\lambda}\leq C\frac{\log N}{N}.
 \end{align}
 Using the definition $H'_{N,\kappa}:=H_{N,\kappa}-4\pi \mathfrak{a}_{N^{1-\kappa}}N^\kappa (N-1)- \frac{1}{2}\sum_{k\neq 0}\left\{\sqrt{A_k^2-B_k^2}-A_k+C_k\right\}$, Eq.~(\ref{Eq:Estimate_on_eigenvalues}) and Eq.~(\ref{Eq:Estimate_on_ground_state_energy}) immediately yield
\begin{align*}
     \braket{\Psi_d,H_{N,\kappa}\Psi_d} & \leq 4\pi \mathfrak{a}_{N^{1-\kappa}}N^\kappa (N-1)+ \frac{1}{2}\sum_{k\neq 0}\left\{\sqrt{A_k^2-B_k^2}-A_k+C_k\right\}+\widetilde{\lambda}_{N,\kappa}^{(d)}+C N^{\frac{13\kappa}{4}-\frac{1}{2}},\\
      \braket{\Psi_1,H_{N,0}\Psi_1} & \leq 4\pi \mathfrak{a}_{N^{1-\kappa}}N^\kappa (N-1)+ \frac{1}{2}\sum_{k\neq 0}\left\{\sqrt{A_k^2-B_k^2}-A_k+C_k\right\}+C\frac{\log N}{N},
 \end{align*}
 which concludes the proof of Theorem \ref{Th:Upper_Bound} since we have by Eq.~(\ref{Eq:Sum_Comparison}) that
 \begin{align*}
   \left|\sum_{k\neq 0}\left\{\sqrt{A_k^2-B_k^2}-A_k+C_k\right\}-\sum_{k\neq 0}\left\{\sqrt{|k|^4+16\pi \mathfrak{a} N^\kappa |k|^2}+|k|^2+8\pi \mathfrak{a} N^\kappa-\frac{(8\pi \mathfrak{a} N^\kappa)^2}{2|k|^2}\right\}\right|
 \end{align*}
 is of order $N^{4\kappa-1}\log N$ in the case $K:=\infty$.

\appendix

\section{Coefficients of the renormalized Potential}
\label{Appendix:Coefficients}

In this part of the Appendix, we are going to verify essential properties of the box scattering length $\mathfrak{a}_L$ and the renormalized potential $V_{L}-V_{L} R V_{L}$, starting with the proof of $|\mathfrak{a}_L-\mathfrak{a}|\lesssim L^{-1}$ in the following Lemma \ref{Lem:Scattering_Comparison}, which we will subsequently use in order to verify Lemma \ref{Lem:Coefficient_Comparison}.

\begin{lem}
\label{Lem:Scattering_Comparison}
    Let $V\in L^1 \! \left(\mathbb{R}^3\right)$. Then there exists a constant $C>0$, such that $|\mathfrak{a}_L-\mathfrak{a}|\leq \frac{C}{L}$. 
\end{lem}
\begin{proof}
Let $\omega$ be the radial solution of $\left(-2\Delta_x+V\right)\omega=V$ in $\mathbb{R}^3$ satisfying $\omega(x) \underset{|x|\rightarrow \infty}{\longrightarrow}0$, which allows us to express the scattering length $\mathfrak{a}$ as $8\pi \mathfrak{a}:=\int_{\mathbb{R}^3}V(x)\, \mathrm{d}x-\braket{V,\omega}$. Note that $\braket{V,\omega}$ is well-defined, since $\omega$ is in $L^2_{\mathrm{loc}}$ and $V$ has compact support. Furthermore we introduce the re-scaled objects $V_L(x):=L^2 V(L x)$ and $\varphi_L(x):=\varphi(L x)$. In order to verify the Lemma, let us use the solution $\varphi$ in order to create an approximation of the box scattering solution $R V_L\in L^2\! \left(\Lambda^2\right)$, i.e. for a smooth cut-off function $0\leq \chi\leq 1$ satisfying $\mathrm{supp}(\varphi)\subset \mathrm{int}(\Lambda)$ and $\chi(x)=1$ for all $|x|\leq \frac{1}{4}$, let us define $\psi_L(x,y):=\chi(x-y)\varphi_L(x-y)$ as well as $\xi_L(x,y):=2(\Delta \chi)(x-y)\varphi_L(x-y)+4(\nabla \chi)(x-y)\nabla \varphi_L(x-y)$, where $x-y$ is the distance on the torus. The function $\psi_L\in L^2\! \left(\Lambda^2\right)$ then satisfies the differential equation 
 \begin{align*}
    \big(\!  -\! \Delta_2\! +\! V_L\big)\psi_L & =V_L-\xi_L,
 \end{align*}
where $\Delta_2:=\Delta_x + \Delta_y$, see Section \ref{Sec:The two-body Problem}. Using $R \big(\!  -\! \Delta_2\! +\! V_L \big)=1-(1-R V_L)(1-\pi_{\mathcal{H}})$, and $(1-\pi_{\mathcal{H}})\psi_L=\left(\int \chi\varphi_L\right) \! 1$, we can therefore express $R V_L$ as
\begin{align}
\label{Eq:Identity_For_The_Resolvent}
    R V_L=\psi_L+R \xi_L-\left(\int_{\mathbb{R}^3} \chi\varphi_L\right) \! \big(1-R V_L\big).
\end{align}
Since $V$ has compact support, we have $L\braket{V_L,\psi_L}=L\braket{V_L,\varphi_L}=\braket{V,\varphi}$ for $L$ large enough. Consequently we can express $ 8\pi\mathfrak{a}_L$ according to Eq.~(\ref{Eq:Identity_For_The_Resolvent}) as
\begin{align*}
    8\pi\mathfrak{a}_L &  \! =  \! L \braket{V_L,(1 \! - \! R V_L)} \! = \! \int_{\mathbb{R}^3}V \! - \! L\braket{V_L,RV_L} \! = \!  \int_{\mathbb{R}^3}V  \! - \! \braket{V,\varphi} \! - \! L \! \braket{V_L,R \xi_L}  \! + \! 8\pi \mathfrak{a}_L  \! \int_{\mathbb{R}^3} \!  \chi\varphi_L\\
    &= 8 \pi \mathfrak{a}-L \! \braket{V_L,R \xi_L} +8\pi \mathfrak{a}_L  \! \int_{\mathbb{R}^3} \chi\varphi_L,
\end{align*}
or equivalently 
\begin{align}
\label{Eq:Expansion}
8\pi\mathfrak{a}-8\pi\mathfrak{a}_L=\left(1-\frac{1}{1-\int_{\mathbb{R}^3} \chi\varphi_L}\right)8\pi \mathfrak{a}+\frac{1}{1-\int_{\mathbb{R}^3} \chi\varphi_L}L \! \braket{R V_L, \xi_L}.
\end{align}
Note that $0\leq \int_\Lambda \chi\varphi_L \leq \frac{C}{L}$, which can be verified easily utilizing the bounds $\|\widehat{\chi}\|_\infty<\infty$ and
\begin{align*}
    2L |k|^2 |\widehat{\varphi_L}(k)|=2 \left|\frac{k}{L}\right|^2 \left|\widehat{\varphi}\left(\frac{k}{L}\right)\right|=\left|\widehat{V} \! \left(\frac{k}{L}\right)-\left\langle e^{i\frac{k}{L}x}V,\frac{1}{-2\Delta+V}V\right\rangle\right|\lesssim 1,
\end{align*}
where we have used $V\frac{1}{-2\Delta+V}V\leq V$ and $|\braket{e^{i\frac{k}{L}x},Ve^{i\frac{k}{L}x}}|= \int_{\mathbb{R}^3}V<\infty$. Consequently the first term on the right hand side of Eq.~(\ref{Eq:Expansion}) is of the order $\frac{1}{L}$. Regarding the second term, note that the support of $\Delta \chi(x-y)$ and $\nabla \chi(x-y)$ is contained in $\mathrm{int}(\Lambda)\setminus B_{\frac{1}{4}}(0)$, which is disjoint to the support of $V_L$ for $L$ large enough, and therefore $\varphi_L(x-y)=\frac{\mathfrak{a}}{L|x-y|}$ on the support of $\Delta \chi(x-y)$, respectively on the support of $\nabla \chi(x-y)$. This especially means that $\xi_L=\frac{1}{L}\xi$, where $\xi(x,y):=2(\Delta \chi)(x-y)\frac{\mathfrak{a}}{|x-y|}-4(\nabla \chi)(x-y)\frac{(x-y)\mathfrak{a}}{|x-y|^3}$ is a $C^\infty_0$ function. Hence
\begin{align*}
    \left| L\braket{R V_L,\xi_L}\right| \! = \! \left|\left\langle(-\Delta_2)^{-1}R V_L,(-\Delta_2)\xi\right\rangle\right| \! \lesssim \! \|(-\Delta_2)^{-1}R V_L\|.
\end{align*}
In order to compute $\|(-\Delta_2)^{-1}R V_L\|$, note that the coefficients $f_k:=\braket{e^{ik(x-y)},(-\Delta_2)R V_L}=\left(V_L-V_L R V_L\right)_{k(-k), 0 0)}$ satisfy $|f_k|\leq \frac{D}{L}$ for a suitable constant $D$. Therefore we obtain 
\begin{align*}
   \|(-\Delta_2)^{-1}R V_L\|^2=\sum_k \frac{|f_k|^2}{|k|^4}\leq \frac{D}{L^2}\sum_k \frac{1}{|k|^4},
\end{align*}
and consequently $\|(-\Delta_2)^{-1}R V_L\|\leq \frac{C}{L}$ for a suitable constant $C$.
\end{proof}

With the result from Lemma \ref{Lem:Scattering_Comparison} we are finally in a position to verify Lemma \ref{Lem:Coefficient_Comparison}.
\begin{proof}[Proof of Lemma \ref{Lem:Coefficient_Comparison}]
Let us first derive the bound on the coefficients $L\Big(V_{L}-V_{L} R V_{L}\Big)_{k_1 k_2, k_3 k_4}$. Since $V\in L^1 \! \left(\mathbb{R}^3\right)$ and $VRV\geq 0$, it is enough to bound $L\Big(V_{L} R V_{L}\Big)_{k_1 k_2, k_1 k_2}$. For $(k_1,k_2)\in \mathcal{H}$
\begin{align*}
L\Big(V_{L} R V_{L}\Big)_{k_1 k_2, k_1 k_2} \!  \! = \! L\braket{e^{ik_1 x}e^{i k_2 y},\pi_{\mathcal{H}}V_L RV_L \pi_{\mathcal{H}}e^{ik_1 x}e^{i k_2 y}} \! \leq \!  L\braket{e^{ik_1 x}e^{i k_2 y},V_L e^{ik_1 x}e^{i k_2 y}} \! = \! \int_{\mathbb{R}^3} \! V,
\end{align*}
where we have used $\pi_{\mathcal{H}}V_L RV_L \pi_{\mathcal{H}}\leq \pi_{\mathcal{H}}V_L \pi_{\mathcal{H}}$. Regarding the momentum pairs $(k_1,k_2)\in \mathcal{L}$, let us define $X:=\pi_{\mathcal{L}}V_L R V_L \pi_{\mathcal{L}}$ and $\chi(x,y):=\mathds{1}_{B_{L^{-1}R}(0)}(x-y)$, where $R$ is such that $\mathrm{supp}(V)\subset B_R(0)$. Then we obtain by Cauchy-Schwarz
\begin{align}
\nonumber
& V_L R V_L  =\chi V_L R V_L \chi \lesssim \chi \pi_{\mathcal{H}} V_L R V_L \pi_{\mathcal{H}} \chi +\chi \pi_{\mathcal{L}} V_L R V_L\pi_{\mathcal{L}}\chi \\
\label{Eq:Helpful_Estimate_X}
& \ \ \leq \chi \pi_{\mathcal{H}} V_L \pi_{\mathcal{H}} \chi +\chi \pi_{\mathcal{L}} X \pi_{\mathcal{L}}\chi \lesssim  V_L   +\chi \pi_{\mathcal{L}} (V_L+X) \pi_{\mathcal{L}}\chi .
\end{align}
Using $\pi_{\mathcal{L}} \chi \pi_{\mathcal{L}}\leq 2\int_{\Lambda^2} \chi\lesssim L^{-3}$ and $\pi_{\mathcal{L}} V_L \pi_{\mathcal{L}}\leq 2\int_{\Lambda^2} V_L\lesssim L^{-1}$ we especially obtain
\begin{align*}
\|X\|\leq \|\pi_{\mathcal{L}} V_L \pi_{\mathcal{L}}\|+\|\pi_{\mathcal{L}} \chi \pi_{\mathcal{L}}\|^2 \left(\|\pi_{\mathcal{L}} V_L \pi_{\mathcal{L}}\|+\|X\|\right)\lesssim L^{-1}+L^{-6}\|X\|,
\end{align*}
and therefore $\|X\|\lesssim L^{-1}$, which implies $L\Big(V_{L} R V_{L}\Big)_{k_1 k_2, k_1 k_2}\lesssim 1$ for $(k_1,k_2)\in \mathcal{L}$.

In order to verify Eq.~(\ref{Eq:Coefficients_Renormalized_Potential}), let $W(x,y):=e^{i k_1 x}e^{i k_2 y}-e^{i(k_1+k_2)x}$. By Eq.~(\ref{Eq:Helpful_Estimate_X}) we have
\begin{align*}
L \braket{V_L W, R V_L W}\lesssim L \braket{W, V_L W}+\|\pi_{\mathcal{L}}\chi W\|^2\lesssim L \braket{W, V_L W}+L^{-6},
\end{align*}
where we have used that $\|\pi_{\mathcal{L}} V_L \pi_{\mathcal{L}}\|, \|X\|\lesssim L^{-1}$ and $\|\pi_{\mathcal{L}}\chi W\|=\left|\int_{\Lambda^2}\chi W\right|\lesssim L^{-3}$. Further
\begin{align*}
L \braket{W, V_L W}=\int_{\mathbb{R}^3} V( z)|e^{i k L^{-1} z}-1|^2\, \mathrm{d}z\leq L^{-2}|k|^2\int_{\mathbb{R}^3}V(z)|z|^2\, \mathrm{d}z\lesssim L^{-2}|k_2|^2,
\end{align*}
where we have used that $V$ has compact support and $V\in L^1 \! \left(\mathbb{R}^3\right)$. Consequently
    \begin{align}
\nonumber
        &\left| \! L\Big(V_{L} R V_{L}\Big)_{k_1 k_2,k_3 k_4} \! - \! L\Big(V_{L} R V_{L}\Big)_{(k_1+k_2) 0 ,k_3 k_4} \! \right| \! = \! L \left| \! \left\langle V_L W, R V_L e^{ik_3 x}e^{ik_4 y}\right\rangle \! \right| \\
            \label{Eq:Zero_Total_Momentum_Comparison}
        &  \ \ \   \ \ \  \leq  \sqrt{L \braket{V_L, R V_L}}\sqrt{L \braket{V_L W, R V_L W}} \lesssim  \! L^{-1} |k_2|.
    \end{align}
By Eq.~(\ref{Eq:Zero_Total_Momentum_Comparison}) and the fact that $\left|L\Big(V_{L}-V_{L} R V_{L}\Big)_{00,00}-8\pi \mathfrak{a} \right|\lesssim L^{-1}$, see Lemma \ref{Lem:Scattering_Comparison}, we observe that it is enough to verify 
\begin{align}
    \label{Eq:Coefficients_Renormalized_Potential_II}
    \left| \! L\Big(V_{L} R V_{L}\Big)_{k 0,k 0} \! - \! L\Big(V_{L} R V_{L}\Big)_{0 0 ,0 0} \! \right|\lesssim L^{-1}|k|
\end{align}
in order to prove Eq.~(\ref{Eq:Coefficients_Renormalized_Potential}). In order to prove Eq.~(\ref{Eq:Coefficients_Renormalized_Potential_II}), let us introduce the projection $\pi$ onto the momentum pairs $(\ell_1,\ell_2)$ with $\ell_2\neq 0$ as well as the boosted Laplace operator $-\Delta^{(p)}_{2}:=(\frac{1}{i}\nabla_x+p)^2+(\frac{1}{i}\nabla_y)^2$, which satisfies $-\Delta^{(p)}_{2}=e^{-ip x}(-\Delta_2)e^{ip x}$ for $p\in 2\pi \mathbb{Z}^3$. Furthermore let us introduce the pseudo-inverse $R_p$ of the operator $\pi\left(-\Delta^{(p)}_{2}+V_L\right)\pi$, and define the state $\psi:=e^{-ipx}R e^{ipx}V_L$. Using $(-\Delta_2 \! + \! V_L)R=1-(1-\pi_\mathcal{H})(1-V_L R)$, yields the equation
    \begin{align*}
        \pi (-\Delta^{(p)}_{2} \! + \! V_L)\psi  & =  \pi e^{-ipx}(-\Delta_2 \! + \! V_L)R e^{ipx}V_L=  \pi V_L  -  \lambda e^{ip(y-x)}
    \end{align*}
    where $\lambda:=\mathds{1}(|p|<K)(V_L-V_L R V_L)_{0 p ,p 0}$. Using $R_p   (-\Delta^{(p)}_{2} \! + \! V_L)=1-(1-R_p V_L)(1-\pi)$ we therefore obtain $\psi=R_p V_L - \lambda  R_p e^{ip(y-x)}$. By the first part of this proof $|\lambda|\lesssim L^{-1}$ and similarly $|\braket{V_L,R_p e^{i\ell_1 x}e^{i \ell_2 y}}|\lesssim L^{-1}$ for $\ell_1,\ell_2\in 2\pi \mathbb{Z}^3$, and therefore
    \begin{align*}
   \left|L\Big(V_{L} R V_{L}\Big)_{p 0, p 0}-L\braket{V_L, R_p  V_L}\right|=L  \left|\lambda\braket{V_L,R_p e^{ip(y-x)}}\right|\lesssim L^{-1}.
    \end{align*}
    In order to prove Eq.~(\ref{Eq:Coefficients_Renormalized_Potential_II}), it is consequently enough to verify
    \begin{align}
    \label{Eq:Boosted_Laplace}
       \left|\braket{ V_L, R_p  V_L}-\braket{ V_L, R_0  V_L}\right|\lesssim L^{-2}|p|. 
    \end{align}
    In order to show Eq.~(\ref{Eq:Boosted_Laplace}), note that $R_p-R_0=R_0\left(-|p|^2+2i\nabla_x\cdot p\right)R_p$ and therefore
    \begin{align}
    \nonumber
 &\left\langle  V_L,R_p V_L \right\rangle  - \left\langle   V_L,R_0   V_L  \right\rangle  =  \frac{1}{2} \! \left\langle   V_L,  \left( R_p   -  R_0  \right)   V_L  \right\rangle  +  \frac{1}{2}  \left\langle  V_L,  \left( \! R_{-p}   -  R_0  \! \right)   V_L \right\rangle\\
\nonumber
       & \ \  = -\frac{|p|^2}{2}\left\langle  V_L,R_0 \left(R_p+R_{-p}\right) V_L \right\rangle+ \left\langle  V_L,R_0 \left( i\nabla_x\cdot p\right) \left(R_p-R_{-p}\right)  V_L \right\rangle  \\
      \label{Eq:Estimate_With_Symmetry}
       & \ \ = -\frac{|p|^2}{2}\sum_{q\neq 0} \overline{\psi_0(q)}\big(\psi_p(q)+\psi_{-p}(q)\big)+\sum_{q\neq 0} (p\cdot q)\overline{\psi_0(q)}\big(\psi_p(q)-\psi_{-p}(q)\big)
    \end{align}
   with $\psi_p:=R_{p} V_L$. Defining $\eta_p:=-\Delta_2^{(p)}\psi_p$, we note that $\eta_p(q)=L^{-1}\widehat{V}\left(\frac{q}{L}\right)-\braket{e^{iq(x-y)}V_L,R_p V_L}$ for $q\neq 0$ and therefore $|\eta_p(q)|\lesssim \frac{1}{L}$, which can be verified similarly to the bounds on $(V_L R V_L)_{k_1 k_2,k_3 k_4}$ from the first part of this proof. Consequently $|\psi_p(q)|\lesssim \frac{1}{L(|q|^2+|q+p|^2)}$ and
   \begin{align}
   \label{Eq:Appendix_New_Bound_I}
       \frac{|p|^2}{2}\left|\sum_{q\neq 0} \overline{\psi_0(q)}\big(\psi_p(q)+\psi_{-p}(q)\big)\right|\lesssim L^{-2}|p|^2 \sum_{q\neq 0}\frac{1}{2|q|^2 (|q|^2+|q+p|^2)}\lesssim L^{-2}|p|.
   \end{align}
   Regarding the second term in Eq.~(\ref{Eq:Estimate_With_Symmetry}), let us identify $\sum_{q\neq 0} (p\cdot q)\overline{\psi_0(q)}\big(\psi_p(q)-\psi_{-p}(q)\big)$ as
   \begin{align}
   \label{Eq:Appendix_New_Bound_II}
\sum_{q\neq 0} \frac{(p\cdot q)\overline{\eta_0(q)}\eta_p(q)}{2|q|^2}\left(\frac{1}{|q|^2 \! + \! |q \! + \! p|^2} \! - \! \frac{1}{|q|^2 \! + \! |q \! - \! p|^2}\right) \! + \!  \sum_{q\neq 0} \frac{(p\cdot q)\overline{\psi_0(q)}\left(\eta_p(q) \! - \! \eta_{-p}(q)\right)}{|q|^2 \! + \! |q \! - \! p|^2}.
   \end{align}
   Starting with the first term in Eq.~(\ref{Eq:Appendix_New_Bound_II}), let $f(q,p):=\frac{1}{2|q|^2}\left(\frac{1}{|q|^2  +  |q  +  p|^2}  -  \frac{1}{|q|^2  +  |q  -  p|^2}\right)$ and note
   \begin{align*}
       \left|\sum_{q\neq 0}(p\cdot q)\overline{\eta_0(q)}\eta_p(q) f(q,p)\right|\lesssim L^{-2}|p|\sum_{q\neq 0}|q| |f(q,p)|\lesssim L^{-2}|p|\int_{\mathbb{R}^3}|q| |f(q,p)| \, \mathrm{d}q=L^{-2}|p|\mu,
   \end{align*}
with $\mu:=\int_{\mathbb{R}^3}|q| |f(q,e)| \, \mathrm{d}q<\infty$ and $e$ being a unit vector. Regarding the second term in Eq.~(\ref{Eq:Appendix_New_Bound_II}), we define $\psi_{p,q}:=R_{p} e^{iq(x-y)}V_L$, which again satisfies $|\psi_{p,q}(\ell)|\lesssim \frac{1}{L(|\ell|^2+|\ell+p|^2)}$. Using Lemma \ref{Lem:Increased_Decay} and the fact that $R_{-p}-R_p=-4 R_{-p} \left(i\nabla_x\cdot p\right)R_{p}$, we observe that
   \begin{align*}
      & \left|\eta_p(q)  -  \eta_{-p}(q)\right|=\left|\left\langle e^{iq(x-y)}V_L,\left(R_{-p}-R_p\right) V_L\right\rangle\right| =4\left|\sum_{\ell \neq 0}(p\cdot \ell)\overline{\psi_{-p,q}(\ell)}\psi_p(\ell)\right|\\
     & \ \ \leq 4|p| \sqrt{\sum_{\ell\neq 0}|\psi_{-p,q}(\ell)|^2}\sqrt{\sum_{\ell\neq 0}|\ell|^2 |\psi_{p}(\ell)|^2}\lesssim |p| \sqrt{L^{-2}|p|^{-1}}\sqrt{L^{-1}}=L^{-\frac{3}{2}}|p|^{\frac{1}{2}}. 
   \end{align*}
   Consequently we can estimate the second term in Eq.~(\ref{Eq:Appendix_New_Bound_II}), using Lemma \ref{Lem:Increased_Decay} again, by
   \begin{align*}
      & \ \ \  \ \ \  \ \  \ \ \  \ \ \left| \sum_{q\neq 0} \frac{(p\cdot q)\overline{\psi_0(q)}\left(\eta_p(q) \! - \! \eta_{-p}(q)\right)}{|q|^2 \! + \! |q \! - \! p|^2}\right|\lesssim L^{-\frac{3}{2}}|p|^{\frac{3}{2}}\sum_{q\neq 0}\frac{|q||\psi_0(q)|}{|q|^2 \! + \! |q \! - \! p|^2}\\
       & \leq L^{-\frac{3}{2}}|p|^{\frac{3}{2}}\sqrt{\sum_{q\neq 0}\left(|q|^2 \! + \! |q \! - \! p|^2\right)^{-2}}\sqrt{\sum_{q\neq 0}|q|^2 |\psi_0(q)|^2}\lesssim L^{-\frac{3}{2}}|p|^{\frac{3}{2}}\sqrt{|p|^{-1}}\sqrt{L^{-1}}=L^{-2}|p|.
   \end{align*}
\end{proof}

With Lemma \ref{Lem:Coefficient_Comparison} at hand, we can furthermore verify the following Lemma \ref{Lem:Approximate_LHY}.

\begin{lem}
\label{Lem:Approximate_LHY}
    Let $A_k$, $B_k$ and $C_k$ be the coefficients defined above Eq.~(\ref{Eq:Writen_In_d}), and let us define $f_k:=\sqrt{A_k^2-B_k^2}-A_k+C_k$ and $g_k:=\sqrt{|k|^4+16\pi \mathfrak{a} N^\kappa |k|^2}-|k|^2-8\pi \mathfrak{a} N^\kappa+\frac{(8\pi \mathfrak{a} N^\kappa)^2}{2|k|^2}$. Then we have $ \sum_{k\neq 0} \Big|f_k-g_k\Big| \lesssim \frac{N^{4\kappa}\log N}{N}+\frac{N^{3\kappa}}{K}$ for $K\geq C N^{\frac{\kappa}{2}}$ and a suitable constant $C$.
\end{lem}
\begin{proof}
    We will verify the statement separately for the different terms $\sum_{0<|k|\leq CN^{\frac{\kappa}{2}}} \big|f_k-g_k\big|$, $\sum_{CN^{\frac{\kappa}{2}}<|k|\leq N} \big|f_k-g_k\big|$ and $\sum_{|k|>N} \big|f_k-g_k\big|$, where $C>0$. First of all observe that we have
\begin{align*}
   & f_k=A_k R\! \left[\left(\frac{B_k}{A_k}\right)^2\right]+\frac{(B_k)^2(A_k-|k|^2)}{2|k|^2 A_k},\\
   & g_k =\left(|k|^2+8\pi \mathfrak{a} N^\kappa\right)R\! \left[\left(\frac{8\pi \mathfrak{a} N^\kappa}{|k|^2+8\pi \mathfrak{a} N^\kappa}\right)^2\right]+\frac{(8\pi \mathfrak{a} N^\kappa)^3}{2|k|^2 \left(|k|^2+8\pi \mathfrak{a} N^\kappa\right)},
\end{align*}
with the definition of the Taylor residuum $R[x]:=\sqrt{1-x}-1+\frac{x}{2}$. By Lemma \ref{Lem:Coefficient_Comparison} we have the estimates $|A_k-|k|^2|\lesssim N^{\kappa}$ and $|B_k|\lesssim N^\kappa$ and therefore we obtain 
\begin{align*}
    \sum_{|k|>N} \big|f_k-g_k\big|\lesssim \sum_{|k|>N} \left(\big|f_k\big|+\big|g_k\big|\right)\lesssim \sum_{|k|>N} \frac{N^{3\kappa}}{|k|^4}\lesssim N^{3\kappa-1},
\end{align*}
where we have used $|R[x]|\lesssim |x|^2$. Furthermore we have $\left||k|^2+8\pi \mathfrak{a} N^{\kappa}-2A_k\right|\lesssim N^{2\kappa-1}(1+|k|)+N^{\kappa}\mathds{1}(|k|>K)$ and $\left|8\pi \mathfrak{a} N^{\kappa}-2B_k\right|\lesssim N^{2\kappa-1}(1+|k|)$ by Lemma \ref{Lem:Coefficient_Comparison}, and as a consequence we have $\big|\frac{(B_k)^2(A_k-|k|^2)}{2|k|^2 A_k}-\frac{(8\pi \mathfrak{a} N^\kappa)^3}{2|k|^2 \left(|k|^2+8\pi \mathfrak{a} N^\kappa\right)}\big|\lesssim \frac{N^{4\kappa}}{N|k|^3}+\frac{N^{3\kappa}\mathds{1}(|k|>K)}{|k|^4}$. Defining $x:=\left(\frac{B_k}{A_k}\right)^2$ as well as $y:=\left(\frac{8\pi \mathfrak{a} N^\kappa}{|k|^2+8\pi \mathfrak{a} N^\kappa}\right)^2$ we have the similar estimate $|x-y|\lesssim \frac{N^{3\kappa}}{N|k|^3}+\frac{N^{2\kappa}\mathds{1}(|k|>K)}{|k|^4}$, and hence
\begin{align*}
   & \big|A_k R[x] \! - \! \left(|k|^2 \! + \! 8\pi \mathfrak{a} N^\kappa\right)R[y]\big|\leq \big|(A_k \! - \! |k|^2-8\pi \mathfrak{a} N^\kappa) R[x] \big| \! + \! \left(|k|^2 \! + \! 8\pi \mathfrak{a} N^\kappa\right)\! \big| R[x] \! - \! R[y]\big|\\
   & \ \  \ \ \lesssim \left(N^{2\kappa-1}|k|+N^{\kappa}\mathds{1}(|k|>K)\right)\frac{N^{4\kappa}}{|k|^8}+|k|^2 \frac{N^{2\kappa}}{|k|^4}\left(\frac{N^{3\kappa}}{N|k|^3}+\frac{N^{2\kappa}\mathds{1}(|k|>K)}{|k|^4}\right)\\
   & \ \  \ \ \lesssim \frac{N^{5\kappa}}{N|k|^5}+\frac{N^{4\kappa}\mathds{1}(|k|>K)}{|k|^6}
\end{align*}
for $|k|>CN^{\frac{\kappa}{2}}$ and $C$ large enough such that $x,y\leq \frac{1}{2}$, where we have used that $\big| R[x] \! - \! R[y]\big|\lesssim \max\{x,y\}|y-x|$ for such $x$ and $y$. Consequently $\sum_{CN^{\frac{\kappa}{2}}<|k|\leq N} \big|f_k-g_k\big|\lesssim \frac{N^{4\kappa}\log N}{N}+\frac{N^{3\kappa}}{K}$. Regarding the final term $\sum_{0<|k|\leq CN^{\frac{\kappa}{2}}} \big|f_k-g_k\big|$, observe that 
\begin{align*}
    \bigg|\sqrt{A_k^2-B_k^2}-\sqrt{|k|^4+16\pi \mathfrak{a} N^\kappa |k|^2}\bigg|\lesssim \frac{\big|A_k^2-B_k^2-|k|^4-16\pi \mathfrak{a} N^\kappa |k|^2\big|}{\sqrt{|k|^4+16\pi \mathfrak{a} N^\kappa |k|^2}}\lesssim \frac{N^{3\kappa-1}}{N^{\frac{\kappa}{2}} |k|}.
\end{align*}
Furthermore we have $\big||k|^2+8\pi \mathfrak{a} N^\kappa-\frac{(8\pi \mathfrak{a} N^\kappa)^2}{2|k|^2}-A_k+C_k\big|\lesssim \frac{N^{3\kappa}}{N |k|}$, and therefore we conclude with the estimate $\sum_{0<|k|\leq CN^{\frac{\kappa}{2}}} \big|f_k-g_k\big|\lesssim \sum_{0<|k|\leq CN^{\frac{\kappa}{2}}} \frac{N^{3\kappa}}{N |k|}\lesssim N^{4\kappa-1}$. 
\end{proof}

While it is clear by Lemma \ref{Lem:Coefficient_Comparison} that the coefficients $\nu_k$ ,respectively $f_{\ell,\ell-k}$, introduced in Section \ref{Sec:Beyond_GP}, respectively Lemma \ref{Lem:Delta_Control}, are bounded from above by $\frac{N^{\kappa}}{|k|^2}$, the following Lemma \ref{Lem:Increased_Decay} shows that the decay in $k$ is even stronger. A similar result holds for the coefficients $\psi_p(\ell)$ defined below Eq.~(\ref{Eq:Estimate_With_Symmetry}).

\begin{lem}
\label{Lem:Increased_Decay}
    There exists a constant $C>0$, such that we have $\sum_{k\neq 0} |k|^2 \nu_k^2\leq C N^{1+\kappa}$ as well as $\sum_{k\neq 0} |k|^2 f_{\ell,\ell-k}^2\leq C N^{\kappa-1}$ for all $\ell$. Furthermore $\sum_{\ell\neq 0}|\ell|^2 |\psi_p(\ell)|^2\leq \frac{C}{L}$.
\end{lem}
\begin{proof}
Let us define $\varphi_{\ell }:=V_{N^{1-\kappa}}(e^{i\ell x}+e^{i\ell y})$ for $|\ell|<K$. Then we have by Lemma \ref{Lem:Coefficient_Comparison}
\begin{align}
\nonumber
 &\sum_{k\neq 0} |k|^2 f_{\ell,\ell-k}^2\leq  \sum_k  \left(|k|^2 \! + \! |\ell \! - \! k|^{2}\right)(f_{\ell,\ell-k})^2  = \braket{R\varphi_{\ell },(-\Delta_2)R\varphi_{\ell } }\leq \braket{R\varphi_{\ell },(-\Delta_2+V_{N^{1-\kappa}})R\varphi_{\ell } }\\
 \label{Eq:Increased_Decay}
 & \ \  \ \  \ \ =\braket{\varphi_{\ell },R\varphi_{\ell } }= 2\left\{(V_{N^{1-\kappa}} R V_{N^{1-\kappa}})_{\ell 0 ,\ell 0}+(V_{N^{1-\kappa}} R V_{N^{1-\kappa}})_{\ell 0 ,0 \ell }\right\}\lesssim N^{\kappa-1}.
\end{align}
The estimate on $\sum_{\ell\neq 0}|\ell|^2 |\psi_p(\ell)|^2$ can be derived in the same way using $\psi_p=R_p \widetilde{\varphi}$ with $\widetilde{\varphi}:=V_L$. Regarding $\nu_k$ note that $|\nu_k|\lesssim |k|^{-2}B_k=N\braket{e^{ik(x-y)},R\varphi_0}$, and therefore $\sum_{k\neq 0} |k|^2 \nu_k^2\lesssim N^2 \braket{R\varphi_{\ell },(-\Delta_2)R\varphi_{\ell } }$. Applying Eq.~(\ref{Eq:Increased_Decay}) for $\ell=0$ concludes the proof.
\end{proof}

\section{Additional Error Estimates}
\label{Appendix:Additional_Error_Estimates}
In this Section we will discuss estimates which will allow us, together with the results in Subsection \ref{Subsection:Error_Control}, to control the error term $\sum_{i=1}^4 \mathcal{E}_i$.
\begin{lem}
\label{Lem:Universal_Error_Estimate}
     Let $0<r<1$. Then there exists a constant $C>0$ such that
     \begin{align}
     \label{Eq:Error_4}
         \pm \mathcal{E}_4\Big|_{\mathcal{F}^+_{M}}\leq C N^{\frac{11\kappa}{2}-1}\left(\widetilde{\mathcal{N}}\Big|_{\mathcal{F}^+_{M}}+1\right)
     \end{align}
      for all $M\leq \min\{rN,N-1\}$. Furthermore
      \begin{align*}
          \pm \sum_k |k|^2 w_k^2\big([\delta_k,d_k^\dagger]+[\delta_k,d_k^\dagger]+[\delta_k,\delta_k^\dagger]\big)\Big|_{\mathcal{F}^+_{M}}\leq C N^{\frac{11\kappa}{2}-1}\left(\widetilde{\mathcal{N}}\Big|_{\mathcal{F}^+_{M}}+1\right).
      \end{align*}
\end{lem}
\begin{proof}
Recall the representation of the error term $\mathcal{E}_4|_{\mathcal{F}^+_{N-1}}$ from Lemma \ref{Lem:GP_Error_Estimate}
\begin{align*}
    \frac{1}{2}\sum_{k\neq 0}\left\{A_k-\sqrt{A_k^2-B_k^2}-C_k\right\}\left(\left[\delta_k,a_k^\dagger \frac{1}{\sqrt{a_0 a_0^\dagger}}a_0\right]+\left[a_0^\dagger \frac{1}{\sqrt{a_0 a_0^\dagger}}a_k,\delta_k^\dagger\right]+\left[\delta_k,\delta_k^\dagger\right]\right)\Bigg|_{\mathcal{F}^+_{N-1}}
\end{align*}
and compute $\left[\delta_k,a_k^\dagger \frac{1}{\sqrt{a_0 a_0^\dagger}}a_0\right]+\left[a_0^\dagger \frac{1}{\sqrt{a_0 a_0^\dagger}}a_k,\delta_k^\dagger\right]=\left(T_1+T_2+\mathrm{H.c.}\right)$ with
\begin{align*}
    T_1: & =\frac{a_0}{\sqrt{N}}\left(\frac{1}{\sqrt{a_0^\dagger a_0}}\sqrt{a_0 a_0^\dagger}-\frac{1}{\sqrt{a_0 a_0^\dagger}}\sqrt{a^\dagger_0 a_0}\right)\frac{a_0}{\sqrt{N}}w_k a_{-k}^\dagger a_k^\dagger,\\
    T_2: & =\frac{1}{\sqrt{a_0 a_0^\dagger}}\left(\sqrt{a_0^\dagger a_0} \! - \! \sqrt{a_0 a_0^\dagger}\right) \! \frac{a_0 a^\dagger_k}{N}\sum_{|\ell|<K}N\left(\left(T \! - \! 1\right)_{(\ell-k)k, \ell 0} \! + \! \left(T-1\right)_{(\ell-k)k, 0 \ell}\right) a^{\dagger}_{\ell-k}a_\ell.
\end{align*}
Using that $\frac{a_0}{\sqrt{N}}\left(\frac{1}{\sqrt{a_0^\dagger a_0}}\sqrt{a_0 a_0^\dagger}-\frac{1}{\sqrt{a_0 a_0^\dagger}}\sqrt{a^\dagger_0 a_0}\right)\frac{a_0}{\sqrt{N}}$ and $\frac{1}{\sqrt{a_0 a_0^\dagger}}\left(\sqrt{a_0^\dagger a_0} \! - \! \sqrt{a_0 a_0^\dagger}\right) \! \frac{a_0 a^\dagger_k}{N}$ is bounded by $\frac{C_r}{N}$ on $\mathcal{F}^+_{rN}$, we immediately obtain $(T_1+T_1^\dagger)|_{\mathcal{F}^+_{rN}}\lesssim N^{\kappa-1}(\mathcal{N}+1)$ from $\sum_{k\neq 0}w_k^2\lesssim N^{\kappa}$ and following the proof of Lemma \ref{Lem:Delta_Control} we furthermore obtain $(T_2+T_2^\dagger)|_{\mathcal{F}^+_{rN}} \lesssim N^{\kappa-1}\mathcal{N}$. In a similar fashion, one arrives at the estimate $\left[\delta_k,\delta_k^\dagger\right]\big|_{\mathcal{F}^+_{rN}}\lesssim N^{2\kappa-1}(\mathcal{N}+1)$. This concludes the proof of Eq.~(\ref{Eq:Error_4}), since $\left|A_k-\sqrt{A_k^2-B_k^2}-C_k\right|\lesssim N^{2\kappa}$ by Lemma \ref{Lem:Coefficient_Comparison} and therefore
\begin{align*}
     \pm \mathcal{E}_4\Big|_{\mathcal{F}^+_{M}}\lesssim N^{4\kappa-1}\mathcal{N}\Big|_{\mathcal{F}^+_{M}}\lesssim N^{\frac{11\kappa}{2}-1}\left(\widetilde{\mathcal{N}}\Big|_{\mathcal{F}^+_{M}}+1\right),
\end{align*}
where we have used Lemma \ref{Lem:True_Particle_Number_Comparison}. The second part of the Lemma can be verified analogously.
\end{proof}

The following Lemma \ref{Lem:True_Particle_Number_Comparison} will allow us to compare terms in the variables $a_k$ with terms in the new variables $b_k$.
\begin{lem}
\label{Lem:True_Particle_Number_Comparison}
    We have $\mathcal{N}\lesssim N^{\frac{\kappa}{2}}\widetilde{\mathcal{N}}+N^{\frac{3\kappa}{2}}$, and for $0<\delta\leq 1$ and $\ell \in \mathbb{N}$
        \begin{align*}
      \left(\sum_{k\neq 0}\mathfrak{a}_k^\dagger \mathfrak{a}_k+1\right)^\ell \! \!  \! \left(\sum_{k}|k|^{1+\delta} \mathfrak{a}_k^\dagger \mathfrak{a}_k+1\right)  \!    & \lesssim  \! N^{\left(\frac{3\ell}{2}+2-\delta\right)\kappa+\delta} \!  \left(\sum_{k\neq 0}\mathfrak{b}_k^\dagger \mathfrak{b}_k+1\right)^\ell  \! \!  \! \left(\sum_{k}|k|^{1+\delta} \mathfrak{b}_k^\dagger \mathfrak{b}_k+1\right) \! ,\\
 \left(\sum_{k\neq 0}\mathfrak{b}_k^\dagger \mathfrak{b}_k+1\right)^\ell \! \!  \! \left(\sum_{k}|k|^{1+\delta} \mathfrak{b}_k^\dagger \mathfrak{b}_k+1\right)  \!    & \lesssim  \! N^{\left(\frac{3\ell}{2}+2-\delta\right)\kappa+\delta} \!  \left(\sum_{k\neq 0}\mathfrak{a}_k^\dagger \mathfrak{a}_k+1\right)^\ell \! \!  \! \left(\sum_{k}|k|^{1+\delta} \mathfrak{a}_k^\dagger \mathfrak{a}_k+1\right).
    \end{align*}
The same estimate holds when we exchange $\sum_{k}|k|^{1+\delta} \mathfrak{b}_k^\dagger \mathfrak{b}_k$ with $\mathbb{K}$ defined in Lemma \ref{Lem:Potential_Epsilon_Tilde} and $\sum_{k}|k|^{1+\delta} \mathfrak{a}_k^\dagger \mathfrak{a}_k$ with $\sum_{k}|k|^{2} \mathfrak{a}_k^\dagger \mathfrak{a}_k+\frac{1}{2}\sum_{k\ell,mn}(V_{N^{1-\kappa}})_{\ell k,mn}\mathfrak{a}_k^\dagger \mathfrak{a}_\ell^\dagger \mathfrak{a}_m \mathfrak{a}_n$. Furthermore 
\begin{align*}
    \left(\mathcal{N}+1\right)^\ell \!  \left(\sum_{k}|k|^{1+\delta} a_k^\dagger a_k+1\right)\big|_{\mathcal{F}^+_{M_\ell}}  \!     \lesssim  \! N^{\left(\frac{3\ell}{2}+2-\delta\right)\kappa+\delta} \!  \left(\widetilde{\mathcal{N}}+1\right)^\ell  \!  \left(\sum_{k}|k|^{1+\delta} b_k^\dagger b_k+1\right)\big|_{\mathcal{F}^+_{M_\ell}}
\end{align*}
with the definition $M_\ell:=N-2(\ell+1)$.
\end{lem}
\begin{proof}
    Using $ a_k=\frac{1}{\sqrt{a_0 a_0^\dagger}}a_0\gamma_k b_k+\frac{1}{\sqrt{a_0 a_0^\dagger}}a_0\nu_{-k}b_{-k}^\dagger$, we can rewrite
\begin{align*}
    \mathcal{N}=\sum_{k\neq 0}\left(\gamma_k b_k+\nu_{-k}b_{-k}^\dagger\right)^\dagger \left(\gamma_k b_k+\nu_{-k}b_{-k}^\dagger\right)\lesssim \sum_{k\neq 0}\left(\gamma_k^2+\nu_k^2\right)b_k^\dagger b_k+\sum_{k\neq 0}\nu_k^2.
\end{align*}
Since $\sum_{k\neq 0}\nu_k^2\lesssim N^{\frac{3\kappa}{2}}$ and $|\nu_k|\leq |\gamma_k|\lesssim N^{\frac{\kappa}{4}}$, we obtain $\mathcal{N}\Big|_{\mathcal{F}^+_{N-1}}\lesssim N^{\frac{\kappa}{2}}\widetilde{\mathcal{N}}\Big|_{\mathcal{F}^+_{N-1}}+N^{\frac{3\kappa}{2}}$. For the other statements recall that $\mathfrak{a}_k=\gamma_k \mathfrak{b}_k+\nu_{-k}\mathfrak{b}_{-k}^\dagger$, and let us define the vector $\nu^\delta_k:=|k|^{\frac{1+\delta}{2}}\nu_k$. Then we can estimate $ \left(\sum_{k\neq 0}\mathfrak{a}_k^\dagger \mathfrak{a}_k+1\right)^\ell \! \!  \! \left(\sum_{k}|k|^{1+\delta} \mathfrak{a}_k^\dagger \mathfrak{a}_k+1\right)$ from above by
\begin{align*}
& C_{\ell} \max\{\|\nu\|^2,\|\gamma\|^2_\infty\}^\ell \max\{\|\nu^\delta\|^2,\|\gamma\|^2_\infty\}\left(\sum_{k\neq 0}\mathfrak{b}_k^\dagger \mathfrak{b}_k+1\right)^\ell  \! \!  \! \left(\sum_{k}|k|^{1+\delta} \mathfrak{b}_k^\dagger \mathfrak{b}_k+1\right)\\
& \ \  \ \  \ \  \ \ \lesssim N^{\left(\frac{3\ell}{2}+2-\delta\right)\kappa+\delta} \!  \left(\sum_{k\neq 0}\mathfrak{b}_k^\dagger \mathfrak{b}_k+1\right)^\ell  \! \!  \! \left(\sum_{k}|k|^{1+\delta} \mathfrak{b}_k^\dagger \mathfrak{b}_k+1\right).
\end{align*}
for a suitable constant $C_{\ell}$, where we have used $\|\nu\|^2=\sum_{k\neq 0}\nu_k^2\lesssim N^{\frac{3\kappa}{2}}$, $\|\gamma\|^2_\infty\lesssim N^{\frac{\kappa}{2}}$ and $\|\nu^\delta\|^2\lesssim N^{2\kappa+\delta(1-\kappa)}$. The same statement holds for $\mathfrak{a}_k$ exchanged with $\mathfrak{b}_k$, since $\mathfrak{b}_k=\gamma_k \mathfrak{a}_k+\nu_{k}\mathfrak{a}_{-k}^\dagger$. In order to verify the final statement, note that $\left(\mathcal{N}+1\right)^\ell \!  \left(\sum_{k}|k|^{1+\delta} a_k^\dagger a_k+1\right)\Psi=\left(\sum_{k\neq 0}\mathfrak{a}_k^\dagger \mathfrak{a}_k+1\right)^\ell \! \!  \! \left(\sum_{k}|k|^{1+\delta} \mathfrak{a}_k^\dagger \mathfrak{a}_k+1\right)\Psi$ for $\Psi\in \mathcal{F}^+_{N-1}$. Using $b_k^\dagger b_k \mathcal{F}^+_{M}\subset \mathcal{F}^+_{M+2}$ we furthermore have $\left(\widetilde{\mathcal{N}}+1\right)^\ell \!  \left(\sum_{k}|k|^{1+\delta} b_k^\dagger b_k+1\right)\Psi=\left(\sum_{k\neq 0}\mathfrak{b}_k^\dagger \mathfrak{b}_k+1\right)^\ell \! \!  \! \left(\sum_{k}|k|^{1+\delta} \mathfrak{b}_k^\dagger \mathfrak{b}_k+1\right)\Psi$ for $\Psi\in \mathcal{F}^+_{M_\ell}$, which concludes the proof.
\end{proof}

While the most prominent error terms appearing in the proof of Theorem \ref{Th:Beyond_GP} are being estimated in Subsection \ref{Subsection:Error_Control}, the following Lemma \ref{Lem:Particle_Number_Comparison} provides good estimates on the residual terms, which will furthermore be important for the proof of Theorem \ref{Th:Upper_Bound}.

\begin{lem}
\label{Lem:Particle_Number_Comparison}
    Let $0<\lambda_0<1-4\kappa$ and $\frac{5\kappa}{2}<\lambda<1-\frac{3\kappa}{2}$. Then 
    \begin{align}
   \label{Eq:Variable_Comparison_Spectrum}
  &  b_k^\dagger b_k\Big|_{\mathcal{F}^{\leq}_{N^{\lambda_0}}\cap \mathcal{F}^+_{N^\lambda}}-\left(\gamma_k d_k+\nu_k d_k^\dagger\right)^\dagger \left(\gamma_k d_k+\nu_k d_k^\dagger\right)\Big|_{\mathcal{F}^{\leq}_{N^{\lambda_0}}\cap \mathcal{F}^+_{N^\lambda}} \\
  \nonumber
   &\  \lesssim \left(N^{2\kappa+\frac{\lambda_0}{2}-\frac{1}{2}}+N^{\frac{3\kappa}{2}+\lambda-1}\right)\left(b_k^\dagger b_k+\frac{\widetilde{\mathcal{N}}+1}{|k|^4}\right)\Big|_{\mathcal{F}^{\leq}_{N^{\lambda_0}}\cap \mathcal{F}^+_{N^\lambda}} , 
    \end{align}
   and there exists a constant $K_0$ such that for $K\geq K_0 N^{\frac{\kappa}{2}}$ and $\delta>0$
    \begin{align}
    \label{Eq:Comparison_With_Proper_Variables}
        \sum_{k}|k|^{1-\delta}b_k^\dagger b_k\Big|_{\mathcal{F}^{\leq}_{N^{\lambda_0}}\cap \mathcal{F}^+_{N^\lambda}}\lesssim N^{-\frac{\kappa}{2}}\sum_{k\neq 0}e_k\left(\gamma_k d_k+\nu_k d_k^\dagger\right)^\dagger \left(\gamma_k d_k+\nu_k d_k^\dagger\right)\Big|_{\mathcal{F}^{\leq}_{N^{\lambda_0}}\cap \mathcal{F}^+_{N^\lambda}}+1.
    \end{align}
 Regarding the error terms $F_1$ and $F_2$ introduced below Eq.~(\ref{Eq:Representation_Of_Error_II}) we have the estimates
    \begin{align}
             \label{Eq:Residuum_II}
        & \ \ \ \  \pm \left(F_1+F_1^\dagger\right)\bigg|_{\mathcal{F}^{\leq}_{N^{\lambda_0}}\cap \mathcal{F}^+_{N^\lambda}}\lesssim N^{\frac{5\kappa}{2}+\lambda-1}\left(\widetilde{\mathcal{N}}\Big|_{\mathcal{F}^{\leq}_{N^{\lambda_0}}\cap \mathcal{F}^+_{N^\lambda}}+1\right),\\
        \label{Eq:Residuum_I}
     &  \pm F_2 \Big|_{\mathcal{F}^{\leq}_{N^{\lambda_0}}\cap \mathcal{F}^+_{N^\lambda}} \lesssim \left(N^{\frac{9\kappa}{2}+\lambda_0-1}+N^{\frac{7\kappa}{2}+2\lambda-2}\right)\left(\widetilde{\mathcal{N}}\Big|_{\mathcal{F}^{\leq}_{N^{\lambda_0}}\cap \mathcal{F}^+_{N^\lambda}}+1\right),
    \end{align}
    and $\pm \left(F_1+F_1^\dagger+F_2\right)\lesssim N^{6\kappa-1}\! \left(\widetilde{\mathcal{N}}+1\right)^3$. Further $\mathcal{N}(\mathcal{N}+1)\Big|_{\mathcal{F}^+_{N^\lambda}} \lesssim  2N^{\frac{\kappa}{2}+\lambda}\left(\widetilde{\mathcal{N}}\Big|_{\mathcal{F}^+_{N^\lambda}}+1\right)$. 
\end{lem}
\begin{proof}
 In order to verify the Lemma, let us first show
\begin{align}
\label{Eq:Strong_Delta_Control}
    \left\{(\delta_k)^\dagger \delta_k+ \delta_k (\delta_k)^\dagger\right\}\Big|_{\mathcal{F}^{\leq}_{N^{\lambda_0}}\cap \mathcal{F}^+_{N^\lambda}}\lesssim |k|^{-4}N^{\frac{5\kappa}{2}}\left(N^{\kappa+\lambda_0-1}+N^{2\lambda-2}\right)\left(\widetilde{\mathcal{N}}\Big|_{\mathcal{F}^{\leq}_{N^{\lambda_0}}\cap \mathcal{F}^+_{N^\lambda}}+1\right).
\end{align}
For this purpose, we follow the proof of Lemma \ref{Lem:Delta_Control} by writing $\delta_k=-\delta_k'+\delta_k''$ and estimating
\begin{align*}
   \left\{(\delta_k'')^\dagger \delta_k''+ \delta_k'' (\delta_k'')^\dagger\right\}\Big|_{\mathcal{F}^+_{N^\lambda}}\lesssim N^{2\kappa-2}|k|^{-4}\left(\mathcal{N}+1\right)^3\lesssim N^{2\kappa-2}|k|^{-4}\Big(\sum_{m,n,p\neq 0} a_m^\dagger a_n^\dagger a_p^\dagger a_p a_n a_m+1\Big).
\end{align*}
By expressing $a_p$ in terms of $b_p$ and $b_{-p}^\dagger$, we furthermore obtain
\begin{align}
\label{Eq:Third_Power_Comparison}
   & \sum_{m,n,p\neq 0} a_m^\dagger a_n^\dagger a_p^\dagger a_p a_n a_m\Big|_{\mathcal{F}^+_{N^\lambda}} \!  \!  \! = \! \sum_{m,n,p\neq 0} (\gamma_p^2 \! + \! \nu_p^2) a_m^\dagger a_n^\dagger b_p^\dagger b_p a_n a_m\Big|_{\mathcal{F}^+_{N^\lambda}} \!  \!  + \! \sum_{m,n,p\neq 0} \nu_p^2 a_m^\dagger a_n^\dagger a_n a_m\Big|_{\mathcal{F}^+_{N^\lambda}}\\
   \nonumber
    & \ \  \ \  \ \  \ \ \lesssim \sum_{m,n,p\neq 0} (\gamma_p^2 \! + \! \nu_p^2)  b_p^\dagger a_m^\dagger a_n^\dagger a_n a_m b_p \Big|_{\mathcal{F}^+_{N-1}} \! 
  \!  + \sum_{m,n,p\neq 0} (\gamma_p^2 \! + \! \nu_p^2)   [a_n a_m, b_p]^\dagger [a_n a_m, b_p] \Big|_{\mathcal{F}^+_{N^\lambda}}\\
  \nonumber
  & \ \  \ \  \ \  \ \  \ \  \ \  \ \  \ \  +\! \sum_{m,n,p\neq 0} \nu_p^2 a_m^\dagger a_n^\dagger a_n a_m\Big|_{\mathcal{F}^+_{N^\lambda}}\lesssim N^{\frac{\kappa}{2}}\sum_{p\neq 0} b_p^\dagger \mathcal{N}^2 b_p \Big|_{\mathcal{F}^+_{N^\lambda}}+N^{2\kappa}\mathcal{N}\Big|_{\mathcal{F}^+_{N^\lambda}}+N^{\frac{3\kappa}{2}}\mathcal{N}^2\Big|_{\mathcal{F}^+_{N^\lambda}}\\
  \nonumber
  &  \ \  \ \  \ \  \ \ \lesssim N^{\frac{\kappa}{2}+2\lambda}\widetilde{\mathcal{N}}\Big|_{\mathcal{F}^+_{N^\lambda}}+N^{\frac{7\kappa}{2}}\left(\widetilde{\mathcal{N}}\Big|_{\mathcal{F}^+_{N^\lambda}}+1\right)+N^{3\kappa+\lambda}\left(\widetilde{\mathcal{N}}\Big|_{\mathcal{F}^+_{N^\lambda}}+1\right).
\end{align}
Since we assume $\lambda>\frac{5\kappa}{2}$, we obtain $ \left\{(\delta_k'')^\dagger \delta_k''+ \delta_k'' (\delta_k')^\dagger\right\}\Big|_{\mathcal{F}^+_{N^\lambda}}\lesssim |k|^{-4} N^{\frac{5\kappa}{2}+2\lambda-2}\left(\widetilde{\mathcal{N}}+1\right)$. 
Let us furthermore estimate $\left\{(\delta_k')^\dagger \delta_k'+ \delta_k' (\delta_k')^\dagger\right\}\Big|_{\mathcal{F}^{\leq}_{N^{\lambda_0}}\cap \mathcal{F}^+_{N^\lambda}}$ from above by
\begin{align}
\label{Eq:Third_Power_Comparison_Alt}
 & \ \ \ \ \ \ \ \ \ \ \ |k|^{-4}N^{2\kappa-1}\left(\sum_{0<|\ell|,|\ell'|<K} a^\dagger_{\ell'-k}a_\ell^\dagger a_\ell a_{\ell'-k}+\sum_{p\neq 0} a_p^\dagger a_p\right)\Big|_{\mathcal{F}^{\leq}_{N^{\lambda_0}}\cap \mathcal{F}^+_{N^\lambda}}\\
 \nonumber
     &  \lesssim \! |k|^{-4}N^{2\kappa-1}\left(\! N^{\lambda_0}\sum_{p\neq 0}a^\dagger_{p} a_{p}\! +\! N^{\frac{3\kappa}{2}}\left(\widetilde{\mathcal{N}}+1\right)\! \! \! \right)\Big|_{\mathcal{F}^{\leq}_{N^{\lambda_0}}\cap \mathcal{F}^+_{N^\lambda}}\! \! \!  \lesssim \! |k|^{-4}N^{\frac{7\kappa}{2}+\lambda_0-1}\! \! \left(\! \widetilde{\mathcal{N}}\Big|_{\mathcal{F}^{\leq}_{N^{\lambda_0}}\cap \mathcal{F}^+_{N^\lambda}}\! +\! 1\! \right)\! .
\end{align}
Combining the estimates on $\delta'_k$ and $\delta''_k$ therefore yields Eq.~(\ref{Eq:Strong_Delta_Control}). In order to verify Eq.~(\ref{Eq:Comparison_With_Proper_Variables}) note that $b_k-\gamma_k d_k-\nu_k d_{-k}^\dagger=\gamma_k \delta_k+\nu_k \delta_{-k}^\dagger$. Using Eq.~(\ref{Eq:Strong_Delta_Control}), as well as the observation that $|k|^{1-\delta} \lesssim N^{-\frac{\kappa}{2}}e_k$ for $K\geq K_0 N^{\frac{\kappa}{2}}$ and $K_0$ large enough, yields
\begin{align*}
      & \sum_{k}|k|^{1-\delta}b_k^\dagger b_k\Big|_{\mathcal{F}^{\leq}_{N^{\lambda_0}}\cap \mathcal{F}^+_{N^\lambda}}\lesssim N^{-\frac{\kappa}{2}}\sum_{k\neq 0}e_k\left(\gamma_k d_k+\nu_k d_k^\dagger\right)^\dagger \left(\gamma_k d_k+\nu_k d_k^\dagger\right)\Big|_{\mathcal{F}^{\leq}_{N^{\lambda_0}}\cap \mathcal{F}^+_{N^\lambda}}\\
     & \ \ \ \ \ \ \ \ \ \ \ \ \ \ \ \ \ \ \ \ +N^{\frac{\kappa}{2}}\sum_k |k|^{1-\delta}(\delta_k^\dagger \delta_k+\delta_k \delta_k^\dagger)\Big|_{\mathcal{F}^{\leq}_{N^{\lambda_0}}\cap \mathcal{F}^+_{N^\lambda}}\\
     &  \lesssim \! N^{-\frac{\kappa}{2}}\! \sum_{k\neq 0}\! e_k \! \left(\gamma_k d_k \! + \! \nu_k d_k^\dagger\right)^\dagger  \! \left(\gamma_k d_k \! + \! \nu_k d_k^\dagger\right) \! \Big|_{\mathcal{F}^{\leq}_{N^{\lambda_0}}\cap \mathcal{F}^+_{N^\lambda}} \! \!   \! \! + \! N^{3\kappa} \! \! \left(N^{\kappa+\lambda_0-1} \! + \! N^{2\lambda-2}  \right) \! \! \left( \! \widetilde{\mathcal{N}}\Big|_{\mathcal{F}^{\leq}_{N^{\lambda_0}}\cap \mathcal{F}^+_{N^\lambda}} \!  \!  \! + \! 1 \! \right)\! \! .  
\end{align*}
By our assumption $\lambda_0<1-4\kappa$ and $\lambda_0<1-\frac{3}{2}\kappa$ we have 
\begin{align*}
  \sum_{k}|k|^{1-\delta}b_k^\dagger b_k & \lesssim \left(1-N^{3\kappa} \! \left(N^{\kappa+\lambda_0-1} \! + \! N^{2\lambda-2}  \right)\right)\sum_{k}|k|^{1-\delta}b_k^\dagger b_k\\
  &\leq \sum_{k}|k|^{1-\delta}b_k^\dagger b_k-N^{3\kappa} \! \left(N^{\kappa+\lambda_0-1} \! + \! N^{2\lambda-2}  \right)\widetilde{\mathcal{N}},
\end{align*}
which concludes the proof of Eq.~(\ref{Eq:Comparison_With_Proper_Variables}). Note that Eq.~(\ref{Eq:Variable_Comparison_Spectrum}) can be verified similarly. Furthermore Eq.~(\ref{Eq:Residuum_I}) follows immediately from Eq.~(\ref{Eq:Strong_Delta_Control}), using the fact that $|k|^2 |w_k|\lesssim N^\kappa$. When it comes to Eq.~(\ref{Eq:Residuum_II}), let us estimate in a similar fashion to Lemma \ref{Lem:Delta_Control}
\begin{align}
\label{Eq:N_Quadrat}
   \pm & \left(F_1+F_1^\dagger \right)\lesssim N^{2\kappa-1}\mathcal{N}(\mathcal{N} \! + \! 1)=N^{2\kappa-1}\sum_{p,q}a_p^\dagger a_q^\dagger a_q a_p+N^{2\kappa-1}\sum_p a_p^\dagger a_p\\
   \nonumber
   & \lesssim N^{\frac{5\kappa}{2}-1}\sum_{p,q}a_p^\dagger b_q^\dagger b_q a_p+N^{5\kappa-1}\! \left(\widetilde{\mathcal{N}}+1\right)\lesssim N^{\frac{5\kappa}{2}-1}\sum_{p,q}b_q^\dagger a_p^\dagger a_p b_q +N^{5\kappa-1}\! \left(\widetilde{\mathcal{N}}+1\right),
\end{align}
following the argument in Eq.~(\ref{Eq:Third_Power_Comparison}). This immediately implies Eq.~(\ref{Eq:Residuum_II}) as well as the bound $\pm \left(F_1+F_1^\dagger\right)\lesssim N^{5\kappa-1}\! \left(\widetilde{\mathcal{N}}+1\right)^2$. Note that $\mathcal{N}^2\lesssim N^{3\kappa} \! \left(\widetilde{\mathcal{N}}+1\right)^2$ by Eq.~(\ref{Eq:N_Quadrat}), and therefore we obtain by a similar argument as in Eq.~(\ref{Eq:Third_Power_Comparison}) and Eq.~(\ref{Eq:Third_Power_Comparison_Alt}) that $(\delta_k)^\dagger \delta_k+ \delta_k (\delta_k)^\dagger\lesssim |k|^{-4}N^{5\kappa-1} \! \left(\widetilde{\mathcal{N}}+1\right)^3$, and consequently $\pm F_2\lesssim N^{6\kappa-1}\! \left(\widetilde{\mathcal{N}}+1\right)^3$.
\end{proof}

\section{A priori Condensation}
\label{Appendix:A_priori_Condensation}
 In this subsection, we will use the strong estimates on the number of expected excited particles $\braket{\Psi,\mathcal{N}\Psi}$ derived in \cite{F}, in order to obtain necessary ad hoc for our proof of Theorem \ref{Th:Beyond_GP}. By \cite[Theorem 1.2]{F} we have the a priori estimate $\braket{\Psi,\mathcal{N}\Psi}\leq C N^{\frac{5\kappa}{2}}$ for states $\Psi$ satisfying $\braket{\Psi,H_{N,\kappa}\Psi}-4\pi \mathfrak{a}N^{1+\kappa}\lesssim N^{\frac{5\kappa}{2}}$. Making use of the fact that the Lee-Huang-Yang correction $\frac{1}{2} \! \sum_{k\neq 0} \! \left\{ \! \sqrt{|k|^4 \! + \! 16\pi \mathfrak{a} N^{\kappa}|k|^2} \! - \! |k|^2 \! - \! 8\pi \mathfrak{a} N^{\kappa} \! + \! \frac{(8\pi \mathfrak{a} N^{\kappa})^2}{2|k|^2} \! \right\}$ is of the order $N^{\frac{5\kappa}{2}}$, therefore yields 
\begin{align}
\label{Eq:A_priori_Condensation}
       \braket{\Psi_d,\mathcal{N}\Psi_d}\leq C N^{\frac{5\kappa}{2}} 
\end{align}
for any eigenstate $\Psi_d$ corresponding to an eigenvalue $E_{N,\kappa}^{(d)}$ with $E_{N,\kappa}^{(d)} - E_{N,\kappa}\lesssim N^{\frac{5\kappa}{2}}$. Note that the Eq.~(\ref{Eq:A_priori_Condensation}) only gives us an estimate on the expected number of particles. The following Lemma \ref{Lem:IMS} however tells us that we can lift the control on the expected number of particles to a control on the number of particles in a spectral sense, i.e. we will argue that we can restrict our attention to states in the spectral subspace $\mathcal{F}^{\leq}_{\lambda_0}\cap \mathcal{F}^{+}_\lambda$ without changing the energy significantly. Here we follow the methods presented in \cite{FGJMO}.

\begin{lem}
\label{Lem:IMS}
Let $E_{N,\kappa}^{(d)}$ satisfy $E_{N,\kappa}^{(d)}\leq E_{N,\kappa}+CN^{\frac{\kappa}{2}}$ for a constant $C>0$ and $K\leq N^{1+\kappa}$, and assume $\lambda_0>\frac{5\kappa}{2}$ as well as $\lambda>3\kappa$. Then there exists a $d$ dimensional subspace $\mathcal{V}_d\subseteq \mathcal{F}^{\leq}_{\lambda_0}\cap \mathcal{F}^{+}_\lambda$, such that for all elements $\Psi\in \mathcal{V}_d$ with $\|\Psi\|=1$
\begin{align}
    \label{Eq:IMS_Result}
      \braket{\Psi,H_{N,\kappa}\Psi}\leq E^{(d)}_{N,\kappa}+N^{-2\lambda_0} \! \left(N^{\frac{9\kappa}{2}}+KN^{2\kappa}\right)+N^{3\kappa-\lambda_0}+N^{-2\lambda}N^{1+\kappa}+N^{3\kappa-\lambda}.
\end{align}
\end{lem}
\begin{proof}
  Let $\mathcal{W}_d$ be the space spanned by the first $d$ eigenfunctions of $H_{N,\kappa}$, $f,g:\mathbb{R}\longrightarrow [0,1]$ smooth functions satisfying $f^2+g^2=1$ and $f(x)=1$ for $x\leq \frac{1}{2}$ as well as $f(x)=0$ for $x\geq 1$ and let $\mathcal{N}_*:=\sum_{0<|k|<K}a_k^\dagger a_k$. With this at hand we define the space $\widetilde{\mathcal{V}}_d:=f\!\left(\frac{\mathcal{N}_*}{N^{\lambda_0}}\right) \! \mathcal{W}_d$, which clearly satisfies $\widetilde{\mathcal{V}}_d\subseteq \mathcal{F}^{\leq}_{\lambda_0}$. Note that by our assumption $\lambda_0>\frac{5\kappa}{2}$ and Eq.~(\ref{Eq:A_priori_Condensation}), we have
  \begin{align}
  \label{Eq:Mass_Estimate}
      \left\|f\!\left(\frac{\mathcal{N}_*}{N^{\lambda_0}}\right)\Psi\right\|^2\geq \left(1-C N^{\frac{5\kappa}{2}-\lambda_0}\right)\|\Psi\|^2>0
  \end{align}
 for any $\Psi\in \mathcal{W}_d\setminus \{0\}$, a suitable constant $C$ and $N$ large enough. Hence it is clear that $\widetilde{V}_d$ is $d$ dimensional again. Using the IMS identity $ f \!  \left(\frac{\mathcal{N}_*}{N^\lambda_0}\right) H_{N,\kappa}f \!  \left(\frac{\mathcal{N}_*}{N^\lambda_0}\right)+g \!  \left(\frac{\mathcal{N}_*}{N^\lambda_0}\right) H_{N,\kappa}g \!  \left(\frac{\mathcal{N}_*}{N^\lambda_0}\right)=H_{N,\kappa}+\mathcal{E}$ with
    \begin{align*}
      \mathcal{E}:=\frac{1}{2}\left[f \!  \left(\frac{\mathcal{N}_*}{N^\lambda_0}\right),\left[H_{N,\kappa},f \!  \left(\frac{\mathcal{N}_*}{N^\lambda_0}\right)\right]\right]+\frac{1}{2}\left[g \!  \left(\frac{\mathcal{N}_*}{N^\lambda_0}\right),\left[H_{N,\kappa},g \!  \left(\frac{\mathcal{N}_*}{N^\lambda_0}\right)\right]\right],
    \end{align*}
see for example \cite{LNSS,FGJMO}, we obtain furthermore for a generic state $\Theta=\frac{f\!\left(\frac{\mathcal{N}_*}{N^{\lambda_0}}\right)\Psi}{\|f\!\left(\frac{\mathcal{N}_*}{N^{\lambda_0}}\right)\Psi\|}\in \widetilde{V}_d$, where $\Psi$ is a state in $\mathcal{W}_d$, and a suitable constant $C>0$ the estimate
\begin{align*}
    \braket{\Theta,H_{N,\kappa}\Theta} & \leq \left\|f\!\left(\frac{\mathcal{N}_*}{N^{\lambda_0}}\right)\Psi\right\|^{-2}\left(E_{N,\kappa}^{(d)}-\braket{g\!\left(\frac{\mathcal{N}_*}{N^{\lambda_0}}\right)\Psi,H_{N,\kappa}g\!\left(\frac{\mathcal{N}_*}{N^{\lambda_0}}\right)\Psi}+\braket{\Psi,\mathcal{E}\Psi}\right)\\
    &\leq E_{N,\kappa}^{(d)}+C \braket{\Psi,\mathcal{E}\Psi}+CN^{3\kappa-\lambda_0},
\end{align*}
where we have used the lower bound $\braket{g\!\left(\frac{\mathcal{N}_*}{N^{\lambda_0}}\right)\Psi,H_{N,\kappa}g\!\left(\frac{\mathcal{N}_*}{N^{\lambda_0}}\right)\Psi}\geq \left(1-\left\|f\!\left(\frac{\mathcal{N}_*}{N^{\lambda_0}}\right)\Psi\right\|^2\right)E_{N,\kappa}\geq \left(1-\left\|f\!\left(\frac{\mathcal{N}_*}{N^{\lambda_0}}\right)\Psi\right\|^2\right)E^{(d)}_{N,\kappa} -C N^{3\kappa-\lambda_0}$, which follows from our assumptions on $E^{(d)}_{N,\kappa}$ together with Eq.~(\ref{Eq:Mass_Estimate}). In order to estimate $\braket{\Psi,\mathcal{E}\Psi}$, let us define $\pi_0$ as the projection onto the zero mode, $\pi_1$ as the projection onto the modes $\{k\in 2\pi\mathbb{Z}^3:0<|k|<K\}$ and $\pi_3$ as the projection onto the modes $\{k\in 2\pi\mathbb{Z}^3:|k|>K\}$, and rewrite $\mathcal{E}$ as
\begin{align*}
    \mathcal{E}=\frac{1}{4N^{2\lambda_0}}\sum_{I\in \{0,1,2\}^4}  \sum_{jk,mn} (\pi_{I_1}\pi_{I_2} V_{N^{1-\kappa}} \pi_{I_3}\pi_{I_4})_{jk, mn} a^\dagger_k a_j^\dagger X_I a_m a_n,
\end{align*}
with $X_I:= N^{2\lambda_0}\left[f \!  \left(\frac{\mathcal{N}_*+\#_{I_1,I_2}}{N^{\lambda_0}}\right)-f \!  \left(\frac{\mathcal{N}_*+\#_{I_3,I_4}}{N^{\lambda_0}}\right)\right]^2+ N^{2\lambda_0}\left[g \!  \left(\frac{\mathcal{N}_*+\#_{I_1,I_2}}{N^{\lambda_0}}\right)-g \!  \left(\frac{\mathcal{N}_*+\#_{I_3,I_4}}{N^{\lambda_0}}\right)\right]^2$ and $\#_{i,j}$ counts how many of the indices $i,j$ are equal to $1$. Before we start with the term-by-term analysis, let us introduce the variables $\widetilde{c}$ and $\widetilde{\psi}$ as the ones defined in Eq.~(\ref{Eq:Definition_c}) and Eq.~(\ref{Eq:Definition_psi}) with the concrete choice of the cut-off parameter $0<K_0<2\pi$. Note that in this case, we obtain in analogy to Eq.~(\ref{Eq:Pseudo_Quadratic}) for a suitable $C>0$
\begin{align*}
    H_{N,\kappa}-4\pi \mathfrak{a}_{N^{1-\kappa}}N^{\kappa}(N-1)\geq \mathbb{P}-CN^{\kappa}\mathcal{N},
\end{align*}
where $\mathbb{P}:=\sum_k |k|^2 \widetilde{c}_k^\dagger \widetilde{c}_k+\frac{1}{2}\sum_{jk,mn\neq 0}\left(V_{N^{1-\kappa}}\right)_{jk,mn}\widetilde{\psi}^\dagger_{jk}\widetilde{\psi}_{mn}$, which especially implies the upper bound $\braket{\Psi,\mathbb{P}\Psi}\lesssim N^{\frac{7\kappa}{2}}$ for $\Psi\in \mathcal{W}_d$. Starting with the case $I_1=I_2=0$ and $I_3=I_4=1$, we identify $\sum_{|k|<K}(V_{N^{1-\kappa}})_{00,(-k)k}\left(a_0^{\dagger}\right)^2 X_{I} a_{-k}a_k+\mathrm{H.c.}$ as
\begin{align*}
    & \left(\sum_{|k|<K}(V_{N^{1-\kappa}})_{00,(-k)k}\left(a_0^{\dagger}\right)^2 X_{I} \widetilde{\psi}_{(-k)k}-N^{-1}\sum_{|k|<K}(V_{N^{1-\kappa}})_{00,(-k)k}\left(a_0^{\dagger}\right)^2 a_0^2 X_{I} w_k \right)+\mathrm{H.c.}\\
    & \ \ \ \ \ \ \ \ \ \ \ \ \lesssim N^{\kappa}\mathbb{P}+KN^{2\kappa},
\end{align*}
which gives an contribution of at most order $N^{\frac{9\kappa}{2}}+K N^{2\kappa}$ when evaluated against $\Psi$. Since we clearly only have to consider $I$ for which $\#_{I_1,I_2}\neq \#_{I_3,I_4}$, the only relevant cases left are the ones where both $I_1,I_2$ and $I_3,I_4$ contain at least one non-zero index, and at least one of these index pairs contains the index $1$. W.l.o.g., let us assume $I_1=1$. In this case 
\begin{align}
\nonumber
    & \sum_{jk,mn} (\pi_{I_1}\pi_{I_2} V_{N^{1-\kappa}} \pi_{I_3}\pi_{I_4})_{jk, mn} a^\dagger_k a_j^\dagger X_I a_m a_n=\sum_{jk,mn} (\pi_{I_1}\pi_{I_2} V_{N^{1-\kappa}} \pi_{I_3}\pi_{I_4})_{jk, mn} a^\dagger_k a_j^\dagger X_I \widetilde{\psi}_{mn}\\
    \label{Eq:General_Case}
    &  \ \ \ \ \ \ \ \ \ \ \ \  \ \ \ \ \ \ \ \ \ -N^{-1} a_0^2\sum_{jk,m} (\pi_{I_1}\pi_{I_2} V_{N^{1-\kappa}} \pi_{I_3}\pi_{I_4})_{jk, m(-m)} a^\dagger_k a_j^\dagger X_I w_m.
\end{align}
Note that the second term on the right hand side of Eq.~(\ref{Eq:General_Case}) can be treated in the same way as the case $I_1=I_2=0$ and $I_3=I_4=1$. Regarding the first term we estimate
\begin{align*}
    & \sum_{jk,mn} (\pi_{I_1}\pi_{I_2} V_{N^{1-\kappa}} \pi_{I_3}\pi_{I_4})_{jk, mn} a^\dagger_k a_j^\dagger X_I \widetilde{\psi}_{mn}+\mathrm{H.c.}\lesssim \sum_{jk,mn} (\pi_{I_1}\pi_{I_2} V_{N^{1-\kappa}} \pi_{I_1}\pi_{I_2})_{jk, mn} a^\dagger_k a_j^\dagger a_m a_n+\mathbb{P}\\
    & \ \ \ \ \ \ \ \ \ \ \lesssim \mathbb{P}+N^{-2}\left(a_0^\dagger\right)^2 a_0^2 \sum_{|j|,|m|<K}(V_{N^{1-\kappa}} )_{j(-j), m(-m)}w_j w_m\lesssim \mathbb{P}+K^2 N^{\kappa-1},
\end{align*}
which is of order $N^{\frac{7\kappa}{2}}+KN^{2\kappa}$ due to our assumption $K\leq N^{1+\kappa}$. Therefore
\begin{align*}
    \braket{\Theta,H_{N,\kappa}\Theta}\leq E^{(d)}_{N,\kappa}+N^{-2\lambda_0} \! \left(N^{\frac{9\kappa}{2}}+KN^{2\kappa}\right)+N^{3\kappa-\lambda_0}
\end{align*}
for any $\Theta\in \widetilde{\mathcal{V}}_d$ with $\|\Theta\|=1$. Note that states in $\Theta\in \widetilde{\mathcal{V}}_d$ still satisfy $\braket{\Theta,\mathcal{N}\Theta}\lesssim N^{\frac{5\kappa}{2}}$, and let us further define $\mathcal{V}_d:=f\!\left(\frac{\mathcal{N}}{N^\lambda}\right)\! \widetilde{\mathcal{V}}_d$. Clearly $\mathcal{V}_d\subseteq \mathcal{F}^{\leq}_{\lambda_0}\cap \mathcal{F}^{+}_\lambda$ and similar to before we see that $\mathcal{V}_d$ is indeed $d$ dimensional. Making again use of the IMS identity, and performing similar estimates then yields for generic states $\Theta=\frac{f\!\left(\frac{\mathcal{N}}{N^\lambda}\right)\Psi}{\left\|f\!\left(\frac{\mathcal{N}}{N^\lambda}\right)\right\|}\in \mathcal{V}_d$, where $\Psi\in \widetilde{\mathcal{V}}_d$ with $\|\Psi\|=1$, and 
\begin{align*}
    \braket{\Theta,H_{N,\kappa}\Theta}\leq E^{(d)}_{N,\kappa}+N^{-2\lambda_0} \! \left(N^{\frac{9\kappa}{2}}+KN^{2\kappa}\right)+N^{3\kappa-\lambda_0}+N^{3\kappa-\lambda}+C\braket{\Psi,\mathcal{E}' \Psi},
\end{align*}
with $\mathcal{E}':=\frac{1}{4 N^{2\lambda}}\left(\sum_{k,\ell \neq 0}\widehat{V_{N,\kappa}}(k)a_k^\dagger a^\dagger_{\ell-k} X_0 a_\ell a_0+\mathrm{H.c.}\right)+\frac{1}{4N^{2\lambda}}\left(\sum_{k\neq 0}\widehat{V_{N,\kappa}}(k)a_k^\dagger a_{-k}^\dagger X_1 a_0^2+\mathrm{H.c.}\right)$ and $X_i:= N^{2\lambda}\left[f \!  \left(\frac{\mathcal{N}+2}{N^\lambda}\right)-f \!  \left(\frac{\mathcal{N}+i}{N^{\lambda}}\right)\right]^2+ N^{2\lambda}\left[g \!  \left(\frac{\mathcal{N}+2}{N^\lambda}\right)-g \!  \left(\frac{\mathcal{N}+i}{N^\lambda}\right)\right]^2$. Using the fact that we have $\mathcal{E}'\lesssim N^{-2\lambda} \left(H_{N,\kappa}+N^{1+\kappa}\right)$ concludes the proof.
\end{proof}

\begin{center}
\textsc{Acknowledgments}
\end{center}
Funding from the ERC Advanced Grant ERC-AdG CLaQS, grant agreement n. 834782, is gratefully acknowledged.

\end{document}